\documentclass[aps,amssymb,showkeys]{revtex4}

\usepackage{array}
\usepackage{multirow}
\usepackage{amsmath}
\usepackage{amsthm}
\usepackage{graphicx} 
\usepackage{amsfonts}
\usepackage{amssymb}
\usepackage{kbordermatrix}
\usepackage[pdftex, bookmarks, colorlinks=true, plainpages = false, citecolor = red, urlcolor = blue, filecolor = blue]{hyperref}
\usepackage{color}
\usepackage{slashed}
\usepackage{yhmath}
\usepackage{mathrsfs}
\usepackage[shortlabels]{enumitem}
\setenumerate{listparindent=\parindent}
\usepackage{upgreek}
\usepackage{stmaryrd}
\usepackage{thm-restate}

\newcommand{\red}{\textcolor{red}}
\newcommand{\be}{\begin{equation}}
\newcommand{\ee}{\end{equation}}
\newcommand{\bea}{\begin{eqnarray}}
\newcommand{\eea}{\end{eqnarray}}
\newcommand{\bean}{\begin{eqnarray*}}
\newcommand{\eean}{\end{eqnarray*}}

\theoremstyle{plain}
\newtheorem{theorem}{Theorem}

\newtheorem{lem}[theorem]{Lemma}

\theoremstyle{definition}
\newtheorem{defn}[theorem]{Definition}

\allowdisplaybreaks

\begin{document}
\title{A solution space for a system of null-state partial differential equations II}

\date{\today}

\author{Steven M. Flores}
\email{steven.flores@helsinki.fi} 
\affiliation{Department of Mathematics \& Statistics, University of New Hampshire, Durham, New Hampshire, 03824,\\
and \\
Department of Mathematics \& Statistics, University of Helsinki, P.O. Box 68, 00014, Finland}

\author{Peter Kleban}
\email{kleban@maine.edu} 
\affiliation{LASST and Department of Physics \& Astronomy, University of Maine, Orono, Maine, 04469-5708, USA}

\begin{abstract}  

This article is the second of four that completely and rigorously characterize a solution space $\mathcal{S}_N$ for a homogeneous system of $2N+3$ linear partial differential equations (PDEs) in $2N$ variables that arises in conformal field theory (CFT) and multiple Schramm-L\"owner evolution (SLE$_\kappa$). The system comprises $2N$ null-state equations and three conformal Ward identities which govern CFT correlation functions of $2N$ one-leg boundary operators.  In the first article \cite{florkleb}, we use methods of analysis and linear algebra to prove that $\dim\mathcal{S}_N\leq C_N$, with $C_N$ the $N$th Catalan number.  The analysis of that article is complete except for the proof of a lemma that it invokes.  The purpose of this article is to provide that proof.

The lemma states that if every interval among $(x_2,x_3),$ $(x_3,x_4),\ldots,(x_{2N-1},x_{2N})$ is a two-leg interval of $F\in\mathcal{S}_N$ (defined in \cite{florkleb}), then $F$ vanishes.  Proving this lemma by contradiction, we show that the existence of such a nonzero function implies the existence of a non-vanishing CFT two-point function involving primary operators with different conformal weights, an impossibility.  This proof (which is rigorous in spite of our occasional reference to CFT) involves two different types of estimates, those that give the asymptotic behavior of $F$ as the length of one interval vanishes, and those that give this behavior as the lengths of two intervals vanish simultaneously.  We derive these estimates by using Green functions to rewrite certain null-state PDEs as integral equations, combining other null-state PDEs to obtain Schauder interior estimates, and then repeatedly integrating the integral equations with these estimates until we obtain optimal bounds.  Estimates in which two interval lengths vanish simultaneously divide into two cases: two adjacent intervals and two non-adjacent intervals.  The analysis of the latter case is similar to that for one vanishing interval length.  In contrast, the analysis of the former case is more complicated, involving a Green function that contains the Jacobi heat kernel as its essential ingredient.  \end{abstract}

\maketitle

\section{Introduction}\label{intro}

This article follows the analysis begun in \cite{florkleb} and continued in \cite{florkleb3,florkleb4}.  In this introduction, we state the problem under consideration and summarize the results from \cite{florkleb}.  The introduction \red{I} and appendix \red{A} of \cite{florkleb} explain the origin of this problem in conformal field theory (CFT) \cite{bpz,fms,henkel}, its relation to multiple Schramm-L\"owner evolution (SLE$_\kappa$) \cite{bbk, dub2, graham, kl, sakai}, and its application \cite{dots,gruz,rgbw,bbk,bauber,bpz,c3,c1} to critical lattice models \cite{bax, grim, wu, fk, stan} and random walks \cite{law1, schrsheff, weintru, zcs, madraslade}.

The goal of this article and \cite{florkleb,florkleb3,florkleb4} is to completely and rigorously determine a certain solution space $\mathcal{S}_N$ of the system of $2N$ null-state partial differential equations (PDEs) from CFT,
\be\label{nullstate}\Bigg[\frac{\kappa}{4}\partial_j^2+\sum_{k\neq j}^{2N}\left(\frac{\partial_k}{x_k-x_j}-\frac{(6-\kappa)/2\kappa}{(x_k-x_j)^2}\right)\Bigg]F(\boldsymbol{x})=0,\quad j\in\{1,2,\ldots,2N\},\ee
and three conformal Ward identities from CFT,
\be\label{wardid}\sum_{k=1}^{2N}\partial_kF(\boldsymbol{x})=0,\quad \sum_{k=1}^{2N}\left[x_k\partial_k+\frac{(6-\kappa)}{2\kappa}\right]F(\boldsymbol{x})=0,\quad \sum_{k=1}^{2N}\left[x_k^2\partial_k+\frac{(6-\kappa)x_k}{\kappa}\right]F(\boldsymbol{x})=0,\ee
with $\boldsymbol{x}:=(x_1,x_2,\ldots,x_{2N})$ and $\kappa\in(0,8)$.  The solution space $\mathcal{S}_N$ of interest comprises all (classical) solutions $F:\Omega_0\rightarrow\mathbb{R}$, where
\be\label{components}\Omega_0^M:=\{\boldsymbol{x}\in\mathbb{R}^M\,|\,x_1<x_2<\ldots< x_{M-1}< x_M\},\quad\Omega_0:=\Omega_0^{2N},\ee
such that for each $F\in\mathcal{S}_N$, there exist some positive constants $C$ and $p$ (which we may choose to be as large as needed) such that with $M=2N$,
\be\label{powerlaw} |F(\boldsymbol{x})|\leq C\prod_{i<j}^M|x_j-x_i|^{\mu_{ij}(p)},\quad\text{with}\quad\mu_{ij}(p):=\begin{cases}-p, & |x_j-x_i|<1 \\ +p, & |x_j-x_i|\geq1\end{cases}\quad\text{for all $\boldsymbol{x}\in\Omega_0.$}\ee
(We use this bound to prove lemma \red{3} in \cite{florkleb}, and we use it to prove lemmas \ref{firstlimitlem} and \ref{farlem}--\ref{closelem} of this article below.)  Restricting our attention to $\kappa\in(0,8)$, our goals are as follows:
\begin{enumerate}
\item\label{item1} Rigorously prove that every element of $\mathcal{S}_N$ is a real-valued Coulomb gas solution.  (See definition \red{1} of \cite{florkleb3} and \cite{df1, df2, dub}.)
\item\label{item2} Rigorously prove that $\dim\mathcal{S}_N=C_N$, with $C_N$ the $N$th Catalan number.
\item\label{item3} Argue that $\mathcal{S}_N$ has a basis consisting of $C_N$ \emph{connectivity weights} (physical quantities defined in the introduction \red{I} of \cite{florkleb}) and find formulas for all of the connectivity weights.
\end{enumerate}

In \cite{florkleb}, we used certain elements of the dual space $\mathcal{S}_N^*$ to prove that $\dim\mathcal{S}_N\leq C_N$, and in \cite{florkleb3}, we use these linear functionals again to complete goals \ref{item1}--\ref{item3}.  Hence, the work of both articles relies on the veracity of lemma \red{14} of \cite{florkleb}, and the purpose of this article is to prove this lemma, renamed lemma \ref{alltwoleglem} in this article below.

We briefly recall the methodology used to prove a weaker statement of item \ref{item1}, namely that $\dim\mathcal{S}_N\leq C_N$, in \cite{florkleb}.  To construct the elements of $\mathcal{S}_N^*$, we prove in \cite{florkleb} that for all $F\in\mathcal{S}_N$ and all $i\in\{2,3,\ldots,2N\}$, the limit
\be\label{lim}\bar{\ell}_1F(x_1,x_2,\ldots,x_{i-1},x_{i+1},\ldots,x_{2N})\,\,\,:=\lim_{x_i\rightarrow x_{i-1}}(x_i-x_{i-1})^{2\theta_1}F(\boldsymbol{x}),\quad\theta_1:=\frac{6-\kappa}{2\kappa}\ee
exists, is independent of $x_{i-1}$, and (after implicitly taking the trivial limit $x_{i-1}\rightarrow x_{i-2}$) is an element of $\mathcal{S}_{N-1}$.  (Another type of limit $\underline{\ell}_1$ fixes $x_{2N}=-x_1=R$ and sends $R\rightarrow\infty$ with the same consequences, and we denote either as $\ell_1$.)  (In \cite{florkleb, florkleb3, florkleb4}, we write $x_{i+1}\rightarrow x_i$, but only in this article, we reduce the index by one for convenience and write $x_{i-1}\rightarrow x_i$ instead.)  Following $\ell_1$, we apply $N-1$ more such limits $\ell_2$, $\ell_3,\ldots,\ell_N$ sequentially to (\ref{lim}), sending $F$ to an element of $\mathcal{S}_0:=\mathbb{R}$.

There are many ways that we may order a sequence of these limits, and in \cite{florkleb}, we list the conditions necessary to avoid various inconsistencies such as having the limit $\bar{\ell}_j$ that sends $x_{i_{2j}}\rightarrow x_{i_{2j-1}}$ precede the limit $\bar{\ell}_k$ that sends $x_{i_{2k}}\rightarrow x_{i_{2k-1}}$ if $x_{i_{2j-1}}<x_{i_{2k-1}}<x_{i_{2k}}<x_{i_{2j}}$.  We call the linear functional $\mathscr{L}:\mathcal{S}_N\rightarrow\mathbb{R}$ with $\mathscr{L}:=\ell_{j_N}\ell_{j_{N-1}}\dotsm\ell_{j_2}\ell_{j_1}$ and with the limits ordered to fulfill these conditions an \emph{allowable sequence of limits}.  Because it is linear, an allowable sequence of limits is an element of the dual space $\mathcal{S}_N^*$.  

In \cite{florkleb}, we further prove that two allowable sequences $\mathscr{L}$ and $\mathscr{L}'$, which bring together the same pairs of coordinates in different orders, have $\mathscr{L}'F=\mathscr{L}F$ for all $F\in\mathcal{S}_N$.  This fact establishes an equivalence relation among the allowable sequences of limits that partitioned them into $C_N$ equivalence classes $[\mathscr{L}_1]$, $[\mathscr{L}_2],\ldots,[\mathscr{L}_{C_N}]$, again with $C_N$ the $N$th Catalan number.

We conclude our analysis in \cite{florkleb} by proving that the linear mapping $v:\mathcal{S}_N\rightarrow\mathbb{R}^{C_N}$ with $v(F)_\varsigma:=[\mathscr{L}_\varsigma]F$ for each $\varsigma\in\{1,2,\ldots,C_N\}$ is injective and therefore $\dim\mathcal{S}_N \leq C_N$, assuming lemma \ref{alltwoleglem} below (denoted lemma \red{14} in \cite{florkleb}).  Lemma \ref{alltwoleglem} states that if the limit (\ref{lim}) equals zero for (i.e., $(x_{i-1},x_i)$ is a \emph{two-leg interval} of) some $F\in\mathcal{S}_N$ and every $i\in\{3,4,\ldots,2N\}$, then $F$ is zero.  In \cite{florkleb3}, we use the mapping $v$ once more to achieve goals \ref{item1} and \ref{item2} stated above.  Hence, what remains to achieve goals \ref{item1} and \ref{item2} is to prove lemma \ref{alltwoleglem}, so that is the purpose of this article.

\subsection{Methodology}\label{Methodology}

In this section, we outline our method for proving lemma \ref{alltwoleglem}, stated below.  The application of CFT to the study of statistical lattice models and loop models at the critical point motivates our method of proof, and we give a very brief survey of this application in appendix \red{A} of \cite{florkleb}.

The intuition behind our proof of the conjecture goes as follows.  We let $\psi_s(x)$ with $s\in\mathbb{Z}^+$ (resp.\ $\psi_0(x)=\mathbf{1}$) denote the $s$-leg boundary operator at $x\in\mathbb{R}$ (resp.\ identity operator) (see appendix \red{A} of \cite{florkleb}), and we denote its conformal weight by 
\be\label{bdysleg}
\theta_s=\frac{s(2s+4-\kappa)}{2\kappa},\quad s\in\mathbb{Z}^+\cup\{0\}.
\ee
Now, supposing that $F\in\mathcal{S}_N\setminus\{0\}$, we wish to study its behavior as $x_{2N}\rightarrow x_{2N-1}$, then $x_{2N-1}\rightarrow x_{2N-2}$, etc., down to $x_3\rightarrow x_2$.  To this end, we identify $F$, a solution of (\ref{nullstate}, \ref{wardid}), with the $2N$-point CFT correlation function
\be\label{firstcorrfunc}F(\boldsymbol{x})=\langle\psi_1(x_1)\psi_1(x_2)\dotsm\psi_1(x_{2N})\rangle,\ee
and we consider the operator product expansion (OPE) of the adjacent one-leg boundary operators in (\ref{firstcorrfunc}).  We recall the OPE of $\psi_1$ with $\psi_s$ as (with $C_{1s}^{s\pm1}$ and $\beta_{1s}^{s\pm1}$ the OPE constants associated with the $\psi_1\times\psi_s=\psi_{s\pm1}$ fusion channel)
\bea\psi_1(x)\psi_s(y)&\underset{y\rightarrow x}{\sim}&C_{1s}^{s-1}(y-x)^{-\theta_1-\theta_s+\theta_{s-1}}\psi_{s-1}(x)+C_{1s}^{s-1}\beta_{1s}^{s-1}(y-x)^{-\theta_1-\theta_s+\theta_{s-1}+1}\partial\psi_{s-1}(x)+\dotsm\nonumber\\
&+&C_{1s}^{s+1}(y-x)^{-\theta_1-\theta_s+\theta_{s+1}}\psi_{s+1}(x)+C_{1s}^{s+1}\beta_{1s}^{s+1}(y-x)^{-\theta_1-\theta_s+\theta_{s+1}+1}\partial\psi_{s+1}(x)+\dotsm.\eea
Now, if all of $(x_2,x_3),$ $(x_3,x_4),\ldots,(x_{2N-2},x_{2N-1})$, and $(x_{2N-1},x_{2N})$ are two-leg intervals of $F$, then we expect that only the $\psi_{2N-1}$ conformal family ultimately remains in the OPE (figure \ref{Combine}).  Hence, only the two-point function
\be\label{2pt}F_{2N-1}(x_1,x_2):=C_{11}^2C_{12}^3\dotsm C_{1,2N-2}^{2N-1}\langle\psi_1(x_1)\psi_{2N-1}(x_2)\rangle\ee
remains at leading order, and it cannot be zero by construction.  On the other hand, (\ref{2pt}) must vanish because the conformal weights of the primary operators within it are different.  From this contradiction, we conclude that $F=0$.

Now we describe our method in more formal terms.  Starting with (\ref{firstcorrfunc}), the OPE of $\psi_1(x_{2N})$ with $\psi_1(x_{2N-1})$ may only contain either the conformal family for $\psi_0$ (the identity) or $\psi_2$ or both conformal families, and the limit (\ref{lim})
\be\label{10lim}\lim_{x_{2N}\rightarrow x_{2N-1}}(x_{2N}-x_{2N-1})^{\theta_1+\theta_1-\theta_0}F(\boldsymbol{x}),\quad\boldsymbol{x}\in\Omega_0^{2N}\ee
partially determines this OPE content.  Indeed, if (\ref{10lim}) does not (resp.\ does) vanish, then the mentioned OPE does (resp.\ does not) contain the identity family.  In the case that (\ref{10lim}) vanishes and $F$ is not zero, then CFT anticipates  
\begin{multline}\label{112ope}\overbrace{\langle\psi_1(x_1)\psi_1(x_2)\dotsm\psi_1(x_{2N-2})\psi_1(x_{2N-1})\psi_1(x_{2N})\rangle}^{F}\\
\underset{x_{2N}\rightarrow x_{2N-1}}{\sim}(x_{2N}-x_{2N-1})^{-\theta_1-\theta_1+\theta_2}\underbrace{C_{11}^2\langle\psi_1(x_1)\psi_1(x_2)\dotsm\psi_1(x_{2N-2})\psi_2(x_{2N-1})\rangle}_{F_2}.\end{multline}
We isolate  $F_2$, shown above, by taking the limit (here, $\pi_i$ is a projection map that removes the $i$th coordinate from $\boldsymbol{x}$, see definition \red{2} of \cite{florkleb})
\be\label{12lim}(F_2\circ\pi_{2N})(\boldsymbol{x})\,\,\,:=\lim_{x_{2N}\rightarrow x_{2N-1}}(x_{2N}-x_{2N-1})^{\theta_1+\theta_1-\theta_2}F(\boldsymbol{x}),\quad\boldsymbol{x}\in\Omega_0^{2N},\ee 
and because we have identified it with the CFT correlation function on the right side of (\ref{112ope}), we anticipate that $F_2$  satisfies the associated null-state PDEs with $j\in\{1,2,\ldots,2N-2\}$,
\be\label{firstnullstatemod}\Bigg[\frac{\kappa}{4}\partial_j^2+\sum_{k\neq j}^{2N-2}\left(\frac{\partial_k}{x_k-x_j}-\frac{\theta_1}{(x_k-x_j)^2}\right)+\frac{\partial_{2N-1}}{x_{2N-1}-x_j}-\frac{\theta_2}{(x_{2N-1}-x_j)^2}\Bigg](F_2\circ\pi_{2N})(\boldsymbol{x})=0,\ee
and conformal Ward identities,
\be\label{firstwardidmod}\begin{gathered}\sum_{k=1}^{2N-1}\partial_k(F_2\circ\pi_{2N})(\boldsymbol{x})=0,\quad
\Bigg[\sum_{k=1}^{2N-2}(x_k\partial_k+\theta_1)+x_{2N-1}\partial_{2N-1}+\theta_2\Bigg](F_2\circ\pi_{2N})(\boldsymbol{x})=0,\\
\Bigg[\sum_{k=1}^{2N-2}(x_k^2\partial_k+2\theta_1x_k)+x_{2N-1}^2\partial_{2N-1}+2\theta_2x_{2N-1}\Bigg](F_2\circ\pi_{2N})(\boldsymbol{x})=0.\end{gathered}\ee  
We derive these anticipated properties of $F\in\mathcal{S}_N\setminus\{0\}$ in the proof of lemmas \ref{secondlimitlem} and \ref{pdelem} below by showing that if the limit (\ref{10lim}) vanishes, then the limit $F_2$ (\ref{12lim}) exists, is not zero, and satisfies (\ref{firstnullstatemod}, \ref{firstwardidmod}).  

Now we suppose that for some $F\in\mathcal{S}_N\setminus\{0\}$, we may repeat this analysis to sequentially generate a collection of functions $\{F_1:=F, F_2, F_3,\ldots\}$.  To construct $F_{s+1}:\Omega_0^{2N-s}\rightarrow\mathbb{R}$ from $F_s:\Omega_0^{2N-s+1}\rightarrow\mathbb{R}$, we first show that the limit
\be\label{ss-1lim}\lim_{\substack{x_{2N-s+1}\\\quad\rightarrow x_{2N-s}}}(x_{2N-s+1}-x_{2N-s})^{\theta_1+\theta_s-\theta_{s-1}}F_s(\boldsymbol{x}),\quad\boldsymbol{x}\in\Omega_0^{2N-s+1}\ee
exists.  Next, to anticipate the asymptotic behavior of $F_s$ as $x_{2N-s+1}\rightarrow x_{2N-s}$, we identify $F_s$ with the $2N$-point CFT correlation function
\be\label{corrfunc}F_s(\boldsymbol{x})=C_{11}^2C_{12}^3\dotsm C_{1,s-1}^s\langle\psi_1(x_1)\psi_1(x_2)\dotsm\psi_1(x_{2N-s})\psi_s(x_{2N-s+1})\rangle,\quad\boldsymbol{x}\in\Omega_0^{2N-s+1}\ee
and consider the OPE of $\psi_s(x_{2N-s+1})$ with $\psi_1(x_{2N-s})$.  This OPE may only contain either the conformal family for $\psi_{s-1}$ or $\psi_{s+1}$ or both conformal families.  If (\ref{ss-1lim}) does vanish, then CFT anticipates the behavior
\begin{multline}\label{1ss+1ope}\overbrace{C_{11}^2C_{12}^3\dotsm C_{1,s-1}^s\langle\psi_1(x_1)\psi_1(x_2)\dotsm\psi_1(x_{2N-s-1})\psi_1(x_{2N-s})\psi_s(x_{2N-s+1})\rangle}^{F_s}\\
\underset{\substack{x_{2N-s+1}\\\qquad\rightarrow x_{2N-s}}}{\sim}(x_{2N-s+1}-x_{2N-s})^{-\theta_1-\theta_s+\theta_{s+1}}\underbrace{C_{11}^2C_{12}^3\dotsm C_{1,s}^{s+1}\langle\psi_1(x_1)\psi_1(x_2)\dotsm\psi_1(x_{2N-s-1})\psi_{s+1}(x_{2N-s})\rangle}_{F_{s+1}}.\end{multline}
That is, the OPE for $\psi_s(x_{2N-s+1})$ with $\psi_1(x_{2N-s})$ does not contain the $\psi_{s-1}$ family if (\ref{ss-1lim}) vanishes.  Equation (\ref{1ss+1ope}) relates $F_s$ to the function $F_{s+1}$, which we wish to obtain from it, and we isolate the latter by taking the limit
\be\label{ss+1lim}(F_{s+1}\circ\pi_{2N-s+1})(\boldsymbol{x})\,\,\,:=\lim_{\substack{x_{2N-s+1}\\\qquad\rightarrow x_{2N-s}}}(x_{2N-s+1}-x_{2N-s})^{\theta_1+\theta_s-\theta_{s+1}}F_s(\boldsymbol{x}),\quad\boldsymbol{x}\in\Omega_0^{2N-s+1}.\ee 
Because we have identified it with the CFT correlation function on the right side of (\ref{1ss+1ope}), we anticipate that $F_{s+1}$  satisfies the null-state PDEs with $j\in\{1,2,\ldots,2N-s\}$,
\be\label{nullstatemods}\Bigg[\frac{\kappa}{4}\partial_j^2+\sum_{k\neq j}^{2N-s-1}\left(\frac{\partial_k}{x_k-x_j}-\frac{\theta_1}{(x_k-x_j)^2}\right)+\frac{\partial_{2N-s}}{x_{2N-s}-x_j}-\frac{\theta_{s+1}}{(x_{2N-s}-x_j)^2}\Bigg](F_{s+1}\circ\pi_{2N-s+1})(\boldsymbol{x})=0,\ee
and conformal Ward identities,
\be\label{wardidmods}\begin{gathered}\sum_{k=1}^{2N-s}\partial_k(F_{s+1}\circ\pi_{2N-s+1})(\boldsymbol{x})=0,\quad
\Bigg[\sum_{k=1}^{2N-s-1}(x_k\partial_k+\theta_1)+x_{2N-s}\partial_{2N-s}+\theta_{s+1}\Bigg](F_{s+1}\circ\pi_{2N-s+1})(\boldsymbol{x})=0,\\
\Bigg[\sum_{k=1}^{2N-s-1}(x_k^2\partial_k+2\theta_1x_k)+x_{2N-s}^2\partial_{2N-s}+2\theta_{s+1}x_{2N-s}\Bigg](F_{s+1}\circ\pi_{2N-s+1})(\boldsymbol{x})=0,\end{gathered}\ee
governing that correlation function.  We derive these anticipated properties of $F\in\mathcal{S}_N$ in the proof of lemma \ref{secondlimitlem} by showing that if the limit (\ref{ss-1lim}) vanishes, then the limit $F_{s+1}$ (\ref{ss+1lim}) exists, is not zero, and satisfies (\ref{nullstatemods}, \ref{wardidmods}). 

Now we seek a sufficient condition for $F\in\mathcal{S}_N\setminus\{0\}$ that guarantees this construction of $F_s$ for each $s\in\{2,3,\ldots,2N-1\}$.  For every such $s$, the condition that $(x_{2N-j+1},x_{2N-j})$ is a two-leg interval of $F$ for each $j\in\{1,2,\ldots,s-1\}$, i.e.,
\be\label{jlim}\lim_{\substack{x_{2N-j+1}\\\qquad\rightarrow x_{2N-j}}}(x_{2N-j+1}-x_{2N-j})^{\theta_1+\theta_1-\theta_0}F(\boldsymbol{x})=0,\quad\boldsymbol{x}\in\Omega_0^{2N},\quad \text{for all $j\in\{1,2,\ldots,s-1\}$}\ee
seems sufficient for this reason.  If, for example, $(x_{2N-2},x_{2N-1})$ and $(x_{2N-1},x_{2N})$ are two-leg intervals of $F$, then the OPE for $\psi_1(x_{2N})\times\psi_1(x_{2N-1})\times\psi_1(x_{2N-2})$ should contain only the $\psi_3$ conformal family.  If true, then we expect
\be\label{2claim}\begin{cases} (x_{2N\hphantom{-0}}-x_{2N-1})^{\theta_1+\theta_1-\theta_0}F(\boldsymbol{x})& \xrightarrow[\substack{x_{2N\hphantom{-0}}\\\qquad\rightarrow x_{2N-1}}]{}0 \\ (x_{2N-1}-x_{2N-2})^{\theta_1+\theta_1-\theta_0}F(\boldsymbol{x})& \xrightarrow[\substack{x_{2N-1}\\\qquad\rightarrow x_{2N-2}}]{}0\end{cases}\quad\Longrightarrow\quad(x_{2N-1}-x_{2N-2})^{\theta_1+\theta_2-\theta_1}F_2(\boldsymbol{x}) \xrightarrow[\substack{x_{2N-1}\\\qquad\rightarrow x_{2N-2}}]{}0, \ee
and we may construct $F_3$.  Similarly, if each $(x_{2N-j+1},x_{2N-j})$ with $j\in\{1,2,\ldots,s-1\}$ is a two-leg interval of $F$, then a multiple-SLE$_\kappa$ argument suggests that the OPE for the fusion induced by the limits $x_{2N}\rightarrow x_{2N-1}\rightarrow\ldots\rightarrow x_{2N-s+1}$,
\be\psi_1(x_{2N})\times\psi_1(x_{2N-2})\times\dotsm\times\psi_1(x_{2N-s+1}),\ee
should contain only the conformal family for $\psi_s$.  Indeed, no two multiple-SLE$_\kappa$ curves anchored to endpoints of these intervals interconnect (figure \ref{Joined}), so pulling these endpoints together anchors all $s$ curves to $x_{2N-s+1}$.  If true, then
\be\label{construction}\left\{\begin{array}{lll} (x_{2N\hphantom{-0}}-x_{2N-1})^{\theta_1+\theta_1-\theta_0}F(\boldsymbol{x})&\xrightarrow[\substack{x_{2N\hphantom{-0}}\\\qquad\rightarrow x_{2N-1}}]{}&0 \\ (x_{2N-1}-x_{2N-2})^{\theta_1+\theta_1-\theta_0}F(\boldsymbol{x})&\xrightarrow[\substack{x_{2N-1}\\\qquad\rightarrow x_{2N-2}}]{}&0 \\ 
\qquad\qquad\vdots & \\ \\
 (x_{2N-s+2}-x_{2N-s+1})^{\theta_1+\theta_1-\theta_0}F(\boldsymbol{x})&\xrightarrow[\substack{x_{2N-s+2}\\\qquad\rightarrow x_{2N-s+1}}]{}&0\end{array}\right.\quad\Longrightarrow\quad\begin{aligned}(x_{2N-s+2}-&x_{2N-s+1})^{\theta_1+\theta_{s-1}-\theta_{s-2}} \\ 
& \times F_{s-1}(\boldsymbol{x}) \xrightarrow[\substack{x_{2N-s+2}\\\qquad\rightarrow x_{2N-s+1}}]{}0,\end{aligned}\ee
and we may construct $F_s$.  Indeed, the first $k$ conditions from the top downward on the left side of (\ref{construction}) guarantee the construction of $F_{k+1}$, and we let $k$ go from one to $s-1$.  We explain this construction further in section \ref{twointervals} below.

\begin{figure}[t]
\centering
\includegraphics[scale=0.27]{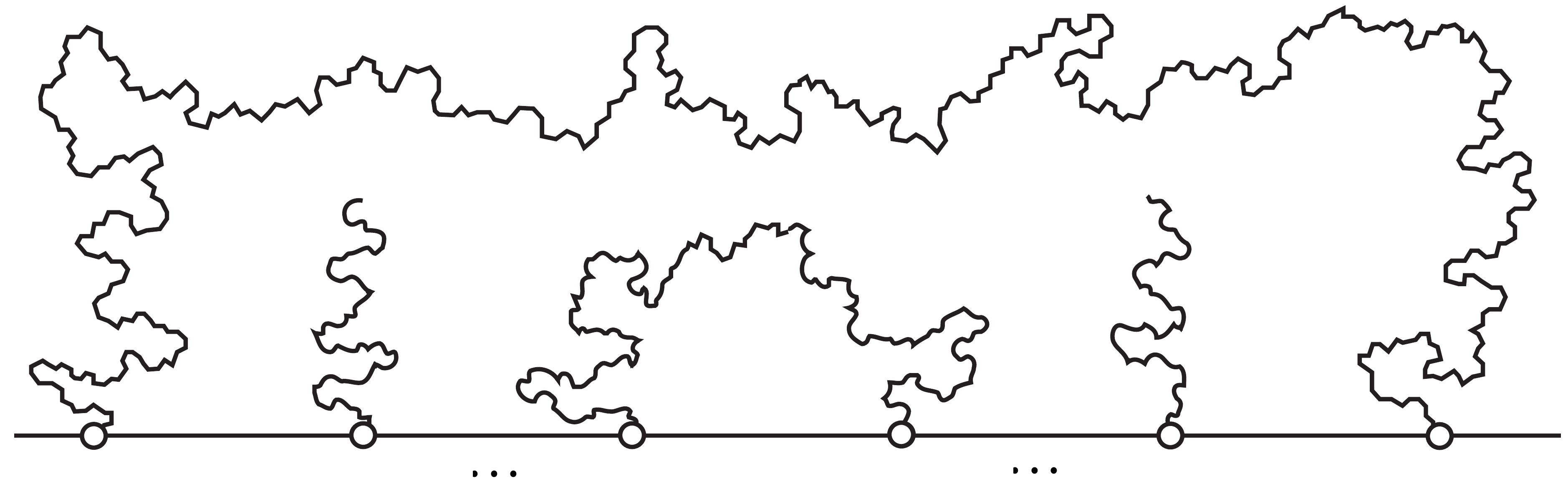}
\caption{If all of the above intervals are two-leg intervals and a multiple-SLE$_\kappa$ curve joins two of their collective endpoints, then at least one other multiple-SLE$_\kappa$ curve joins two endpoints of a single two-leg interval, an impossibility.  See figure \ref{Combine} next.}
\label{Joined}
\end{figure}

\begin{figure}[b]
\centering
\includegraphics[scale=0.27]{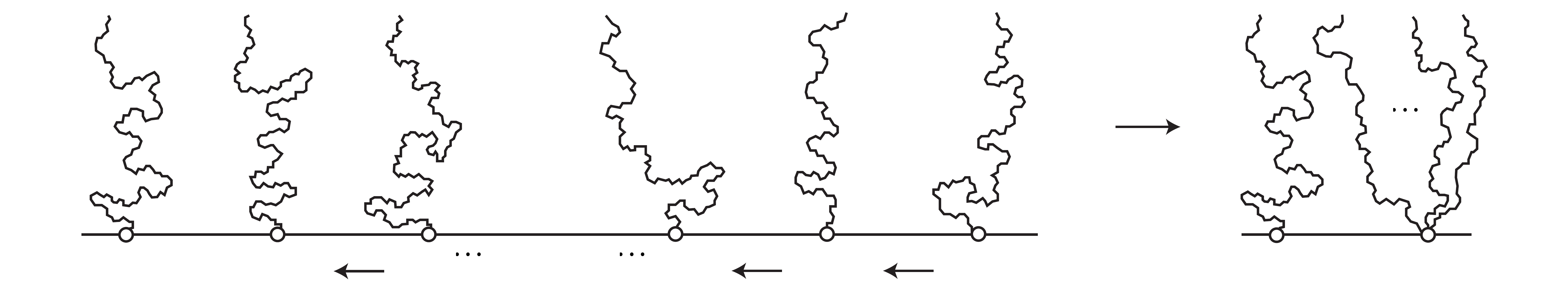}
\caption{If all of the above intervals are two-leg intervals, then no multiple-SLE$_\kappa$ curves may join any of their collective endpoints.  Thus, we may pull all of the rightmost curves together at their base, falsely implying that $\langle\psi_1(x_1)\psi_{2N-1}(x_2)\rangle$ is not zero.}
\label{Combine}
\end{figure}

Finally, if all of $(x_2,x_3),$ $(x_3,x_4),\ldots,(x_{2N-2},x_{2N-1})$, and $(x_{2N-1},x_{2N})$ are two-leg intervals of $F\in\mathcal{S}_N\setminus\{0\}$ and condition (\ref{construction}) is true for all $s\in\{1,2,\ldots,2N-1\}$, then we may construct the two-point function $F_{2N-1}$ of (\ref{2pt}) from $F$ (figure \ref{Combine}).  By construction, $F_{2N-1}$ is not zero but also satisfies the system (\ref{wardidmods}) with $s=2N-2$, whose only solution is zero.  This contradiction implies our main result:
\begin{restatable}{lem}{alltwoleglem}\label{alltwoleglem}  Suppose that $\kappa\in(0,8)$ and $F\in\mathcal{S}_N$ with $N>1$.  If all of $(x_2,x_3),$ $(x_3,x_4),\ldots,(x_{2N-2},x_{2N-1})$, and $(x_{2N-1},x_{2N})$  are two-leg intervals of $F$, then $F=0$.
\end{restatable}

Our proof of this lemma follows the strategy presented above, and much of it mirrors the proofs presented in \cite{florkleb} with minor adjustments.  However, its most important ingredient, justifying (\ref{construction}), is not analogous with anything presented in \cite{florkleb} because it requires us to study the behavior of the functions $F_1$, $F_2,\ldots, F_{2N-1}$ as we collapse two \emph{adjacent} intervals \emph{simultaneously}.  To obtain this estimate, we construct a Green function analogous with that (\ref{Jgreen}) used in the case in which we collapse one interval, except that the new Green function depends on four variables rather than two.  Equation (\ref{Greenfunc}) gives this Green function, which happily leads to the precise estimates for $F_1$, $F_2,\ldots, F_{2N-1}$ needed in order to prove (\ref{construction}) for all $s\in\{1,2,\ldots,2N-1\}$.  Interestingly, the essential part of this Green function is the Jacobi heat kernel \cite{nowsj1, nowsj2}.

Because our analysis is rigorous, none of our proofs actually relies on the interpretation of various functions as CFT correlation functions.  Instead, they only assume that such various functions satisfy certain systems of PDEs, the same PDEs that the correlation functions with which we identify them satisfy.

\subsection{Organization} 
This article is organized as follows.  Section \ref{prelim} establishes some preliminary estimates that we need for the proof of lemma \ref{alltwoleglem} as outlined in section \ref{Methodology}.  Section \ref{oneinterval} gives estimates for the behavior of solutions of the system (\ref{nullstatemods}, \ref{wardidmods}) as we collapse just one interval among $(x_1,x_2)$, $(x_2,x_3),\ldots,(x_{2N-s-1},x_{2N-s})$, and $(x_{2N-s},x_{2N-s+1})$.  The system (\ref{nullstatemods}, \ref{wardidmods}) with $s=0$ is identical to the original system (\ref{nullstate}, \ref{wardid}) at the heart of this article and its relatives \cite{florkleb, florkleb3, florkleb4}.  We show that the power-law behavior as $x_i\rightarrow x_{i-1}$ of solutions $F$ of either system is
\be\label{normalpowers}F(\boldsymbol{x})\underset{x_i\rightarrow x_{i-1}}{\sim}(x_i-x_{i-1})^{-\theta_1-\theta_1+\theta_0}(F_0\circ\pi_i)(\boldsymbol{x})\,\,+\,\,\dotsm\,\,+(x_i-x_{i-1})^{-\theta_1-\theta_1+\theta_2}(F_2\circ\pi_i)(\boldsymbol{x})\,\,+\,\,\dotsm\ee
for all $i\in\{2,3,\ldots,2N\}$ and some functions $F_0$ and $F_2$, in agreement with CFT predictions.  (We derive the first term on the right side of (\ref{normalpowers}) for functions $F$ satisfying (\ref{nullstate}, \ref{wardid}) in lemma \red{4} of \cite{florkleb}.)

Now, if $s>1$, then the system (\ref{nullstatemods}, \ref{wardidmods}) with $s$ replaced by $s-1$ differs from the original system (\ref{nullstate}, \ref{wardid}) because, unlike the latter, the rightmost point $x_{2N-s+1}$ of the former has conformal weight $\theta_s$ rather than $\theta_1$.  We show that, while (\ref{normalpowers}) still gives the power-law behavior of this former system's solutions $F_s$ as $x_i\rightarrow x_{i-1}$ for $i<2N-s+1$, the behavior of $F_s$ as $x_{2N-s+1}\rightarrow x_{2N-s}$ is
\begin{multline}\label{abnormalpowers}F_s(\boldsymbol{x})\underset{\substack{x_{2N-s+1}\\ \rightarrow x_{2N-s}}}{\sim}(x_{2N-s+1}-x_{2N-s})^{-\theta_1-\theta_s+\theta_{s-1}}(F_{s-1}\circ\pi_{2N-s+1})(\boldsymbol{x})\,\,+\,\,\dotsm\\
\,\,+\,\,(x_{2N-s+1}-x_{2N-s})^{-\theta_1-\theta_s+\theta_{s+1}}(F_{s+1}\circ\pi_{2N-s+1})(\boldsymbol{x})\,\,+\,\,\dotsm\end{multline}
for some functions $F_{s-1}$ and $F_{s+1}$, in agreement with CFT predictions.  We derive the first and second terms on the right side of (\ref{normalpowers}, \ref{abnormalpowers}) in the proofs of lemmas \ref{firstlimitlem} and \ref{secondlimitlem} respectively.  Furthermore, we prove in lemmas \ref{secondlimitlem} and \ref{pdelem} respectively that if $F_{s-1}$, appearing in the right side of (\ref{abnormalpowers}), is zero, then $F_{s+1}$ is nonzero and satisfies (\ref{nullstatemods}, \ref{wardidmods}).  

To obtain these results, we use a Green function (\ref{Jgreen}) with two variables to express the null-state PDE centered on one of the interval's endpoints as an integral equation.  This equation contains derivatives of $F$ with respect to variables not involved with the Green function.  Then, we use the other null-state PDEs and conformal Ward identities to construct an elliptic PDE, that implies Schauder interior estimates that bound these derivatives by $F$ itself. Repeated integration of the integral equation successively improves bounds on the growth of $F$ as the interval length vanishes, until we reach an optimal bound.  This summarizes the content of section \ref{oneinterval}.

Section \ref{twointervals} concerns estimates for the behavior of solutions of (\ref{nullstatemods}, \ref{wardidmods}) as we collapse two intervals simultaneously.  Here, there are two cases to consider, depending on whether the intervals are not (resp.\ are) adjacent to one another.  Lemma \ref{farlem} (resp.\ lemma \ref{closelem}) gives an estimate for the former (resp.\ latter) case.  The difference in the analysis between these two cases is considerable.  In the former, we employ the analysis of  lemmas \ref{firstlimitlem} and \ref{secondlimitlem} twice, once for each collapsing interval.   We use the Green function (\ref{Jgreen}) in two variables both times, with these variables  related to the length of the corresponding interval.  However, this strategy cannot be applied to adjacent collapsing intervals.  Indeed, this case  leads to a PDE in two variables associated with the lengths of the two adjacent intervals.  The corresponding Green function (\ref{Greenfunc}) therefore depends on four variables.  Interestingly, this Green function factors into power functions multiplying the Jacobi heat kernel \cite{nowsj1, nowsj2}, which we then use to obtain the desired estimates for $F$ in this last case.

In section \ref{theproof}, we use the mentioned estimates of section \ref{prelim} to complete the proof of lemma \ref{alltwoleglem} as described in section \ref{Methodology} above.  Our strategy employs the two-interval estimates of section \ref{twointervals} to justify the claims (\ref{2claim}, \ref{construction}) that allow us to construct the collection of functions $\{F_1,F_2,\ldots,F_{2N-1}\}$ if we assume that the lemma is false.

\section{Preliminary estimates}\label{prelim}

In this section, we establish some preliminary estimates that we need for the proof of lemma \ref{alltwoleglem} in section \ref{theproof}.  Below, we frequently use the KPZ formula \cite{gruz, kpz}, which relates conformal weights $ d$ of CFT primary operators (expressed in terms of the SLE parameter $\kappa$) in a flat metric with conformal weights $\Delta^+( d)$ in the fluctuating metrics of two-dimensional quantum gravity.
\begin{defn}\label{kpzdefn} For each $\kappa\in(0,8)$, we call the function $\Delta^\pm:\mathbb{R}^+\rightarrow\mathbb{R}$, with the formula
\be\label{kpz}\Delta^\pm( d):=\frac{\kappa-4\pm\sqrt{(\kappa-4)^2+16\kappa d}}{2\kappa},\ee
the \emph{$\pm$KPZ formula}.  Furthermore, we define the function $\bindnasrepma:\mathbb{R}^+\rightarrow\mathbb{R}^+$ with formula
\be\bindnasrepma( d):=\Delta^+( d)-\Delta^-( d).\ee
\end{defn}
\noindent
The $\pm$KPZ formula is also useful in CFT, thanks to the following lemma.
\begin{lem}\label{kpzlem} Suppose that $\kappa\in(0,8)$, and let $\theta_s$ be given by (\ref{bdysleg}) with $s\in\mathbb{Z}^+$.  Then 
\be \label{kpzeq}
\Delta^\pm(\theta_s)=-\theta_1-\theta_s+\theta_{s\pm1}=\begin{cases}2s/\kappa,&+ \\1-(2s+4)/\kappa,&-\end{cases}.\ee
\end{lem}
\begin{proof} One may prove the lemma with straightforward algebra.\end{proof}
\noindent
Lemma \ref{kpzlem} implies that $\Delta^-(\theta_s)$ and $\Delta^+(\theta_s)$ are respectively the indical powers for the $\psi_{s-1}$ and $\psi_{s+1}$ conformal families that appear in the OPE of $\psi_1$ with $\psi_s$.  Also, they are respectively the powers of the first and second terms in (\ref{abnormalpowers}).  

With the $\pm$KPZ formula, we may write the generalization of (\ref{abnormalpowers}) predicted by CFT \cite{bpz, fms, henkel}.  If $F(\boldsymbol{x})$ is a CFT correlation function with primary operators at $x_{i-1}$ and $x_i$ with respective conformal weights $\theta_1$ and $ d$, then these two primary operators fuse in the limit $x_i\rightarrow x_{i-1}$, and their OPE contains two conformal families whose leading primary operators have respective conformal weights $ d^-=\theta_1+ d+\Delta^-( d)$ and $ d^+=\theta_1+ d+\Delta^+( d)$.  Thus, we have
\be\label{ababnormalpowers}F(\boldsymbol{x})\underset{x_i\rightarrow x_{i-1}}{\sim}(x_i-x_{i-1})^{\Delta^-( d)}
(F_{ d^-}\circ\pi_i)(\boldsymbol{x})\,\,+\,\,\dotsm\,\,+\,\,(x_i-x_{i-1})^{\Delta^+( d)}(F_{ d^+}\circ\pi_i)(\boldsymbol{x})\,\,+\,\,\dotsm\ee
for some functions $F_{ d^\pm}$.  We recover (\ref{abnormalpowers}) from (\ref{ababnormalpowers}) after setting $ d=\theta_s$ and $i=2N-s+1$ and replacing $F$ by $F_s$.

\subsection{Estimates involving one interval}\label{oneinterval}

In this section, we investigate the behavior of solutions $F$ of the system (\ref{nullstatemods}, \ref{wardidmods}), with $\theta_{s+1}$ replaced by an (almost) arbitrary conformal weight $h$ and $2N-s$ replaced by $M$, as we collapse just one interval among $(x_1,x_2)$, $(x_2,x_3),\ldots,(x_{M-2},x_{M-1})$, and $(x_{M-1},x_M)$.  Although the rightmost point $x_M$ bears the anomalous conformal weight $h$ in this system, we allow exactly one of $x_1$, $x_2,\ldots,x_M$ to bear this weight, denoting its index by $\iota\in\{1,2,\ldots,M\}$.

In lemma \ref{firstlimitlem}, we isolate the first term that appears on the right side of (\ref{ababnormalpowers}) as we collapse the interval $(x_{i-1},x_i)$.  The coefficient of this term is the limit (\ref{hlim-}) appearing in the statement of this lemma, and we 
 identify this limit with the limit (\ref{ss-1lim}) after we set $i=\iota=M$, $M=2N-s+1$, and $h=\theta_s$.  This point is relevant because we need the latter limit to exist in order to execute the methodology proposed in section \ref{Methodology}.  The statement and proof of lemma \ref{firstlimitlem} is almost identical to those of lemmas \red{3} and \red{4} of \cite{florkleb}, but the differences, though slight, are great enough to warrant a separate proof, which we provide below (with some steps replaced by references to identical steps in \cite{florkleb} for brevity).  

\begin{lem}\label{firstlimitlem}Suppose that $\kappa\in(0,8)$ and $M>2$, and for some $i\in\{2,3,\ldots,M\}$, let 
\be\label{xdelta}\boldsymbol{x}_\delta:=(x_1,x_2,\ldots,x_{i-1},x_{i-1}+\delta,x_{i+1},\ldots,x_M)\in\Omega_0^M.\ee
If $F:\Omega_0^M\rightarrow\mathbb{R}$ satisfies these conditions,
\begin{enumerate}
\item\label{cond1} $F$ satisfies the growth bound (\ref{powerlaw}) for some positive constants $C$ and $p$, and
\item\label{cond2} there is an $\iota\in\{1,2,\ldots,M\}$  and an $h>-(\kappa-4)^2/16\kappa$ such that $F$ satisfies the null-state PDE centered on $x_j$
\be\label{hpde}\Bigg[\frac{\kappa}{4}\partial_j^2\,\,\,+\sum_{k\neq j,\iota}^M\left(\frac{\partial_k}{x_k-x_j}-\frac{\theta_1}{(x_k-x_j)^2}\right)+\frac{\partial_\iota}{x_\iota-x_j}-\frac{h}{(x_\iota-x_j)^2}\Bigg]F(\boldsymbol{x})=0\ee
for each $j\in\{1,2,\ldots,\iota-1,\iota+1,\ldots, M\}$ and the conformal Ward identities 
\be\label{wardidh}\begin{gathered}\sum_{k=1}^M\partial_kF(\boldsymbol{x})=0,\qquad\Bigg[\sum_{k\neq \iota}^M(x_k\partial_k+\theta_1)+x_\iota\partial_\iota+h\Bigg]F(\boldsymbol{x})=0,\\
\Bigg[\sum_{k\neq \iota}^M(x_k^2\partial_k+2\theta_1x_k)+x_\iota^2\partial_\iota+2hx_\iota\Bigg]F(\boldsymbol{x})=0,\end{gathered}\ee
\end{enumerate}
then the limits
\be\label{hlim-}\begin{gathered}(F_{d^-}\circ\pi_i)(\boldsymbol{x}):=\lim_{\delta\downarrow0}\delta^{-\Delta^-(d)}F(\boldsymbol{x}_\delta),\qquad\lim_{\delta\downarrow0}\delta^{-\Delta^-(d)}\partial_\iota F(\boldsymbol{x}_\delta),\qquad d:=\begin{cases} \theta_1,& i\not\in\{\iota,\iota+1\} \\ h, & i\in\{\iota,\iota+1\}\end{cases}\\
\lim_{\delta\downarrow0}\delta^{-\Delta^-(d)}\partial_jF(\boldsymbol{x}_\delta),\qquad\lim_{\delta\downarrow0}\delta^{-\Delta^-(d)}\partial_j^2F(\boldsymbol{x}_\delta),\qquad j\not\in\begin{cases}\{i,\iota\},& i\not\in\{\iota,\iota+1\} \\ \{i-1,i\}, & i\in\{\iota,\iota+1\}\end{cases}\end{gathered}\ee
exist and are approached uniformly over every compact subset of $\pi_i(\Omega_0^M)$.
\end{lem}

\begin{proof} Throughout, we let $\mathcal{K}$ be an arbitrary compact subset of $\pi_i(\Omega_0^M)$.  Below, we divide the proof into proofs for three cases: case \ref{firstlimitlempart1} with $i\not\in\{\iota,\iota+1\}$, case \ref{firstlimitlempart2} with $i=\iota+1$, and case \ref{firstlimitlempart3} with $i=\iota$. The proofs for cases \ref{firstlimitlempart2} and \ref{firstlimitlempart3} are almost identical to that for case \ref{firstlimitlempart1}, so our exposition for these latter cases focuses on their differences with the former.
\begin{enumerate}[leftmargin=*]
\item\label{firstlimitlempart1} We suppose that $i\not\in\{\iota,\iota+1\}$.  For later convenience, we let $x:=x_{i-1}$ and $y:=x_\iota$, we relabel the coordinates in $\{x_j\}_{j\neq i-1,i,\iota}$ as $\{\xi_1,\xi_2,\ldots,\xi_{M-3}\}$ in increasing order, and we let $\boldsymbol{\xi}:=(\xi_1,\xi_2,\ldots,\xi_{M-3})$.  Also, we define
\be\label{tilde} {\rm F}(\boldsymbol{\xi};x,\delta;y):=F\left(\begin{aligned}\xi_1,\xi_2,\ldots,\xi_{i-2},x,x&+\delta,\xi_{i-1},\ldots\\
&\ldots,\xi_{\iota-3},y,\xi_{\iota-2},\ldots,\xi_{M-3}\end{aligned}\right),\quad H(\boldsymbol{\xi};x,\delta;y):=\delta^{-\Delta^-(\theta_1)} {\rm F}(\boldsymbol{\xi};x,\delta;y)\ee
if $x<y$, and we define ${\rm F}$ and $H$ similarly if $y<x$.

First, we bound the growth of the supremum of ${\rm F}(\boldsymbol{\xi};x,\delta;y)$ over $\mathcal{K}$ as $\delta\downarrow0$.  Following the proof of lemma \red{3} in \cite{florkleb}, we write the null-state PDE (\ref{nullstate}) centered on $x_i$ as $\mathcal{L}[ {\rm F}]=\mathcal{M}[ {\rm F}]$, where $\mathcal{L}$ is the Euler differential operator
\be\label{i+11}\mathcal{L}[{\rm F}](\boldsymbol{\xi};x,\delta;y):=\left[\frac{\kappa}{4}\partial_\delta^2+\frac{\partial_\delta}{\delta}-\frac{\theta_1}{\delta^2}\right]{\rm F}(\boldsymbol{\xi};x,\delta;y)\ee
with characteristic exponents $\Delta^-(\theta_1)=1-6/\kappa$ and $\Delta^+(\theta_1)=2/\kappa$.  $\mathcal{L}[{\rm F}](\boldsymbol{\xi};x,\delta;y)$ seems to comprise the largest terms of this PDE as $\delta\downarrow0$.  Also, $\mathcal{M}$ is the following differential operator with derivatives in all variables except $\delta$:
\be\label{i+12}\mathcal{M}[{\rm F}](\boldsymbol{\xi};x,\delta;y):=\Bigg[\frac{\partial_x}{\delta}+\frac{h}{(y-x-\delta)^2}-\frac{\partial_y}{y-x-\delta}+\sum_{j=1}^{M-3}\left(\frac{\theta_1}{(\xi_j-x-\delta)^2}-\frac{\partial_j}{\xi_j-x-\delta}\right)\Bigg]{\rm F}(\boldsymbol{\xi};x,\delta;y).\ee
Next, we use an appropriate causal Green function \cite{florkleb} to invert $\mathcal{L}$ and write $\mathcal{L}[ {\rm F}]=\mathcal{M}[ {\rm F}]$ as an integral equation.  The proof of lemma \red{3} in \cite{florkleb} presents the analysis for this process, and in terms of $H$, this integral equation is
\be\label{tildeH} H(\boldsymbol{\xi};x,\delta;y)=H(\boldsymbol{\xi};x,b;y)-\frac{\kappa}{4}J(\delta,b)\partial_b H(\boldsymbol{\xi};x,b;y)+\sideset{}{_\delta^b}\int J(\delta,\eta)\mathcal{M}[H](\boldsymbol{\xi};x,\eta;y)\,{\rm d}\eta\ee
(compare with (\red{58}) in \cite{florkleb}), where $b$ is small enough so the denominators of (\ref{i+12}) are bounded away from zero on $0<\delta<b$.  Here, $J$ is the Green function (recall that $\bindnasrepma:=\Delta^+-\Delta^-$, so $\bindnasrepma(\theta_1)=8/\kappa-1$)
\be\label{Jgreen} J(\delta,\eta):=\frac{4/\kappa}{\bindnasrepma(\theta_1)}\eta\,\Theta(\eta-\delta)\Bigg[1-\left(\frac{\delta}{\eta}\right)^{\bindnasrepma(\theta_1)}\Bigg],\ee
with $\Theta$ the Heaviside step function.  We choose bounded open sets $ U_0,$ $ U_1,\ldots, U_m$, with $m:=\lceil q\rceil$ and $q:=p+\Delta^-(\theta_1)$ (we choose $p$ large enough so $m>2$) and such that they are sequentially compactly embedded thus:
\be\label{compactembedd}\mathcal{K}\subset\subset U_m\subset\subset U_{m-1}\subset\subset\ldots\subset\subset U_0\subset\subset\pi_i(\Omega_0^M).\ee
After taking the supremum of (\ref{tildeH}) over each of these subsets, we find
\be\label{estH}\sup_{U_n}|H(\boldsymbol{\xi};x,\delta;y)|\leq\sup_{ U_n}|H(\boldsymbol{\xi};x,b;y)|\,+\,\bindnasrepma(\theta_1)^{-1}\sup_{ U_n}|b\hspace{.03cm}\partial_b H(\boldsymbol{\xi};x,b;y)|+\frac{4/\kappa}{\bindnasrepma(\theta_1)}\sideset{}{_\delta^b}\int\sup_{ U_n}|\eta\mathcal{M}[H](\boldsymbol{\xi};x,\eta;y)|\,{\rm d}\eta.\ee

Next, we construct a strictly elliptic PDE in order to bound the derivatives of $H$ in the integrand of (\ref{estH}).  The construction is nearly identical to that in the proof of lemma \red{3} in \cite{florkleb}, so for brevity we present only an outline.
\begin{enumerate}
\item\label{step11}We sum the null-state PDEs (\ref{hpde}) over $j\not\in\{i-1,i,\iota\}$ and cast the result in terms of $\boldsymbol{\xi}$, $x$, $y$, $\delta$, and $H$. 
\item\label{step12} Next, we use the conformal Ward identities (\ref{wardidh})  to replace the derivatives $\partial_xH$, $\partial_yH$, and $\delta\partial_\delta H$ in the PDE created by step \ref{step11} exclusively with $H$ and its derivatives with respect to the coordinates of $\boldsymbol{\xi}$.
\item\label{step13}Thus, $H$ satisfies a strictly elliptic PDE with all derivatives in the coordinates of $\boldsymbol{\xi}$, with $x$, $y,$ and $\delta$ as parameters, and with none of its coefficients vanishing or blowing up as $\delta\downarrow0$.
\end{enumerate}
Thus, the Schauder interior estimate (Cor.\ 6.3 of \cite{giltru}) says that with $d_n:=\text{dist}(\partial U_n,\partial U_{n-1})$ and $R_n$ $:=\text{diam}( U_n)/2$, the inequality
\be\label{Schauder}d_{n+1}\sup_{ U_{n+1}}|\partial^\varpi H(\boldsymbol{\xi};x,\delta;y)|\leq C(R_n)\sup_{ U_n}|H(\boldsymbol{\xi};x,\delta;y)|\quad\begin{array}{l}\\ \end{array}\ee
holds for all $0<\delta<b$ and $n\in\{0,1,\ldots,m\}$, where $C$ is some positive-valued function and $\varpi$ is any multi-index for the coordinates of $\boldsymbol{\xi}$ with length $|\varpi|\leq2$.  Actually, step \ref{step12} in the construction of the elliptic PDE implies that $\varpi$ may also involve the coordinates $x$ and $y$ and the derivative $\delta\partial_\delta$ too. That is,
\be\label{multiindex1}\partial^\varpi\in\left\{\begin{array}{c}\partial_j,\quad \partial_x,\quad \partial_y,\quad \delta\partial_\delta, \quad\partial_j^2,\quad\partial_x^2,\quad\partial_y^2,\quad(\delta\partial_\delta)^2,\\ \partial_j\partial_k, \quad \partial_j\partial_x,\quad\partial_j\partial_y,\quad\partial_j\delta\partial_\delta,\quad \partial_x\partial_y,\quad\partial_x\delta\partial_\delta,\quad\partial_y\delta\partial_\delta \end{array}\right\}.\ee

Condition \ref{cond1} of the lemma implies that both sides of (\ref{Schauder}) with $n=0$ are $O(\delta^{-q})$ as $\delta\downarrow0$.  By combining (\ref{Schauder}) with (\ref{estH}), we may improve these bounds on the growth of $H$ and $\partial^\varpi H$.  Specifically,  after substituting the bound (\ref{Schauder}) with $n=0$ into (\ref{i+12}, \ref{estH}) with $n=1$ and evaluating the definite integral in (\ref{estH}), we find 
\be\label{thefirstest}\sup_{U_1}|H(\boldsymbol{\xi};x,\delta;y)|=O(\delta^{-q+1})\quad\Longrightarrow\quad\sup_{U_2}|\partial^\varpi H(\boldsymbol{\xi};x,\delta;y)|=O(\delta^{-q+1}),\ee
with the right estimate of (\ref{thefirstest}) following from (\ref{Schauder}).  After repeating this process another $m-2$ times, we find that the left side of (\ref{estH}) with $n=m-1$ is $O(\delta^{-q+m-1})$ as $\delta\downarrow0$.  Repeating one last time finally gives
\be\label{finbound}\sup_{ U_m}|H(\boldsymbol{\xi};x,\delta;y)|=O(1)\quad\text{as $\delta\downarrow0$}\quad\Longrightarrow\quad\sup_{\mathcal{K}}|\partial^\varpi H(\boldsymbol{\xi};x,\delta;y)|=O(1)\quad\text{as $\delta\downarrow0$},\ee
thanks to the compact embedding $\mathcal{K}\subset\subset U_m$ (\ref{compactembedd}).  Further iterations fail to improve the bound (\ref{finbound}) because the other $O(1)$ terms in (\ref{estH}) may no longer be ignored.  (We follow the main lines of this argument in the proofs of lemmas \ref{secondlimitlem}, \ref{farlem}, \ref{firstcloselem}, and \ref{closelem} below.)

Now, to finish the proof, we must show that the limit of $H(\boldsymbol{\xi};x,\delta;y)$ and its derivatives in (\ref{hlim-}) as $\delta\downarrow0$ exist and are approached uniformly over $\mathcal{K}$.  But the proof of this claim is identical to the first part of the proof of lemma \red{4}, from (\red{58}) to (\red{66}), in \cite{florkleb}.  Thus, we have proven the lemma for case \ref{firstlimitlempart1}.

\item\label{firstlimitlempart2} We suppose that $i=\iota+1$, so with $x=y$ in the notation of case \ref{firstlimitlempart1}, we relabel the coordinates in $\{x_j\}_{j\neq i-1,i}$ in increasing order as $\{\xi_1,\xi_2,$ $\ldots,\xi_{M-2}\}$, and we let $\boldsymbol{\xi}:=(\xi_1,\xi_2,\ldots,\xi_{M-2})$.  Also, we define
\be\label{FtoH} {\rm F}(\boldsymbol{\xi};x,\delta):=F(\xi_1,\xi_2,\ldots,\xi_{i-2},x,x+\delta,\xi_{i-1},\ldots,\xi_M),\quad H(\boldsymbol{\xi};x,\delta):=\delta^{-\Delta^-(h)} {\rm F}(\boldsymbol{\xi};x,\delta).\ee
Second, we write the null-state PDE (\ref{nullstate}) centered on $x_i$ as $\mathcal{L}[{\rm F}]=\mathcal{P}[{\rm F}]$, where $\mathcal{L}$ is given by (\ref{i+11}) with $\theta_1$ replaced by $h$ (so the characteristic exponents of $\mathcal{L}$ are now $\Delta^-(h)$ and $\Delta^+(h)$), and where $\mathcal{P}$ is 
\be\label{newi+12}\mathcal{P}[{\rm F}](\boldsymbol{\xi};x,\delta):=\Bigg[\frac{\partial_x}{\delta}+\sum_{j=1}^{M-2}\left(\frac{\theta_1}{(\xi_j-x-\delta)^2}-\frac{\partial_j}{\xi_j-x-\delta}\right)\Bigg]{\rm F}(\boldsymbol{\xi};x,\delta).\ee
With these modifications, (\ref{tildeH}) becomes (again, compare with (\red{58}) in \cite{florkleb})
\be\label{tildeHold} H(\boldsymbol{\xi};x,\delta)=H(\boldsymbol{\xi};x,b)-\frac{\kappa}{4}J(\delta,b)\partial_b H(\boldsymbol{\xi};x,b)+\sideset{}{_\delta^b}\int J(\delta,\eta)\mathcal{P}[H](\boldsymbol{\xi};x,\eta)\,{\rm d}\eta\ee
for all $0<\delta<b$, with $\theta_1$ replaced by $h$ in the definition (\ref{Jgreen}) for $J(\delta,\eta)$.  Now, the construction \ref{step11}--\ref{step13} in part \ref{firstlimitlempart1} of this proof (with $x=y$) shows that $H$ satisfies a PDE with the properties of \ref{step13}.  Thus, the Schauder estimate
\be\label{otherSchauder}d_{n+1}\sup_{ U_{n+1}}|\partial^\varpi H(\boldsymbol{\xi};x,\delta)|\leq C(R_n)\sup_{ U_n}|H(\boldsymbol{\xi};x,\delta)|,\quad \partial^\varpi\in\left\{\begin{array}{c}\partial_j,\quad \partial_x,\quad \delta\partial_\delta, \quad\partial_j^2,\quad\partial_x^2,\\ (\delta\partial_\delta)^2,\quad\partial_j\partial_k, \quad \partial_j\partial_x,\quad\partial_j\delta\partial_\delta,\quad\partial_x\delta\partial_\delta \end{array}\right\}\ee
holds for all $0<\delta<b$ and $n\in\{0,1,\ldots,m\}$.  By taking the supremum of (\ref{tildeHold}) over each of the open sets $U_n$ of (\ref{compactembedd}) and using (\ref{otherSchauder}) to perform the same iterative sequence of steps as in part \ref{firstlimitlempart1} of this proof, we find that
\be\label{lastSchauderold}\sup_{\mathcal{K}}|\partial^\varpi H(\boldsymbol{\xi};x,\delta)|=O(1)\quad\text{as $\delta\downarrow0$}.\ee
The rest of the proof of case \ref{firstlimitlempart2} proceeds identically to the proof of case \ref{firstlimitlempart1}.

\item\label{firstlimitlempart3} 
Finally, the proof of case \ref{firstlimitlempart3} with $i=\iota$ is identical to the proof of case \ref{firstlimitlempart2} with $i=\iota+1$, with some small differences.  For brevity, we only point out those differences.

With $i=\iota$, we let $x:=x_i$, $\delta:=x_i-x_{i-1}$, we relabel the coordinates in $\{x_j\}_{j\neq i-1,i}$ in increasing order as $\{\xi_1,\xi_2,$ $\ldots,\xi_{M-2}\}$, and we let $\boldsymbol{\xi}:=(\xi_1,\xi_2,\ldots,\xi_{M-2})$.  Also, we define
\be\label{FtoH2} {\rm F}(\boldsymbol{\xi};x,\delta):=F(\xi_1,\xi_2,\ldots,\xi_{i-2},x-\delta,x,\xi_{i-1},\ldots,\xi_M),\quad H(\boldsymbol{\xi};x,\delta):=\delta^{-\Delta^-(h)} {\rm F}(\boldsymbol{\xi};x,\delta).\ee
Then, we write the null-state PDE (\ref{nullstate}) centered on $x_i$ as $\mathcal{L}[{\rm F}]=\mathcal{Q}[{\rm F}]$, where $\mathcal{L}$ is given by (\ref{i+11}) with $\theta_1$ replaced by $h$ (so the characteristic exponents of $\mathcal{L}$ are now $\Delta^-(h)$ and $\Delta^+(h)$), and where $\mathcal{Q}$ is (similar to $\mathcal{P}$ (\ref{newi+12}))
\be\label{Q}\mathcal{Q}[{\rm F}](\boldsymbol{\xi};x,\delta):=\Bigg[-\frac{\partial_x}{\delta}+\sum_{j=1}^{M-2}\left(\frac{\theta_1}{(\xi_j-x+\delta)^2}-\frac{\partial_j}{\xi_j-x+\delta}\right)\Bigg]{\rm F}(\boldsymbol{\xi};x,\delta).\ee
By following the reasoning presented in the proof of case \ref{firstlimitlempart2}, we again find (\ref{tildeHold}) with $\mathcal{P}$ replaced by $\mathcal{Q}$, and the estimate (\ref{lastSchauderold}).  The rest of the proof of case \ref{firstlimitlempart3} proceeds identically to the proof of case \ref{firstlimitlempart2}.
\end{enumerate}
With cases \ref{firstlimitlempart1}--\ref{firstlimitlempart3} justified, we have proven the lemma.  \end{proof}

The proof of lemma \ref{firstlimitlem} implies some interesting integral equations that $H(\boldsymbol{\xi};x,\delta;y)$ and $ {\rm F}(\boldsymbol{\xi};x,\delta;y)$ must satisfy in the case $i\not\in\{\iota,\iota+1\}$.  After replacing $\delta$ with zero and then replacing $b$ with $\delta$ in (\ref{tildeH}), we find
\be\label{futurepropagator}H(\boldsymbol{\xi};x,\delta;y)=H(\boldsymbol{\xi};x,0;y)\,+\,\bindnasrepma(\theta_1)^{-1}\delta\partial_\delta H(\boldsymbol{\xi};x,\delta;y)-\frac{4/\kappa}{\bindnasrepma(\theta_1)}\sideset{}{_0^\delta}\int \eta\mathcal{M}[H](\boldsymbol{\xi};x,\eta;y)\,{\rm d}\eta,\ee
where $H(\boldsymbol{\xi};x,0;y)$ is the limit of $H(\boldsymbol{\xi};x,\delta;y)$ as $\delta\downarrow0$. This integral equation is interesting because it integrates over $0<\eta<\delta$ instead of over $\delta<\eta<b$ for some small, arbitrary cutoff $b$.  Furthermore, we find that for all $0<\delta<b$,
\begin{multline}\label{Fpropagator} {\rm F}(\boldsymbol{\xi};x,\delta;y)=\left(\frac{\delta}{b}\right)^{\Delta^+(\theta_1)} {\rm F}(\boldsymbol{\xi};x,b;y)+\Bigg[\left(\frac{\delta}{b}\right)^{\Delta^-(\theta_1)}-\hspace{.2cm}\left(\frac{\delta}{b}\right)^{\Delta^+(\theta_1)}\Bigg]b^{\Delta^-(\theta_1)}H(\boldsymbol{\xi};x,0;y)\\
-\frac{4}{\kappa}\delta^{\Delta^+(\theta_1)}\sideset{}{_\delta^b}\int\frac{1}{\beta}\sideset{}{_0^\beta}\int\beta^{-\bindnasrepma(\theta_1)}\eta^{-\Delta^-(\theta_1)}\,\eta\mathcal{M}[ {\rm F}](\boldsymbol{\xi};x,\eta;y)\,{\rm d}\eta\,{\rm d}\beta\end{multline}
after we move the middle term on the right side of (\ref{futurepropagator}) to the left side, replace $\delta$ with $\beta$, and integrate both sides over $\beta\in[\delta,b]$ with $b$ small.  If $i\in\{\iota,\iota+1\}$, then after starting with (\ref{tildeHold}) and following the same steps, we find that
\begin{multline}\label{otherFpropagator} {\rm F}(\boldsymbol{\xi};x,\delta)=\left(\frac{\delta}{b}\right)^{\Delta^+(h)} {\rm F}(\boldsymbol{\xi};x,b)+\Bigg[\left(\frac{\delta}{b}\right)^{\Delta^-(h)}-\hspace{.2cm}\left(\frac{\delta}{b}\right)^{\Delta^+(h)}\Bigg]b^{\Delta^-(h)}H(\boldsymbol{\xi};x,0)\\
-\frac{4}{\kappa}\delta^{\Delta^+(h)}\sideset{}{_\delta^b}\int\frac{1}{\beta}\sideset{}{_0^\beta}\int\beta^{-\bindnasrepma(h)}\eta^{-\Delta^-(h)}\,\left\{\begin{array}{ll}\eta\mathcal{P}[ {\rm F}](\boldsymbol{\xi};x,\eta), & i=\iota+1 \\ \eta\mathcal{Q}[ {\rm F}](\boldsymbol{\xi};x,\eta), & i=\iota\end{array}\right\}\,{\rm d}\eta\,{\rm d}\beta\end{multline}
for all $0<\delta<b$, with $H(\boldsymbol{\xi};x,0)$ the limit of $H(\boldsymbol{\xi};x,\delta)$ as $\delta\downarrow0$.  We use these integral equations in the proof of lemma \ref{secondlimitlem} below.

In lemma \ref{secondlimitlem}, we isolate the second term that appears on the right side of (\ref{ababnormalpowers}) as we collapse the interval $(x_{i-1},x_i)$.  The coefficient of this term is the limit (\ref{hlim+}) appearing in the statement of the lemma, and we identify this limit with the limit (\ref{ss+1lim}) after we set $i=\iota=M$, $M=2N-s+1$, and $h=\theta_s$ in the proof of lemma \ref{alltwoleglem} below.  This point is relevant because we need the latter limit to exist in order to execute the methodology proposed in section \ref{Methodology}.  We note that condition \ref{nextcond3} of lemma \ref{secondlimitlem}, in CFT parlance, implies that only the $+$ term is present in the OPE of (\ref{ababnormalpowers}).  In particular, if $d=\theta_1$, then this OPE comprises terms from only the $\theta_2$ conformal family.

Again, some steps in the proof of lemma \ref{secondlimitlem} below are identical to steps in the proofs of lemmas \red{3} and \red{4} of \cite{florkleb}.  Rather than write them out again, we reference them in \cite{florkleb} for brevity.

\begin{lem}\label{secondlimitlem}
Suppose that $\kappa\in(0,8)$ and $M>2$, and define $\boldsymbol{x}_\delta$ as in (\ref{xdelta}).  If $F:\Omega_0^M\rightarrow\mathbb{R}$ satisfies these conditions,
\begin{enumerate}
\item\label{nextcond1} $F$ satisfies the growth bound (\ref{powerlaw}) for some positive constants $C$ and $p$,
\item\label{nextcond2} $F$ solves the $M-1$ null-state PDEs (\ref{hpde}) and three conformal Ward identities (\ref{wardidh}) stated in condition \ref{cond2} of lemma \ref{firstlimitlem} for some $h>-(\kappa-4)^2/16\kappa$, and
\item\label{nextcond3} the limit $F_{d^-}$ (\ref{hlim-}) is zero for the index $i\in\{2,3,\ldots,M\}$ selected in the definition (\ref{xdelta}) for $\boldsymbol{x}_\delta$,
\end{enumerate}
then the limits
\be\label{hlim+}\begin{gathered}(F_{d^+}\circ\pi_i)(\boldsymbol{x}):=\lim_{\delta\downarrow0}\delta^{-\Delta^+(d)}F(\boldsymbol{x}_\delta),\qquad\lim_{\delta\downarrow0}\delta^{-\Delta^+(d)}\partial_\iota F(\boldsymbol{x}_\delta),\qquad d:=\begin{cases} \theta_1,& i\not\in\{\iota,\iota+1\} \\ h, & i\in\{\iota,\iota+1\}\end{cases}\\
\lim_{\delta\downarrow0}\delta^{-\Delta^+(d)}\partial_jF(\boldsymbol{x}_\delta),\qquad\lim_{\delta\downarrow0}\delta^{-\Delta^+(d)}\partial_j^2F(\boldsymbol{x}_\delta),\qquad j\not\in\begin{cases}\{i,\iota\},& i\not\in\{\iota,\iota+1\} \\ \{i-1,i\}, & i\in\{\iota,\iota+1\}\end{cases}\end{gathered}\ee
exist and are approached uniformly over every compact subset of $\pi_i(\Omega_0^M)$. Finally, if $F$ is not zero, then the first limit $F_{d^+}$ of (\ref{hlim+}) is not zero.  
\end{lem}

\begin{proof}  Throughout, we let $\mathcal{K}$ be an arbitrary compact subset of $\pi_i(\Omega_0^M)$.  Below, we divide the proof into two cases: case \ref{secondlimitlempart1} with $i\not\in\{\iota,\iota+1\}$, and case \ref{secondlimitlempart2} with $i\in\{\iota,\iota+1\}$. The proof of case \ref{secondlimitlempart2} is almost identical to that for case \ref{secondlimitlempart1}, so our exposition for the former case is correspondingly abbreviated.

\begin{enumerate}[leftmargin=*]
\item\label{secondlimitlempart1} 
We suppose that $i\not\in\{\iota,\iota+1\}$.  For this case, we define $x$, $y$, $\delta$, ${\rm F}$, $H$ (\ref{tilde}), $\boldsymbol{\xi}$, and $b$ as in case \ref{firstlimitlempart1} of the proof of lemma \ref{firstlimitlem}, and we define 
\be\label{Kdefn}E(\boldsymbol{\xi};x,\delta;y):=\delta^{-\Delta^+(\theta_1)}{\rm F}(\boldsymbol{\xi};x,\delta;y)=\delta^{-\bindnasrepma(\theta_1)}H(\boldsymbol{\xi};x,\delta;y).\ee 
Lemma \ref{firstlimitlem} implies that $E(\boldsymbol{\xi};x,\delta;y)=O(\delta^{-q})$ as $\delta\downarrow0$, where we define $q:=\bindnasrepma(\theta_1)>0$ throughout this proof.

To begin, we bound the growth of $E(\boldsymbol{\xi};x,\delta;y)$ as $\delta\downarrow0$.  Condition \ref{nextcond3} of the lemma implies that $H(\boldsymbol{\xi};x,0;y)=0$ in (\ref{Fpropagator}), so after expressing (\ref{Fpropagator}) in terms of $E$ with this condition, we find the integral equation
\be\label{Kpropagator}E(\boldsymbol{\xi};x,\delta;y)=E(\boldsymbol{\xi};x,b;y)-\frac{4}{\kappa}\sideset{}{_\delta^b}\int\frac{1}{\beta}\sideset{}{_0^\beta}\int\left(\frac{\eta}{\beta}\right)^{\bindnasrepma(\theta_1)}\eta\mathcal{M}[E](\boldsymbol{\xi};x,\eta;y)\,{\rm d}\eta\,{\rm d}\beta\ee 
for all $0<\delta<b$ and with $\mathcal{M}$ given in (\ref{i+12}).  Next, we choose bounded open sets $ U_0$, $ U_1,\ldots, U_m$, with $m:=\lceil q\rceil$, compactly embedded within each other and $\mathcal{K}$ as in (\ref{compactembedd}).  Taking the supremum of (\ref{Kpropagator}) over $ U_n$ gives
\be\label{Fpropagator2nd} \sup_{ U_n}|E(\boldsymbol{\xi};x,\delta;y)|\leq\sup_{ U_n}|E(\boldsymbol{\xi};x,b;y)|+\frac{4}{\kappa}\sideset{}{_\delta^b}\int\frac{1}{\beta}\sideset{}{_0^\beta}\int\left(\frac{\eta}{\beta}\right)^{\bindnasrepma(\theta_1)}\sup_{ U_n}|\eta
[E](\boldsymbol{\xi};x,\eta;y)|\,{\rm d}\eta\,{\rm d}\beta\ee
for all $0<\delta<b$.  Lemma \ref{firstlimitlem} with (\ref{Kdefn}) implies that the supremum of the integrand in (\ref{Fpropagator2nd})  is $O(\eta^{-q})$ as $\eta\downarrow0$.  (We recall that $q:=\bindnasrepma(\theta_1)$ in this proof.)  Hence, we may estimate the definite integral in (\ref{Fpropagator2nd}) to find
\be\label{firstsup}\sup_{ U_1}|E(\boldsymbol{\xi};x,\delta;y)|=\begin{cases}O(\delta^{-q+1}),& m>1 \\ O(1), & m=1\end{cases}\quad\text{as $\delta\downarrow0$}.\ee
Supposing that $m>1$, we recall from the proof of lemma \ref{firstlimitlem} that $H$, and therefore $E$, satisfies the Schauder interior estimate (\ref{Schauder})
\be\label{ESchauder}d_{n+1}\sup_{ U_{n+1}}|\partial^\varpi E(\boldsymbol{\xi};x,\delta;y)|\leq C(R_n)\sup_{ U_n}|E(\boldsymbol{\xi};x,\delta;y)|\quad\begin{array}{l}\\ \end{array}\ee
where $C$ is some positive-valued function, $d_n=\text{dist}(\partial U_n,\partial U_{n-1})$, $R_n=\text{diam}( U_n)/2$, and $\partial^\varpi$ is given by (\ref{multiindex1}).  If $m>1$, then after using (\ref{firstsup}) with (\ref{ESchauder}) to estimate the integrand of (\ref{Fpropagator2nd}) with $n=2$, we find that
\be\sup_{ U_2}|E(\boldsymbol{\xi};x,\delta;y)|=\begin{cases} O(\delta^{-q+2}), & m>2 \\ O(1), & m=2\end{cases}\quad\text{as $\delta\downarrow0$}.\ee
After repeating this process another $m-3$ times, we ultimately find that the left side of (\ref{Fpropagator2nd}) with $n=m-1$ is $O(\delta^{-q+m-1})$ as $\delta\downarrow0$.  Repeating this process one last time and invoking (\ref{Schauder}), we find that because $\mathcal{K}\subset\subset U_m$,
\be\label{lastSchauderK}\sup_{ U_m}|E(\boldsymbol{\xi};x,\delta;y)|=O(1)\quad\text{as $\delta\downarrow0$}\quad\Longrightarrow\quad \sup_{\mathcal{K}}|\partial^\varpi E(\boldsymbol{\xi};x,\delta;y)|=O(1)\quad\text{as $\delta\downarrow0$}.\ee

Now we use (\ref{lastSchauderK}) to show that the limits in (\ref{hlim+}) exist and are approached uniformly over $\mathcal{K}$.  The reasoning follows that used for the proof of lemma \red{4} in \cite{florkleb}.  Because the integrand of (\ref{Kpropagator}) with $\beta$ fixed is bounded over $0<\eta<b$,
\be \sup_{0<\delta<b}|E(\boldsymbol{\xi};x,\delta;y)-E(\boldsymbol{\xi};x,b;y)|\xrightarrow[b\downarrow0]{}0.\ee
Hence, the superior and inferior limits of $E(\boldsymbol{\xi};x,\delta;y)$ as $\delta\downarrow0$ are equal, so the limit ${\rm F}_2(\boldsymbol{\xi};x;y)$ of $E(\boldsymbol{\xi};x,\delta;y)$ as $\delta\downarrow0$ exists.  After taking the supremum of (\ref{Kpropagator}) over $\mathcal{K}$, sending $\delta\downarrow0$, and then replacing $b$ with $\delta$, we also find
\bea\label{supKK}\sup_{\mathcal{K}}|E(\boldsymbol{\xi};x,\delta;y)-{\rm F}_2(\boldsymbol{\xi};x;y)|&\leq&\frac{4}{\kappa}\sideset{}{_0^\delta}\int\frac{1}{\beta}\sideset{}{_0^\beta}\int\sup_\mathcal{K}|\eta\mathcal{M}[E](\boldsymbol{\xi};x,\eta;y)|\,{\rm d}\eta\,{\rm d}\beta\\
\label{followinglim}&\xrightarrow[\delta\downarrow0]{}&0.\eea
The limit (\ref{followinglim}) follows because the supremum on the right side of (\ref{supKK}) is bounded over $0<\eta<b$ thanks to (\ref{i+12}, \ref{lastSchauderK}).  Hence, the limit ${\rm F}_2(\boldsymbol{\xi};x;y)$ is approached uniformly over $\mathcal{K}$, and we identify it with $(F_{d^+}\circ\pi_i)(\boldsymbol{x})$ in (\ref{hlim+}).

We may show that the derivatives of $E$ with respect to the coordinates of $\boldsymbol{\xi}$, $x$, and $y$ approach limits (\ref{hlim+}) as $\delta\downarrow0$ uniformly over $\mathcal{K}$ by differentiating (\ref{Kpropagator}) with respect to these variables and following the same procedure.  Finally, we may prove the same for second derivatives of $E$ with respect to the coordinates of $\boldsymbol{\xi}$ by isolating these derivatives from (\ref{hpde}) in terms of quantities with the limits  (\ref{hlim+}) as $\delta\downarrow0$.

Next, we prove that if the limit ${\rm F}_2$ is zero, then ${\rm F}$ is zero.  After sending $\delta\downarrow0$ and replacing $b$ with $\delta$ in (\ref{Kpropagator}), inserting the assumption that ${\rm F}_2=0$, and taking the supremum over an open ball $\mathcal{B}_1\subset\subset\mathcal{K}$ of radius $R_1$, we find 
\be\label{V1}\sup_{\mathcal{B}_1}|E(\boldsymbol{\xi};x,\delta;y)|\leq\frac{4}{\kappa}\sideset{}{_0^\delta}\int\frac{1}{\beta}\sideset{}{_0^\beta}\int \sup_{\mathcal{B}_1}|\eta\mathcal{M}[E](\boldsymbol{\xi};x,\eta;y)|\,{\rm d}\eta\,{\rm d}\beta.\ee
Because the integrand is bounded over $0<\eta<\beta$ by a constant independent of $\beta$, the left side of (\ref{V1}) is $O(\delta)$.  With $\mathcal{B}_2$ an open ball concentric with $\mathcal{B}_1$ and of radius $R_2=R_1+d_2>R_1$, the Schauder estimate (\ref{ESchauder}) then gives
\be\label{Ksups}\sup_{\mathcal{B}_1}| E(\boldsymbol{\xi};x,\delta;y)|\leq c_1\frac{\delta}{d_1}\quad\Longrightarrow\quad  d_2\sup_{\mathcal{B}_2}|\delta\mathcal{M}[E](\boldsymbol{\xi};x,\delta;y)|\leq c_1c_2\frac{\delta}{d_1},\ee
for some constants $c_1$, $d_1$ (to be specified in (\ref{dk}) below), and $c_2:=C(R_1)$, with $C$ a continuous function over $(0,R_1]$ (slightly different from $C$ defined in (\ref{Schauder})).  Next, we iterate this estimation an infinite number of times.  We let 
\be\mathcal{B}_\infty\subset\subset\ldots\subset\subset\mathcal{B}_{k+1}\subset\subset\mathcal{B}_k\subset\subset\ldots\subset\subset\mathcal{B}_1\ee
be an infinite sequence of concentric balls, with $R_k$ the radius of $\mathcal{B}_k$, such that their intersection $\bigcap_k\mathcal{B}_k$ is a ball $\mathcal{B}_\infty$ of radius $R_\infty<R_k$ for all $k\in\mathbb{Z}^+$.  We choose the radii of the balls such that for all $k>2$,
\be\label{dk}d_k:=R_{k-1}-R_k=\frac{d_1}{k^{3/2}},\quad\text{where $d_1:=\frac{R_1-R_\infty}{\zeta(3/2)-1}$}\ee
and $\zeta$ is the Riemann zeta function.  This choice satisfies the necessary condition $\sum_{k=2}^\infty d_k=R_1-R_\infty$.  After using (\ref{Ksups}) to estimate the definite integral in (\ref{V1}) with $\mathcal{B}_2$ replacing $\mathcal{B}_1$ and applying (\ref{ESchauder}) again, we find 
\be\label{Ksupsnext}\sup_{\mathcal{B}_2}|E(\boldsymbol{\xi};x,\delta;y)|\leq\frac{4}{\kappa}c_1c_2\frac{\delta^2}{2^2d_1d_2}\quad\Longrightarrow\quad  d_3\sup_{\mathcal{B}_3}|\delta\mathcal{M}[E](\boldsymbol{\xi};x,\delta;y)|\leq\frac{4}{\kappa}c_1c_2c_3\frac{\delta^2}{2^2d_1d_2},\ee
with $c_3:=C(R_2)$.  Letting $c_k:=C(R_{k-1})$ for $k>1$, we repeat this process an infinite number of times to ultimately find that for all $k\in\mathbb{Z}^+$,  
\be\label{Ksupsnextk}\sup_{\mathcal{B}_k}|E(\boldsymbol{\xi};x,\delta;y)|\leq\frac{\kappa}{4}c_1c_2\dotsm c_k\frac{(4\delta/\kappa)^k}{(k!)^2d_1d_2\dotsm d_k}.\ee
Because $C$ is continuous on $(R_\infty,R_1)$, the sequence $c_k$ is bounded.  Therefore, after substituting the formula for $d_k$ (\ref{dk}) into (\ref{Ksupsnextk}) and recalling that $\mathcal{B}_\infty\subset\subset\mathcal{B}_k$ for all $k\in\mathbb{Z}^+$, we find 
\bea\label{Ksupslast}\sup_{\mathcal{B}_\infty}|E(\boldsymbol{\xi};x,\delta;y)|&\leq& \frac{\kappa}{4\sqrt{k!}}\left(\frac{4\delta}{\kappa d_1}\sup_{k\in\mathbb{Z}^+}c_k\right)^k,\quad k\in\mathbb{Z}^+\\
&\xrightarrow[k\rightarrow\infty]{}&0.\eea
Because $\mathcal{B}_\infty$ is an arbitrary ball in $\mathcal{K}$ and $\mathcal{K}$ is an arbitrary compact subset of $\pi_i(\Omega_0^M)$, it follows that $E$, and therefore ${\rm F}$ (\ref{Kdefn}), is zero.  We conclude that if ${\rm F}$ is not zero, then ${\rm F}_2$ is not zero.

\item\label{secondlimitlempart2} We suppose that $i\in\{\iota,\iota+1\}$.  The proof of the case $i=\iota+1$ (resp.\ $i=\iota$) is identical to that of the previous case, except that we replace ${\rm F}(\boldsymbol{\xi};x,\delta,y)$ with ${\rm F}(\boldsymbol{\xi};x,\delta)$ (\ref{FtoH}) (resp.\ (\ref{FtoH2})), $\theta_1$ with $h$, and $\mathcal{M}$ (\ref{i+12}) with $\mathcal{P}$ (\ref{newi+12}) (resp.\ $\mathcal{Q}$ (\ref{Q})).  After defining
\be\label{2ndKdefn} E(\boldsymbol{\xi};x,\delta):=\delta^{-\Delta^+(h)}{\rm F}(\boldsymbol{\xi};x,\delta)=\delta^{-\bindnasrepma(h)}H(\boldsymbol{\xi};x,\delta;y),\ee
we repeat the steps of the previous case \ref{secondlimitlempart1} to derive from (\ref{otherFpropagator}) the integral equation
\be\label{otherKpropagator}E(\boldsymbol{\xi};x,\delta)=E(\boldsymbol{\xi};x,b)-\frac{4}{\kappa}\sideset{}{_\delta^b}\int\frac{1}{\beta}\sideset{}{_0^\beta}\int\left(\frac{\eta}{\beta}\right)^{\bindnasrepma(h)}\left\{\begin{array}{ll}\eta\mathcal{P}[E](\boldsymbol{\xi};x,\eta), & i=\iota+1 \\ \eta\mathcal{Q}[E](\boldsymbol{\xi};x,\eta), & i=\iota\end{array}\right\}\,{\rm d}\eta\,{\rm d}\beta\ee
for all $0<\delta<b$.  From here, the rest of the proof proceeds exactly as did the proof of the case \ref{secondlimitlempart1}, except that we now use the other Schauder estimate (\ref{otherSchauder}).
\end{enumerate} 
With both cases justified, the proof is complete.
\end{proof}

\begin{lem}\label{pdelem} Suppose that $\kappa\in(0,8)$, $M>2$, and $h>-(\kappa-4)^2/16\kappa$, and let $h^+:=\theta_1+h+\Delta^+(h)$.  If $F:\Omega_0^M\rightarrow\mathbb{R}$ satisfies conditions \ref{nextcond1}--\ref{nextcond3} of lemma \ref{secondlimitlem} with $i=\iota>1$, then the limit $F_{h^+}$ (\ref{hlim+}) satisfies the (modified) null-state PDE centered on $x_j$
\be\label{nullstatemod}\Bigg[\frac{\kappa}{4}\partial_j^2+\sum_{k\neq j,\iota-1,\iota}^M\left(\frac{\partial_k}{x_k-x_j}-\frac{\theta_1}{(x_k-x_j)^2}\right)+\frac{\partial_{\iota-1}}{x_{\iota-1}-x_j}-\frac{h^+}{(x_{\iota-1}-x_j)^2}\Bigg](F_{h^+}\circ\pi_\iota)(\boldsymbol{x})=0\ee
for each $j\in\{1,2,\ldots,\iota-2,\iota+1,\ldots,M\}$ and the (modified) conformal Ward identities
\be\label{wardidmod} \begin{gathered}\sum_{k\neq\iota}^M\partial_k(F_{h^+}\circ\pi_\iota)(\boldsymbol{x})=0,\qquad
\Bigg[\sum_{k\neq\iota-1,\iota}^M(x_k\partial_k+\theta_1)+x_{\iota-1}\partial_{\iota-1}+h^+\Bigg](F_{h^+}\circ\pi_\iota)(\boldsymbol{x})=0,\\
\Bigg[\sum_{k\neq\iota-1,\iota}^M(x_k^2\partial_k+2\theta_1x_k)+x_{\iota-1}^2\partial_{\iota-1}+2h^+x_{\iota-1}\Bigg](F_{h^+}\circ\pi_\iota)(\boldsymbol{x})=0.\end{gathered}\ee
\end{lem}

\begin{proof} We define $\delta:=x_\iota-x_{\iota-1}$, $x:=x_\iota$, $\boldsymbol{\xi}$, and ${\rm F}(\boldsymbol{\xi};x,\delta)$ (\ref{FtoH2}) as in case \ref{firstlimitlempart3} from the proof of lemma \ref{firstlimitlem} with $i=\iota$, and we define $E$ as in (\ref{2ndKdefn}).  We note that the limit $F_{h^+}$ (\ref{hlim+}) is given by
\be (F_{h^+}\circ\pi_\iota)(\boldsymbol{x})=\lim_{\delta\downarrow0}E(\boldsymbol{\xi};x,\delta).\ee
From the proof of lemma \ref{secondlimitlem}, we know that $E$ satisfies the integral equation (\ref{otherKpropagator}) with $i=\iota$.  After differentiating this integral equation with respect to $\delta$, we find
\be\label{partialKpropagator}\delta\partial_\delta E(\boldsymbol{\xi};x,\delta)=\frac{4}{\kappa}\sideset{}{_0^\delta}\int\left(\frac{\eta}{\delta}\right)^{\bindnasrepma(h)}\eta\mathcal{Q}[E](\boldsymbol{\xi};x,\eta)\,{\rm d}\eta.\ee
Now, lemma \ref{secondlimitlem} implies that the integrand of (\ref{partialKpropagator}) is bounded as $\delta\downarrow0$.  Therefore, the right side of (\ref{partialKpropagator}) vanishes as $\delta\downarrow0$, so we have
\be\label{zerolim}\lim_{\delta\downarrow0}\delta\partial_\delta E(\boldsymbol{\xi};x,\delta)=0.\ee

With the limit (\ref{zerolim}) established, we straightforwardly prove the lemma by examining the system (\ref{hpde}, \ref{wardidh}) in the limit $\delta\downarrow0$.  In terms of the variables $\boldsymbol{\xi}$, $x$, and $\delta$, the null-state PDE (\ref{hpde}) centered on $x_j$ with $j\not\in\{\iota-1,\iota\}$ is
\begin{multline}\label{Knullstate}\Bigg[\frac{\kappa}{4}\partial_j^2+\sum_{k\neq j}\left(\frac{\partial_k}{\xi_k-\xi_j}-\frac{\theta_1}{(\xi_k-\xi_j)^2}\right)+\frac{\partial_x}{x-\xi_j}-\frac{\theta_1}{(x-\xi_j)^2}\\
-\frac{\delta\partial_\delta}{(x-\xi_j)(x-\delta-\xi_j)}-\frac{\Delta^+(h)}{(x-\xi_j)(x-\delta-\xi_j)}-\frac{h}{(x-\delta-\xi_j)^2}\Bigg]E(\boldsymbol{\xi};x,\delta)=0,\end{multline}
and the conformal Ward identities (\ref{wardidh}) are
\begin{gather}\begin{aligned}&\label{wK}\bigg[\sideset{}{_k}\sum\partial_k+\partial_x\bigg]E(\boldsymbol{\xi};x,\delta)=0,\quad\bigg[\sideset{}{_k}\sum(\xi_k\partial_k+\theta_1)+x\partial_x+\theta_1+\delta\partial_\delta+h+\Delta^+(h)\bigg]E(\boldsymbol{\xi};x,\delta)=0,\\
&\bigg[\sideset{}{_k}\sum(\xi_k^2\partial_k+2\theta_1\xi_k+x^2\partial_x+2\theta_1x+(2x-\delta)\delta\partial_\delta+2h(x-\delta)+(2x-\delta)\Delta^+(h)\bigg]E(\boldsymbol{\xi};x,\delta)=0.\end{aligned}\end{gather}
Because all of the quantities in (\ref{Knullstate}, \ref{wK}) approach their limits uniformly over compact subsets of $\pi_\iota(\Omega_0^M)$ as $\delta\downarrow0$, we may commute this limit with all differentiations in these equations that are not with respect to $\delta$.  After doing this and applying (\ref{zerolim}), we find that (\ref{Knullstate}) and (\ref{wK}) respectively go to (\ref{nullstatemod}) and (\ref{wardidmod}).
\end{proof}
If $F$ satisfies conditions \ref{nextcond1} and \ref{nextcond2} but not \ref{nextcond3} of lemma \ref{secondlimitlem}, then it is easy to show that the limit $F_{h^-}$ (\ref{hlim-}) satisfies the system (\ref{nullstatemod}) and (\ref{wardidmod}) with $h^+$ replaced by $h^-:=\theta_1+h+\Delta^-(h).$  (Indeed, to prove this claim, we follow the proof of lemma \red{5} in \cite{florkleb}.)  In either case, we interpret the system (\ref{nullstatemod}, \ref{wardidmod}) as the collection of null-state PDEs and conformal Ward identities for a certain $(M-1)$-point CFT correlation function.  This correlation function has a one-leg boundary operator at the coordinates in $\{x_j\}_{j\neq\iota-1}$ and a primary operator at $x_{\iota-1}$ with conformal weight $h^+$ or $h^-$.

\subsection{Estimates involving two intervals}\label{twointervals}

Lemmas \ref{kpzlem}--\ref{pdelem} begin the task of constructing the functions $F_1$, $F_2,\ldots,F_{2N-1}$ described in section \ref{Methodology}, but they are not sufficient to complete it because of the subtleties involved in collapsing neighboring intervals.  In this section, we derive two estimates, stated in lemmas \ref{farlem} and \ref{closelem}, that complete the construction.

To see how lemmas \ref{kpzlem}--\ref{pdelem} work together, we study the first steps of the construction that take us from $F_1:=F\in\mathcal{S}_N\setminus\{0\}$ to $F_2$ to $F_3$ to $F_4$.  First supposing that $(x_{2N-1},x_{2N})$ is a two-leg interval of $F$, lemmas \ref{kpzlem}, \ref{secondlimitlem}, and \ref{pdelem} imply
\begin{multline}\label{firsttwoleg}(x_{2N}-x_{2N-1})^{-\Delta^-(\theta_1)}F_1(\boldsymbol{x}) \xrightarrow[\substack{x_{2N}\\\qquad\rightarrow x_{2N-1}}]{\boldsymbol{x}\in\Omega_0^{2N}}0\quad \text{(the two-leg interval condition)} \\
\Longrightarrow\quad (F_2\circ\pi_{2N})(\boldsymbol{x})\,\,\,:=\lim_{\substack{x_{2N}\rightarrow x_{2N-1} \\ \boldsymbol{x}\in\Omega_0^{2N}}}(x_{2N}-x_{2N-1})^{-\Delta^+(\theta_1)}F_1(\boldsymbol{x})\,\,\left\{\begin{array}{ll} 1.\,\,\text{exists and is not zero},\\ 2.\,\,\text{satisfies (\ref{hpde}, \ref{wardidh}) of lemma \ref{firstlimitlem}} \\ \text{with $h=\theta_2$ and $\iota=M=2N-1$.}\end{array}\right.\end{multline}

Next, we suppose that both $(x_{2N-2},x_{2N-1})$ and $(x_{2N-1},x_{2N})$ are two-leg intervals of $F\in\mathcal{S}_N\setminus\{0\}$, and we wish to construct $F_3$ from $F_1$ by proving the following similar statement (figure \ref{Collapses}):
\begin{multline}\label{toprove}\left\{\begin{array}{l}(x_{2N\hphantom{-0}}-x_{2N-1})^{-\Delta^-(\theta_1)}F_1(\boldsymbol{x}) \xrightarrow[\substack{x_{2N\hphantom{-0}}\\\qquad\rightarrow x_{2N-1}}]{\boldsymbol{x}\in\Omega_0^{2N}}0\\
(x_{2N-1}-x_{2N-2})^{-\Delta^-(\theta_1)}F_1(\boldsymbol{x}) \xrightarrow[\substack{x_{2N-1}\\\qquad\rightarrow x_{2N-2}}]{\boldsymbol{x}\in\Omega_0^{2N}}0\end{array}\right.\quad \text{(the two-leg interval conditions)} \\
\Longrightarrow\quad (F_3\circ\pi_{2N-1})(\boldsymbol{x})\,\,\,:=\lim_{\substack{x_{2N-1}\rightarrow x_{2N-2} \\ \boldsymbol{x}\in\Omega_0^{2N-1}}}(x_{2N-1}-x_{2N-2})^{-\Delta^+(\theta_2)}F_2(\boldsymbol{x})\,\,\left\{\begin{array}{ll} 1.\,\,\text{exists and is not zero},\\ 2.\,\,\text{satisfies (\ref{hpde}, \ref{wardidh}) of lemma \ref{firstlimitlem}} \\ \text{with $h=\theta_3$ and $\iota=M=2N-2$.}\end{array}\right.\end{multline}
Taken together, lemmas \ref{kpzlem}--\ref{pdelem} almost prove (\ref{toprove}).  Indeed, (\ref{firsttwoleg}) gives $F_2$ with the stated properties, and then lemma \ref{firstlimitlem} with $h=\theta_2$ and $\iota=M=2N-1$ says that the limit
\be\label{thelim}\lim_{x_{2N-1}\rightarrow x_{2N-2}}(x_{2N-1}-x_{2N-2})^{-\Delta^-(\theta_2)}F_2(\boldsymbol{x}),\quad\boldsymbol{x}\in\Omega_0^{2N-1}\ee
exists.  Finally, lemmas \ref{kpzlem}, \ref{secondlimitlem}, and \ref{pdelem} imply that if this limit vanishes, then the conclusion of (\ref{toprove}) regarding $F_3$ follows.  Hence, proving that (\ref{thelim}) does vanish is what is left.  Indeed, this follows from the estimate
\begin{multline}\label{2claim3}\begin{cases} (x_{2N\hphantom{-0}}-x_{2N-1})^{-\Delta^-(\theta_1)}F_1(\boldsymbol{x})& \xrightarrow[\substack{x_{2N\hphantom{-0}}\\\qquad\rightarrow x_{2N-1}}]{\boldsymbol{x}\in\Omega_0^{2N}}0 \\ (x_{2N-1}-x_{2N-2})^{-\Delta^-(\theta_1)}F_1(\boldsymbol{x})& \xrightarrow[\substack{x_{2N-1}\\\qquad\rightarrow x_{2N-2}}]{\boldsymbol{x}\in\Omega_0^{2N}}0\end{cases}\\
\Longrightarrow\quad F_1(\boldsymbol{x})=O\left(\begin{array}{l}\begin{aligned}(x_{2N-1}-&x_{2N-2})^{\Delta^+(\theta_1)} \\ \times & (x_{2N}-x_{2N-1})^{\Delta^+(\theta_1)}(x_{2N}-x_{2N-2})^{\Delta^+(\theta_1)}\end{aligned}\end{array}\right)\quad\text{as}\,\,\,\left\{\begin{array}{l}\text{$x_{2N}\rightarrow x_{2N-1}$} \\ \text{$x_{2N-1}\rightarrow x_{2N-2}$}\end{array}\right.,\end{multline}
derived in lemma \ref{closelem} below.  Estimate (\ref{2claim3}) reveals the behavior of $F(\boldsymbol{x})$ as we collapse two adjacent intervals \emph{simultaneously.}

\begin{figure}[t]
\centering
\includegraphics[scale=0.27]{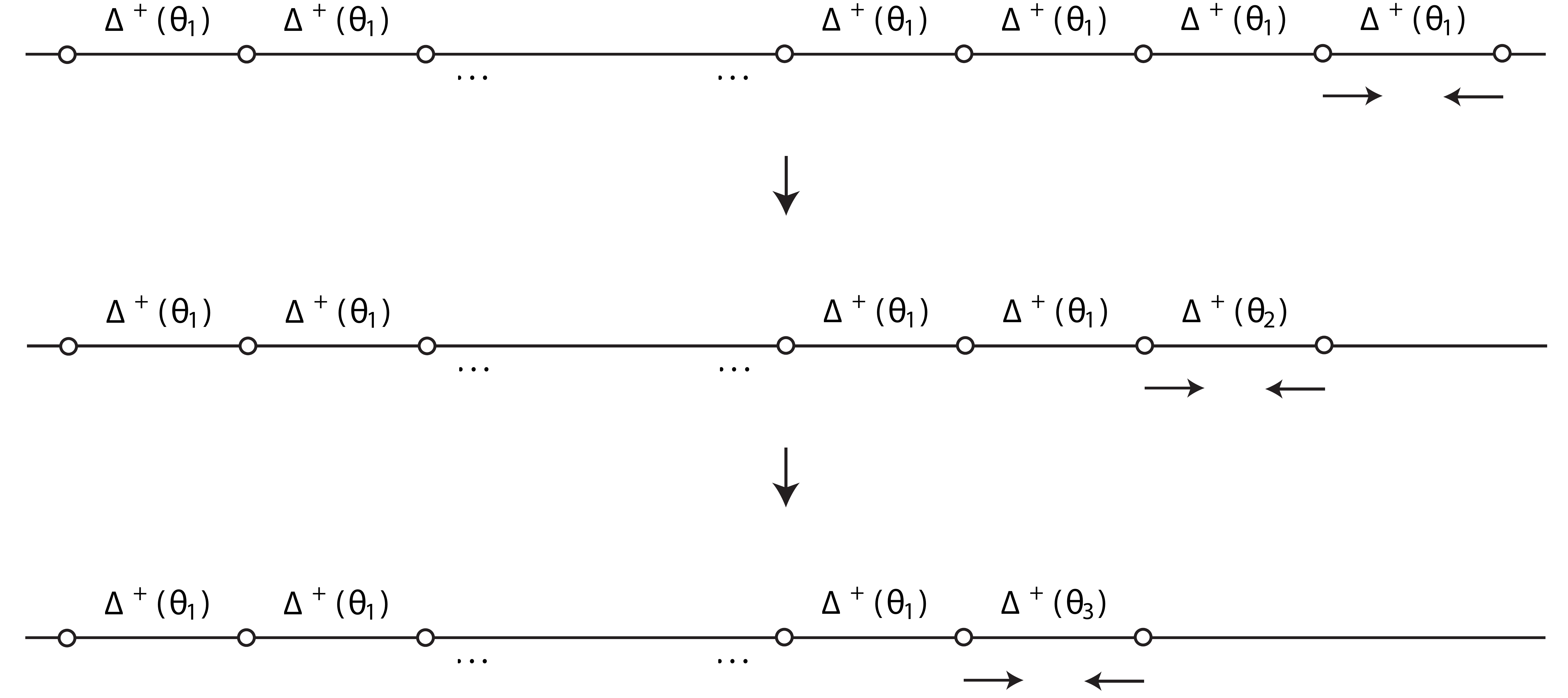}
\caption{Interval collapses sending $F_1$ (top) to $F_2$ (middle) to $F_3$ (bottom), etc.  On each line, we collapse the rightmost interval to take us to the next line beneath.  The power law for each interval collapse is shown above the respective interval.}
\label{Collapses}
\end{figure}

Continuing, we suppose that $(x_{2N-3},x_{2N-2})$, $(x_{2N-2},x_{2N-1})$, and $(x_{2N-1},x_{2N})$ are two-leg intervals of $F\in\mathcal{S}_N\setminus\{0\}$, and we wish to construct $F_4$ from $F_1$ by proving the following statement (figure \ref{Collapses}):
\begin{multline}\label{toprove2}\left\{\begin{array}{l}(x_{2N\hphantom{-0}}-x_{2N-1})^{-\Delta^-(\theta_1)}F_1(\boldsymbol{x}) \xrightarrow[\substack{x_{2N\hphantom{-0}}\\\qquad\rightarrow x_{2N-1}}]{\boldsymbol{x}\in\Omega_0^{2N}}0\\
(x_{2N-1}-x_{2N-2})^{-\Delta^-(\theta_1)}F_1(\boldsymbol{x}) \xrightarrow[\substack{x_{2N-1}\\\qquad\rightarrow x_{2N-2}}]{\boldsymbol{x}\in\Omega_0^{2N}}0\\
(x_{2N-2}-x_{2N-3})^{-\Delta^-(\theta_1)}F_1(\boldsymbol{x}) \xrightarrow[\substack{x_{2N-2}\\\qquad\rightarrow x_{2N-3}}]{\boldsymbol{x}\in\Omega_0^{2N}}0\end{array}\right.\quad \text{(the two-leg interval conditions)} \\
\Longrightarrow\quad (F_4\circ\pi_{2N-2})(\boldsymbol{x}):=\lim_{\substack{x_{2N-2}\rightarrow x_{2N-3} \\ \boldsymbol{x}\in\Omega_0^{2N-2}}}(x_{2N-2}-x_{2N-3})^{-\Delta^+(\theta_3)}F_3(\boldsymbol{x})\,\,\left\{\begin{array}{ll} 1.\,\,\text{exists and is not zero},\\ 2.\,\,\text{satisfies (\ref{hpde}, \ref{wardidh}) of lemma \ref{firstlimitlem}} \\ \text{with $h=\theta_4$ and $\iota=M=2N-3$.}\end{array}\right.\end{multline} 
Taken together, lemmas \ref{kpzlem}--\ref{pdelem} again almost prove (\ref{toprove2}).  Indeed, (\ref{toprove}) gives $F_3$ with the stated properties, and then lemma \ref{firstlimitlem} with $h=\theta_3$ and $\iota=M=2N-2$ says that the limit
\be\label{thelim2}\lim_{x_{2N-2}\rightarrow x_{2N-3}}(x_{2N-2}-x_{2N-3})^{-\Delta^-(\theta_3)}F_3(\boldsymbol{x}),\quad\boldsymbol{x}\in\Omega_0^{2N-2}\ee
exists.  Finally, lemmas \ref{kpzlem}, \ref{secondlimitlem}, and \ref{pdelem} imply that if this limit vanishes, then the conclusion of (\ref{toprove2}) regarding $F_4$ follows.  Hence, proving that (\ref{thelim2}) does vanish is what is left.  Indeed, this follows from the estimate
\begin{multline}\label{2claim4}\begin{cases} (x_{2N-1}-x_{2N-2})^{-\Delta^-(\theta_2)}F_2(\boldsymbol{x})& \xrightarrow[\substack{x_{2N-1}\\\qquad\rightarrow x_{2N-2}}]{\boldsymbol{x}\in\Omega_0^{2N-1}}0 \\ (x_{2N-2}-x_{2N-3})^{-\Delta^-(\theta_1)}F_2(\boldsymbol{x})& \xrightarrow[\substack{x_{2N-2}\\\qquad\rightarrow x_{2N-3}}]{\boldsymbol{x}\in\Omega_0^{2N-1}}0\end{cases}\\
\Longrightarrow\quad F_2(\boldsymbol{x})=O\left(\begin{array}{l}\begin{aligned}(x_{2N-2}-&x_{2N-3})^{\Delta^+(\theta_1)} \\ \times & (x_{2N-1}-x_{2N-2})^{\Delta^+(\theta_2)}(x_{2N-1}-x_{2N-3})^{\Delta^+(\theta_1)}\end{aligned}\end{array}\right)\quad\text{as}\,\,\,\left\{\begin{array}{l}\text{$x_{2N-1}\rightarrow x_{2N-2}$} \\ \text{$x_{2N-2}\rightarrow x_{2N-3}$}\end{array}\right.,\end{multline}
derived in lemma \ref{closelem} below.  We may use (\ref{2claim4}) because the first two vanishing limits of the ``two-leg interval conditions"  in (\ref{toprove2}) imply that the limit (\ref{thelim}) with $x_{2N-1}\rightarrow x_{2N-2}$ vanishes thanks to (\ref{toprove}), and the estimate
\be\label{implies}(x_{2N-2}-x_{2N-3})^{-\Delta^-(\theta_1)}F_1(\boldsymbol{x})\xrightarrow[\substack{x_{2N-2}\\\qquad\rightarrow x_{2N-3}}]{\boldsymbol{x}\in\Omega_0^{2N}}0\quad\Longrightarrow\quad(x_{2N-2}-x_{2N-3})^{-\Delta^-(\theta_1)}F_2(\boldsymbol{x})\xrightarrow[\substack{x_{2N-2}\\\qquad\rightarrow x_{2N-3}}]{\boldsymbol{x}\in\Omega_0^{2N-1}}0,\ee
derived in lemma \ref{farlem} below, implies that the other limit in (\ref{2claim4}) with $x_{2N-2}\rightarrow x_{2N-3}$ vanishes.  Estimate (\ref{implies}) reveals the behavior of $F(\boldsymbol{x})$ as we collapse two non-adjacent intervals $(x_{2N-3},x_{2N-2})$ and $(x_{2N-1},x_{2N})$ \emph{simultaneously.}

Combined with lemmas \ref{kpzlem}--\ref{pdelem} of section \ref{oneinterval}, lemmas \ref{farlem} and \ref{closelem} are the last ingredients that we need to construct the two-point function $F_{2N-1}$ (\ref{2pt}) described in section \ref{Methodology} in order to prove lemma \ref{alltwoleglem}.  In the rest of this section, we present lemmas \ref{farlem} and \ref{closelem} with their proofs, and in section \ref{theproof}, we use them with lemmas \ref{kpzlem}--\ref{pdelem} to prove lemma \ref{alltwoleglem}.

But first, we motivate our analysis of simultaneous interval collapse by reconsidering the case in which we collapse just one interval.  In section \ref{oneinterval}, we found that we may write a solution $F$ of the system (\ref{hpde}, \ref{wardidh}) as
\be\label{whatwefound} F(\boldsymbol{x})=O((x_i-x_{i-1})^{\Delta^-(d_i)})+O((x_i-x_{i-1})^{\Delta^+(d_i)})\quad\text{as $x_i\rightarrow x_{i-1}$},\quad d_i=\begin{cases} \theta_1,& i\not\in\{\iota,\iota+1\} \\ h, & i\in\{\iota,\iota+1\}\end{cases}\ee
if we collapse just the interval $(x_{i-1},x_i)$.  Supposing that $i\not\in\{\iota,\iota+1\}$, we anticipate (\ref{whatwefound}) from the null-state PDE (\ref{hpde}) with $j=i$ after writing this PDE as $\mathcal{L}[F]=\mathcal{M}[F]$, with $\mathcal{L}$ and $\mathcal{M}$ respectively given by (\ref{i+11}) and (\ref{i+12}).  Indeed,
\be\label{Lscale}\text{$F(\boldsymbol{x}_\delta)\underset{\delta\downarrow0}\sim C\delta^{-p}$ for $C,p\in\mathbb{R}$}\quad\Longrightarrow\quad\begin{cases}\mathcal{L}[F](\boldsymbol{x}_\delta)=O(\delta^{-p-2})\\ \mathcal{M}[F](\boldsymbol{x}_\delta)=O(\delta^{-p-1})\end{cases}\quad\text{as $\delta\downarrow0$},\quad\text{($\boldsymbol{x}_\delta$ given by (\ref{xdelta})),}\ee  
if we make the natural assumption that for all $m\in\mathbb{Z}^+$, we have $\partial_\delta^mF(\boldsymbol{x}_\delta)=O(\delta^{-p-m})$ and $\partial_j^mF(\boldsymbol{x}_\delta)=O(\delta^{-p})$ if $j\neq i$.  If (\ref{Lscale}) is true, then we may approximate solutions of $\mathcal{L}[F]=\mathcal{M}[F]$ in the limit $\delta\downarrow0$ thus:
\be\label{nextbehavior}\mathcal{L}[F](\boldsymbol{x}_\delta)\underset{\delta\downarrow0}{\approx}0\quad\Longrightarrow\quad F(\boldsymbol{x}_\delta)\underset{\delta\downarrow0}{\sim}O(\delta^{\Delta^-(d_i)})+O(\delta^{\Delta^+(d_i)}).\ee
This decomposition for small $\delta>0$ (\ref{nextbehavior}) matches that of (\ref{whatwefound}) previously derived in the proofs of lemmas \ref{firstlimitlem} and \ref{secondlimitlem}.  (Almost identical arguments anticipate the same result if $i\in\{\iota,\iota+1\}$.)

In the case of lemma \ref{farlem} stated below, we simultaneously collapse two non-adjacent intervals $(x_{i-1},x_i)$ and $(x_{j-1},x_j)$ with respective lengths $\epsilon$ and $\delta$.  Reasoning similar to that of the previous paragraph suggests that a solution $F$ of the system (\ref{hpde}, \ref{wardidh}) behaves as
\be\label{thetwoterms}F(\boldsymbol{x}_{\delta,\epsilon})\underset{\delta,\epsilon\downarrow0}{\sim}[O(\delta^{\Delta^-(d_j)})+O(\delta^{\Delta^+(d_j)})][O(\epsilon^{\Delta^-(d_i)})+O(\epsilon^{\Delta^+(d_i)})],\ee
where $\boldsymbol{x}_{\delta,\epsilon}\in\Omega_0$ (\ref{xdeltaepsilon}) is such that $x_j=x_{j-1}+\delta$ and $x_i=x_{i-1}+\epsilon$.  To derive (\ref{thetwoterms}), we execute the analysis of lemmas \ref{firstlimitlem} and \ref{secondlimitlem} twice, once per collapsing interval, and in either iteration we use the Green function (\ref{Jgreen}) of the case with one interval collapse.  (The particular conditions of lemma \ref{farlem} imply that we should keep only the second term in either bracket of (\ref{thetwoterms}).  In so doing, we obtain estimate (\ref{farest}) below.)

In the case of lemmas \ref{firstcloselem} and \ref{closelem}, we simultaneously collapse the two adjacent intervals $(x_{\iota-2},x_{\iota-1})$ and $(x_{\iota-1},x_\iota)$ with respective lengths $\delta$ and $\varepsilon-\delta$.  Again, we wish to find the behavior of a solution $F$ of the system (\ref{hpde}, \ref{wardidh}) in this situation.  After writing the null-state PDE (\ref{hpde}) with $j=\iota-1>1$ as $\mathscr{L}[F]=\mathscr{M}[F]$, where
\bea\label{Lformula}\mathscr{L}[F](\boldsymbol{x}_{\delta,\varepsilon})&:=&\left[\frac{\kappa}{4}\partial_\delta^2+\frac{\partial_\delta}{\delta}+\left(\frac{1}{\delta}+\frac{1}{\varepsilon-\delta}\right)\partial_\varepsilon-\left(\frac{\theta_1}{\delta^2}+\frac{h}{(\varepsilon-\delta)^2}\right)\right]F(\boldsymbol{x}_{\delta,\varepsilon}),\\
\label{preMformula}\mathscr{M}[F](\boldsymbol{x}_{\delta,\varepsilon})&:=&\Bigg[\frac{\partial_{\iota-1}}{\delta}\,\,+\sum_{j\neq \iota-1,\iota}^M\left(\frac{\theta_1}{(x_j-x_\iota-\delta)^2}-\frac{\partial_j}{x_j-x_\iota-\delta}\right)\Bigg]F(\boldsymbol{x}_{\delta,\varepsilon}),\eea
where $\boldsymbol{x}_{\delta,\varepsilon}\in\Omega_0$ (\ref{boldx}) is such that $x_{\iota-1}=x_{\iota-2}+\delta$ and $x_\iota=x_{\iota-2}+\varepsilon$, and following the (non-rigorous) arguments of the previous paragraphs, we find that
\be\label{Lscale1}\text{$F(\boldsymbol{x}_{\delta,\varepsilon})\underset{\delta,\varepsilon\downarrow0}\sim C\delta^{-p}\varepsilon^{-p}$ for $C,p\in\mathbb{R}$}\quad\Longrightarrow\quad\begin{cases}\mathscr{L}[F](\boldsymbol{x}_{\delta,\varepsilon})=O(\delta^{-p-2}\varepsilon^{-p-2})\\ \mathscr{M}[F](\boldsymbol{x}_{\delta,\varepsilon})=O(\delta^{-p-1}\varepsilon^{-p})\end{cases}\quad\text{as $\delta,\varepsilon\downarrow0$}.\ee
Hence, because $\mathscr{L}[F]$ contains all of the terms that are ostensibly largest as $\delta,\varepsilon\downarrow0$, an appropriate solution of the PDE $\mathscr{L}[u]=0$ may presumably predict the behavior of $F(\boldsymbol{x}_{\delta,\varepsilon})$ in this limit.

Unlike its relative $\mathcal{L}$ (\ref{i+11}), $\mathscr{L}$ is a partial differential operator, so finding solutions of the PDE $\mathscr{L}[u]=0$ is more difficult than finding solutions of $\mathcal{L}[u]=0$.  Fortunately, $\mathscr{L}$ becomes separable after we change coordinates via (\ref{translate}) below.  Requiring that $u(\delta,\varepsilon)$ and $F(\boldsymbol{x}_{\delta,\varepsilon})$ exhibit the same asymptotic behavior as either $\delta\downarrow0$ or $\delta\uparrow\varepsilon$ with $\varepsilon>0$ fixed yields a unique, separable solution $u$ of this PDE, and we anticipate that its asymptotic behavior as $\delta,\varepsilon\downarrow0$ matches that of $F(\boldsymbol{x}_{\delta,\varepsilon})$ in the same limit.  This reasoning motivates the proof of lemma \ref{closelem} below just as (\ref{nextbehavior}) motivates the proofs of lemmas \ref{firstlimitlem} and \ref{secondlimitlem} above.

\begin{lem}\label{farlem}Suppose that $\kappa\in(0,8)$ and $M>3$, and for some $\iota\in\{2,3,\ldots,M\}$ and $j\in\{2,3,\ldots,M\}\setminus\{\iota-1,\iota,\iota+1\}$ (so the intervals $(x_{j-1},x_j)$ and $(x_{\iota-1},x_\iota)$ are neither adjacent nor identical), let
\be\label{xdeltaepsilon}\boldsymbol{x}_{\delta,\epsilon}:=\left.\begin{cases} (x_1,x_2,\ldots,x_{j-1},x_{j-1}+\delta,x_{j+1},\ldots,x_{\iota-1},x_{\iota-1}+\epsilon,x_{\iota+1},\ldots,x_M), & j<\iota-1 \\ (x_1,x_2,\ldots,x_{\iota-1},x_{\iota-1}+\epsilon,x_{\iota+1},\ldots,x_{j-1},x_{j-1}+\delta,x_{j+1},\ldots,x_M), & j>\iota+1\end{cases}\right\}\in\Omega_0^M.\ee
If $F:\Omega_0^M\rightarrow\mathbb{R}$ satisfies conditions \ref{nextcond1}--\ref{nextcond3} of lemma \ref{secondlimitlem} for $i\in\{\iota,j\}$, then for any compact $\mathcal{K}\subset\pi_{\iota,j}(\Omega_0^M)$,
\be\label{farest}\sup_{\mathcal{K}}|F(\boldsymbol{x}_{\delta,\epsilon})|=O(\delta^{\Delta^+(\theta_1)}\epsilon^{\Delta^+(h)})\quad\text{as $\delta,\epsilon\downarrow0$.}\ee
\end{lem}

\begin{proof} We let $\mathcal{K}$ be an arbitrary compact subset of $\pi_i(\Omega_0^M)$, we define 
\be\label{vars}\delta:=x_j-x_{j-1},\quad x:=x_{j-1},\quad\epsilon:=x_\iota-x_{\iota-1},\quad y:=x_\iota,\ee
and we restrict $\delta$ and $\epsilon$ to $0<\delta,\epsilon<b$ where $b$ is small enough to ensure that $x_\iota$ and $x_j$ are respectively less than $x_{\iota+1}$ and $x_{j+1}$ (if these coordinates exist).  Furthermore, we relabel the other $M-4$ coordinates of $\boldsymbol{x}$ in increasing order as $\xi_1$, $\xi_2,\ldots,\xi_{M-4}$, we let $\boldsymbol{\xi}:=(\xi_1,\xi_2,\ldots,\xi_{M-4})$, and we define
\be\label{Fdef}{\rm F}(\boldsymbol{\xi};x,\delta;y,\epsilon):=\begin{cases} F(\xi_1,\xi_2,\ldots,\xi_{j-2},x,x+\delta,\xi_{j-1},\ldots,\xi_{\iota-4},y-\epsilon,y,\xi_{\iota-3},\ldots,\xi_{M-4}),& j<\iota-1 \\ F(\xi_1,\xi_2,\ldots,\xi_{j-2},y-\epsilon,y,\xi_{j-1},\ldots,\xi_{\iota-4},x,x+\delta,\xi_{\iota-3},\ldots,\xi_{M-4}),& j>\iota+1\end{cases}.\ee

Although this proof resembles those of lemmas \ref{firstlimitlem} and \ref{secondlimitlem}, it has a key difference.  In the latter proofs, the Schauder estimate (\ref{Schauder}) follows from the strictly elliptic PDE that we constructed in steps \ref{step11}--\ref{step13} of the proof of lemma \ref{firstlimitlem}.  If we send $\epsilon\downarrow0$, then some of the PDE's coefficients, and thus the function $C$ in (\ref{Schauder}), blow up, destroying our estimates.

Therefore, before we proceed with the proof of lemma \ref{farlem}, we construct a new, strictly elliptic PDE that avoids this problem.  This PDE has the following features:
\begin{itemize}
\item Derivatives of ${\rm F}$ are only with respect to $x$ and the coordinates of $\boldsymbol{\xi}$,
\item $y$, $\delta$, and $\epsilon$ appear as parameters,
\item the coefficients  do not vanish or blow up as $\delta,\epsilon\downarrow0$, and
\item the PDE is strictly elliptic in a given compactly embedded subset of $\pi_{\iota,j}(\Omega_0)$.
\end{itemize}
The construction of this PDE is nearly identical to that of a similar PDE used in the first half of the proof of lemma \red{12} in \cite{florkleb}.  For this reason, we only sketch the steps,  leaving the algebraic details to the reader.
\begin{enumerate}
\item\label{step0}  To begin, we express the null-state PDE centered on $x_k$ (\ref{hpde}) with $k\not\in\{\iota-1,\iota\}$ in terms of the variables in (\ref{vars}) and the coordinates of $\boldsymbol{\xi}$.  The null-state PDE (\ref{hpde}) centered on $x_{j-1}$ becomes
\begin{multline}\label{Iith}\Bigg[\frac{\kappa}{4}(\partial_x-\partial_\delta)^2+\frac{\partial_\delta}{\delta}+\sum_l\left(\frac{\partial_l}{\xi_l-x}-\frac{\theta_1}{(\xi_l-x)^2}\right)\\
+\frac{\partial_y}{y-x}-\frac{\epsilon\partial_\epsilon}{(y-x)(y-\epsilon-x)}-\frac{\theta_1}{(y-x)^2}-\frac{h}{(y-\epsilon-x)^2}\Bigg]{\rm F}(\boldsymbol{\xi};x,\delta;y,\epsilon)=0,\end{multline}
the null-state PDE (\ref{hpde}) centered on $x_j$ becomes
\newpage
\begin{multline}\label{Ii+1th}\Bigg[\frac{\kappa}{4}\partial_\delta^2-\frac{(\partial_x-\partial_\delta)}{\delta}+\sum_k\left(\frac{\partial_l}{\xi_l-x-\delta}-\frac{\theta_1}{(\xi_l-x-\delta)^2}\right)\\
+\frac{\partial_y}{y-x}-\frac{\epsilon\partial_\epsilon}{(y-x)(y-\epsilon-x)}-\frac{\theta_1}{(y-x)^2}-\frac{h}{(y-\epsilon-x)^2}\Bigg]{\rm F}(\boldsymbol{\xi};x,\delta;y,\epsilon)=0,\end{multline}
and the null-state PDE (\ref{hpde}) centered on the coordinate of $\boldsymbol{x}$ that we now call $\xi_k$ becomes
\begin{multline}\label{Inotii+1}\Bigg[\frac{\kappa}{4}\partial_k^2+\sum_{l\neq k}\left(\frac{\partial_l}{\xi_l-\xi_k}-\frac{\theta_1}{(\xi_l-\xi_k)^2}\right)\\
\begin{aligned}&+\frac{\partial_x}{x-\xi_k}-\frac{\delta\partial_\delta}{(x-\xi_k)(x+\delta-\xi_k)}-\frac{\theta_1}{(x-\xi_k)^2}-\frac{\theta_1}{(x+\delta-\xi_k)^2}\\
&+\frac{\partial_y}{y-\xi_k}-\frac{\epsilon\partial_\epsilon}{(y-\xi_k)(y-\epsilon-\xi_k)}-\frac{\theta_1}{(y-\xi_k)^2}-\frac{h}{(y-\epsilon-\xi_k)^2}\Bigg]{\rm F}(\boldsymbol{\xi};x,\delta;y,\epsilon)=0.\end{aligned}\end{multline}
Also, the three conformal Ward identities (\ref{wardidh}) become
\begin{multline}\label{Iw1}\begin{aligned}&\bigg[\sideset{}{_k}\sum\partial_k+\partial_x+\partial_y\bigg]{\rm F}(\boldsymbol{\xi};x,\delta;y,\epsilon)=0,\\
&\bigg[\sideset{}{_k}\sum\xi_k\partial_k+x\partial_x+\delta\partial_\delta+y\partial_y+\epsilon\partial_\epsilon+(M-1)\theta_1+h\bigg]{\rm F}(\boldsymbol{\xi};x,\delta;y,\epsilon)=0,\\
&\bigg[\sideset{}{_k}\sum(\xi_k^2\partial_k+2\theta_1\xi_k)+x^2\partial_x+2\theta_1x+y^2\partial_y+2\theta_1y+(2x+\delta)\delta\partial_\delta\end{aligned}\\
+2\theta_1(x+\delta)+(2y-\epsilon)\epsilon\partial_\epsilon+2h(y-\epsilon)\bigg]{\rm F}(\boldsymbol{\xi};x,\delta;y,\epsilon)=0.\end{multline}
\item\label{step1}Next, we subtract (\ref{Ii+1th}) from (\ref{Iith}) and multiply the result by $\delta$ to find a PDE with principal part
\be\label{Idiffpde}\bigg[\frac{\kappa}{4}\delta\partial_x^2-\frac{\kappa}{2}\partial_x\delta\partial_\delta\bigg]{\rm F}(\boldsymbol{\xi};x,\delta;y,\epsilon)\ee
and with coefficients that neither vanish nor blow up as $\epsilon\downarrow0$ or $\delta\downarrow0$.
\item\label{step2}Next, we use (\ref{Iw1}) to solve for $\delta\partial_\delta{\rm F}$ strictly in terms of ${\rm F}$ and its derivatives with respect to $x$ and the coordinates of $\boldsymbol{\xi}$.  We insert the result into the principal part (\ref{Idiffpde}) of the PDE from step \ref{step1} to generate a PDE whose principal part only contains $\partial_x^2{\rm F}$ and the mixed partial derivatives $\partial_x\partial_k{\rm F}$ for $k\in\{1,2,\ldots,M-4\}$.  Again, none of the coefficients in this PDE vanish or grow without bound as $\delta,\epsilon\downarrow0$.  We let $a(x,\delta;y,\epsilon)$ be the coefficient of $\partial_x^2{\rm F}(\boldsymbol{\xi};x,\delta;y,\epsilon)$ in this PDE.
\item\label{step3} Next, we take a linear combination of the PDE that we constructed in step \ref{step2}, with coefficient   $a(x,\delta;y,\epsilon)^{-1}$, and the $M-4$ null-state PDEs in (\ref{Inotii+1}), each with the same coefficient $4c/\kappa$ with $c>0$, to find a new PDE whose principal part only contains $\partial_x^2{\rm F},$ $\partial_k^2{\rm F}$, and the mixed partial derivatives $\partial_x\partial_k{\rm F}$ with $k\in\{1,2,\ldots,M-4\}$.  Again, none of the coefficients in this PDE grow without bound as $\epsilon\downarrow0$ or $\delta\downarrow0$.
\item\label{step4} Finally, we use (\ref{Iw1}) to replace the first derivatives $\delta\partial_\delta{\rm F}$, $\epsilon\partial_\epsilon{\rm F}$, and $\partial_y{\rm F}$ in the PDE that we constructed in step \ref{step3} with linear combinations of first derivatives of ${\rm F}$ with respect to either $x$ or the coordinates of $\boldsymbol{\xi}$.  This produces a final, complicated PDE for which $x$ and the coordinates of $\boldsymbol{\xi}$ are independent variables while $y$, $\delta$, and $\epsilon$ are parameters.  The coefficients of this final PDE do not vanish or grow without bound as $\delta,\epsilon\downarrow0$.
\end{enumerate}

So far, the PDE that we have constructed manifestly satisfies the first three bullet points presented above for any positive $c$.  Now we argue that in the bounded open set $ U_0\subset\subset\pi_{\iota,j}(\Omega_0)$, there exists a choice for $c$ such that this PDE is strictly elliptic in $U_0$.  The coefficient matrix for its principal part  takes the form
\renewcommand{\kbldelim}{(}
\renewcommand{\kbrdelim}{)}
\be\label{coefmatrix}\kbordermatrix{
    & \xi_1 & \xi_2 & \hdots & \xi_{M-5} & \xi_{M-4} & x\\
    \xi_1 & c & 0 & \hdots & 0 & 0 & a_{1} \\
    \xi_2 & 0 & c & &  & 0 & a_{2} \\
    \vdots & \vdots &  & \ddots &  & \vdots & \vdots  \\
    \xi_{M-5} & 0 &  &  & c & 0 & a_{M-5} \\
    \xi_{M-4} & 0 & 0 & \hdots & 0 & c & a_{M-4} \\
   x & a_{1} & a_{2} & \hdots & a_{M-5} & a_{M-4} & 1
  },
\ee
where $a_k$ is half of the coefficient of $\partial_x\partial_k{\rm F}$.  According to Sylvester's criterion (Thm.\ 7.5.2 of \cite{horn}), this matrix is positive definite if all of its leading principal minors are positive.  Because $c$ is positive, the determinants of the first $M-4$  leading principal minors are evidently positive.  We find the remaining principal minor, really the determinant of (\ref{coefmatrix}), by using the formula
\be\label{detformula}M=\left(\begin{matrix} A & B \\ C & D\end{matrix}\right)\quad\Longrightarrow\quad\det M=\det A\det(D-CA^{-1}B),\ee
where $A$ and $D$ are square blocks of the matrix $M$, and where $B$ and $C$ are blocks that fill the part of $M$ above $D$ and beneath $A$ respectively.  Using this formula, we find that the $(M-3)$rd  leading principal minor of (\ref{coefmatrix}) is 
\be\label{firstdet}c^{M-4}\left(1-\frac{|a_x|^2}{c}\right),\ee
where $a_x$ is the vector in $\mathbb{R}^{M-4}$ whose $k$th entry is $a_k$.  Because each entry of $a_x$ is bounded on $U_0$, by choosing $c$ large enough, we ensure that (\ref{firstdet}) is greater than, say, one at all points in $ U_0$ and for all $0<\delta,\epsilon<b$.  Thus, the coefficient matrix (\ref{coefmatrix}) is positive definite.  Furthermore, because all components of $a_x$ are bounded on $ U_0$, the eigenvalues of the matrix (\ref{coefmatrix}) (which is Hermitian and therefore diagonalizable) are bounded on $U_0$ too.  This fact together with the fact that the product of these eigenvalues, equaling (\ref{firstdet}), is  strictly positive on $ U_0$, imply that all of these eigenvalues are bounded away from zero over this set for all $0<\delta,\epsilon<b$.  Thus, the constructed PDE is strictly elliptic in $ U_0$ for all $0<\delta,\epsilon<b$.

The existence of this PDE implies Schauder estimates.  We choose open sets $ U_0,$ $ U_1,\ldots, U_{4m}$, where $m:=\lceil q\rceil$ and $q:=p+\Delta^-(\theta_1)$ ($p$ is given by condition \ref{nextcond1} of lemma \ref{secondlimitlem}, and we choose $p$ big enough so $m>2$), such that 
\be\mathcal{K}\subset\subset U_{4m}\subset\subset U_{4m-1}\subset\subset\ldots\subset\subset U_0\subset\subset\pi_{\iota,j}(\Omega_0^M),\ee
and we let $d_n:=\text{dist}(\partial U_{n+1},\partial U_n)$.  With these choices, the Schauder interior estimate gives (Cor.\ 6.3 of \cite{giltru})
\be\label{ISchauder}d_{n+1}\sup_{ U_{n+1}}|\partial^\varpi{\rm F}(\boldsymbol{\xi};x,\delta;y,\epsilon)|\leq C_n\sup_{ U_n}|{\rm F}(\boldsymbol{\xi};x,\delta;y,\epsilon)|\ee
for all $0<\delta,\epsilon<b$ and $n\in\{0,1,\ldots,4m-1\}$, where $C_n$ is some constant and $\varpi$ is a multi-index of $x$ and the coordinates of $\boldsymbol{\xi}$ with length $|\varpi|\leq2$.  Furthermore, step \ref{step4} of the above construction implies that $\varpi$ may include the coordinate $y$ and the derivatives $\delta\partial_\delta$ and $\epsilon\partial_\epsilon$ too.  Overall, we have
\be\label{multiindex}\partial^\varpi\in\left\{\begin{array}{c}\partial_j,\quad \partial_x,\quad \partial_y,\quad \delta\partial_\delta,\quad\epsilon\partial_\epsilon, \quad\partial_j^2,\quad\partial_x^2,\quad\partial_y^2,\quad(\delta\partial_\delta)^2,\quad(\epsilon\partial_\epsilon)^2,\quad \partial_j\partial_k, \\ \partial_j\partial_x,\quad\partial_j\partial_y,\quad\partial_j\delta\partial_\delta,\quad\partial_j\epsilon\partial_\epsilon,\quad\partial_x\partial_y,\quad\partial_x\delta\partial_\delta,\quad\partial_x\epsilon\partial_\epsilon\quad\partial_y\delta\partial_\delta,\quad \partial_y\epsilon\partial_\epsilon,\quad\delta\partial_\delta\epsilon\partial_\epsilon \end{array}\right\}.\ee
Estimate (\ref{ISchauder}) is the main ingredient for our principal goal, deriving the bound (\ref{farest}).  

The proof of the bound (\ref{farest})  proceeds similarly to the proofs of lemmas \ref{firstlimitlem} and \ref{secondlimitlem} above and very similarly to the proof of lemma \red{12} in \cite{florkleb}.  First, we define  
\be\label{Idefn}I(\boldsymbol{\xi};x,\delta;y,\epsilon):=\delta^{-\Delta^-(\theta_1)}\epsilon^{-\Delta^-(h)}{\rm F}(\boldsymbol{\xi};x,\delta;y,\epsilon),\ee
and we show that $I(\boldsymbol{\xi};x,\delta;y,\epsilon)=O(1)$ as $\delta,\epsilon\downarrow0$.  Expressed in terms of $I$, $\boldsymbol{\xi}$, $x$, $\delta$, $y$, and $\epsilon$, the integral equation (\ref{tildeH}) becomes
\be\label{I} I(\boldsymbol{\xi};x,\delta;y,\epsilon)=I(\boldsymbol{\xi};x,b;y,\epsilon)-\frac{\kappa}{4}J(\delta,b)\partial_b I(\boldsymbol{\xi};x,b;y,\epsilon)+\sideset{}{_\delta^b}\int J(\delta,\eta)\mathcal{N}[I](\boldsymbol{\xi};x,\eta;y,\epsilon)\,{\rm d}\eta,\ee
for all $0<\delta<b$, where $J$ is defined in (\ref{Jgreen}) and $\mathcal{N}$ is the differential operator
\begin{multline}\label{N}\mathcal{N}[I](\boldsymbol{\xi};x,\eta;y,\epsilon):=\Bigg[\frac{\partial_x}{\eta}-\sum_{k=1}^{M-4}\left(\frac{\partial_k}{\xi_k-x-\eta}-\frac{\theta_1}{(\xi_k-x-\eta)^2}\right)-\frac{\partial_y}{y-x-\eta}+\frac{\theta_1}{(y-x-\eta)^2}\\
+\frac{\epsilon\partial_{\epsilon}}{(y-x-\eta)(y+\epsilon-x-\eta)}+\frac{h}{(y+\epsilon-x-\eta)^2}+\frac{\Delta^+(h)}{(y-x-\eta)(y+\epsilon-x-\eta)}\Bigg]I(\boldsymbol{\xi};x,\eta;y,\epsilon).\end{multline}
Taking the supremum of (\ref{I}) over $ U_n$ gives
\begin{multline}\label{estI2}\sup_{ U_n}|I(\boldsymbol{\xi};x,\delta;y,\epsilon)|\leq \sup_{ U_n}|I(\boldsymbol{\xi};x,b;y,\epsilon)|\\
+\,\bindnasrepma(\theta_1)^{-1}\sup_{ U_n}|b\hspace{.03cm}\partial_b I(\boldsymbol{\xi};x,b;y,\epsilon)|+\frac{4/\kappa}{{\bindnasrepma(\theta_1)}}\sideset{}{_\delta^b}\int\sup_{ U_n}|\eta\,\mathcal{N}[I](\boldsymbol{\xi};x,\eta;y,\epsilon)|\,{\rm d}\eta\end{multline}
for all $0<\delta,\epsilon<b$.  Thanks to condition \ref{nextcond1} of lemma \ref{secondlimitlem}, we may use (\ref{ISchauder}) with $n=0$ to estimate the integrand of (\ref{estI2}) with $n=1$ as $O(\eta^{-q}\epsilon^{-q})$ when $\eta\downarrow0$ or $\epsilon\downarrow0$.  After inserting this estimate into (\ref{estI2}), we find that
\be\sup_{ U_1}|I(\boldsymbol{\xi};x,\delta;y,\epsilon)|=O(\delta^{-q+1}\epsilon^{-q}).\ee
Just as we did in the proofs of lemmas \ref{firstlimitlem} and \ref{secondlimitlem}, we repeat this process $m-1$ more times, using the subsets $U_m\subset\subset U_{m-1}\subset\subset\ldots\subset\subset U_0$ until we finally arrive with
\be\sup_{ U_m}|I(\boldsymbol{\xi};x,\delta;y,\epsilon)|=O(\epsilon^{-q})\quad\text{as $\delta,\epsilon\downarrow0$.}\ee
Now to decrease the power on $\epsilon$, we switch $(x,\delta,\theta_1;y,\epsilon,h)$ to $(y,\epsilon,h;x,\delta,\theta_1)$ and replace $\mathcal{N}$ with $\mathcal{Q}$ (\ref{Q}) (but still use the variables of (\ref{vars}) and $\boldsymbol{\xi}$) in (\ref{estI2}), and we continue the previous steps for the $\epsilon$ variable, using the subsets $ U_{2m}\subset\subset U_{2m-1}\subset\subset\ldots\subset\subset U_m$.  We ultimately find
\be\label{almostlastest1}\sup_{ U_{2m}}|I(\boldsymbol{\xi};x,\delta;y,\epsilon)|=O(1)\quad\text{as $\delta,\epsilon\downarrow0$.}\ee

Now, to finish the derivation of the bound (\ref{farest}), we define
\be\label{Ldefn}L(\boldsymbol{\xi};x,\delta;y,\epsilon):=\delta^{-\Delta^+(\theta_1)}\epsilon^{-\Delta^+(h)}{\rm F}(\boldsymbol{\xi};x,\delta;y,\epsilon)=\delta^{-\bindnasrepma(\theta_1)}\epsilon^{-\bindnasrepma(h)}I(\boldsymbol{\xi};x,\delta;y,\epsilon),\ee
and we show that $L(\boldsymbol{\xi};x,\delta;y,\epsilon)=O(1)$ as $\delta,\epsilon\downarrow0$ as a consequence of condition \ref{nextcond3} of lemma \ref{secondlimitlem} for $i\in\{\iota,j\}$.   Now if $i=j$, then this condition implies that $L$ satisfies the integral equation (\ref{Kpropagator}) for all $0<\delta,\epsilon<b$, which reads
\be\label{L}L(\boldsymbol{\xi};x,\delta;y,\epsilon)=L(\boldsymbol{\xi};x,b;y,\epsilon)-\frac{4}{\kappa}\sideset{}{_\delta^b}\int\frac{1}{\beta}\sideset{}{_0^\beta}\int\left(\frac{\eta}{\beta}\right)^{\bindnasrepma(\theta_1)}\eta\,\mathcal{N}[L](\boldsymbol{\xi};x,\eta;y,\epsilon)\,{\rm d}\eta\,{\rm d}\beta,\ee
after we express it in terms of the function $L$ and variables $\boldsymbol{\xi}$, $x$, $\delta$, $y$, and $\epsilon$ used in this proof.  Taking the supremum of (\ref{L}) over the open set $ U_n$ immediately gives
\be\label{Lpropagator2nd} \sup_{ U_n}|L(\boldsymbol{\xi};x,\delta;y,\epsilon)|\leq\sup_{ U_n}|L(\boldsymbol{\xi};x,b;y,\epsilon)|+\frac{4}{\kappa}\sideset{}{_\delta^b}\int\frac{1}{\beta}\sideset{}{_0^\beta}\int\left(\frac{\eta}{\beta}\right)^{\bindnasrepma(\theta_1)}\sup_{ U_n}|\eta\,\mathcal{N}[L](\boldsymbol{\xi};x,\eta;y,\epsilon)|\,{\rm d}\eta\,{\rm d}\beta.\ee
According to (\ref{almostlastest1}), the supremum in the integrand, with $\epsilon>0$ fixed, is bounded on $0<\eta<b$.
Using the same methods as in the proof of lemma \ref{secondlimitlem} with a continued sequence of nested open sets $ U_{3m}\subset\subset U_{3m-1}\subset\subset\ldots\subset\subset U_{2m}$, we repeatedly integrate (\ref{Lpropagator2nd}) to improve the bound for $L(\boldsymbol{\xi};x,\delta;y,\epsilon)$ as $\epsilon\downarrow0$, eventually finding that
\be\label{almostalmostlastest}\sup_{ U_{3m}}|L(\boldsymbol{\xi};x,\delta;y,\epsilon)|=O(\epsilon^{-\bindnasrepma(h)})\quad\text{as $\delta,\epsilon\downarrow0$.}\ee
Finally, after switching $(x,\delta,\theta_1;y,\epsilon,h)$ to $(y,\epsilon,h;x,\delta,\theta_1)$ and replacing $\mathcal{N}$ with $\mathcal{Q}$ (\ref{Q}) (using, again, the variables of (\ref{vars}) and $\boldsymbol{\xi}$) in (\ref{Lpropagator2nd}) and repeating the previous steps for the $\epsilon$ variable with the sets $\mathcal{K}\subset\subset U_{4m}\subset\subset U_{4m-1}\subset\subset\ldots\subset\subset U_{3m}$, we ultimately find that
\be\label{almostlastest}\sup_{\mathcal{K}}|L(\boldsymbol{\xi};x,\delta;y,\epsilon)|=O(1)\quad\text{as $\delta,\epsilon\downarrow0$.}\ee
After recalling the relation (\ref{Ldefn}) between $L$ and ${\rm F}$ (\ref{Fdef}), we see that (\ref{almostlastest}) is the estimate (\ref{farest}).
\end{proof}

Before we prove the estimate of lemma \ref{closelem} for collapsing adjacent intervals, we need the following lemma \ref{firstcloselem}. This lemma gives a power-law for this estimate up to an undetermined constant $p_0$, which we then optimize in lemma \ref{closelem}.

\begin{lem}\label{firstcloselem}Suppose that $\kappa\in(0,8)$ and $M>2$, and for some $\iota\in\{3,4,\ldots,M\}$, let 
\be\label{boldx}\boldsymbol{x}_{\delta,\varepsilon}:=(x_1,x_2,\ldots,x_{\iota-2},x_{\iota-2}+\delta,x_{\iota-2}+\varepsilon,x_{\iota+1},\ldots,x_M)\in\Omega_0^M.\ee
If $F:\Omega_0^M\rightarrow\mathbb{R}$ satisfies conditions \ref{nextcond1}--\ref{nextcond3} of lemma \ref{secondlimitlem} for $i\in\{\iota-2,\iota-1\}$, 
then for any compact $\mathcal{K}\subset\pi_{\iota-1,\iota}(\Omega_0^M)$ there is a $p_0\in\mathbb{R}$ such that
\be\label{initialest}\sup_\mathcal{K}|F(\boldsymbol{x}_{\delta,\varepsilon})|=O(\delta^{\Delta^+(\theta_1)}\varepsilon^{-p_0}(\varepsilon-\delta)^{\Delta^+(h)})\quad\text{as $\delta,\varepsilon\downarrow0$}.\ee
\end{lem}

\begin{proof} We let $\mathcal{K}$ be an arbitrary compact subset of $\pi_i(\Omega_0^M)$.  In order to prove (\ref{initialest}), it suffices to prove that there exist positive numbers $c,$ $p_1$, and $p_2$ such that
\be\label{initialestmod}
0<\delta<\varepsilon<1\quad\Longrightarrow\quad\sup_\mathcal{K}|F(\boldsymbol{x}_{\delta,\varepsilon})|\leq c\begin{cases}\delta^{\Delta^+(\theta_1)}\varepsilon^{-p_1},& 0<\delta<\varepsilon/2\\
\varepsilon^{-p_2}(\varepsilon-\delta)^{\Delta^+(h)},&\varepsilon/2<\delta<\varepsilon\end{cases}.
\ee
Indeed, because $\Delta^+(h)>0$ and $\Delta^+(\theta_1)>0$, (\ref{initialestmod}) implies that
\be\label{initialestmodmod}0<\delta<\varepsilon<1\quad\Longrightarrow\quad\sup_\mathcal{K}|F(\boldsymbol{x}_{\delta,\varepsilon})|\leq c\begin{cases}2^{\Delta^+(h)}\delta^{\Delta^+(\theta_1)}\varepsilon^{-p_1-\Delta^+(h)}(\varepsilon-\delta)^{\Delta^+(h)},& 0<\delta<\varepsilon/2,\\
2^{\Delta^+(\theta_1)}\delta^{\Delta^+(\theta_1)}\varepsilon^{-p_2-\Delta^+(\theta_1)}(\varepsilon-\delta)^{\Delta^+(h)},&\varepsilon/2<\delta<\varepsilon\end{cases},\ee
which in turn implies (\ref{initialest}) with $p_0=\max\{p_1+\Delta^+(h),p_2+\Delta^+(\theta_1)\}$.  (We note that the upper bound $\varepsilon<1$ is somewhat arbitrary and perhaps even unnecessary if $\iota<2N$, as $\varepsilon$ is already bounded above by $x_{\iota+1}-x_{\iota-2}$.)

To prove the first estimate of (\ref{initialestmod}), we repeat the steps in case \ref{firstlimitlempart1} of lemmas \ref{firstlimitlem} and \ref{secondlimitlem} with $i=\iota-1$, retaining  $x$ but using $x+\varepsilon$ in place of $y$, and making the following adjustments.  First, we set $b=\varepsilon/2$ in (\ref{tildeH}) in order to bound $\delta$ away from $\varepsilon$.  Second, we compactly embed the open subsets $U_n$ (\ref{compactembedd}) in $\pi_{\iota-1,\iota}(\Omega_0^M)$ instead of $\pi_{\iota-1}(\Omega_0^M)$ so $\varepsilon$ may approach zero.  Thus, (\ref{estH}) becomes
\begin{multline}\label{estHmod}\sup_{U_n}|H(\boldsymbol{\xi};x,\delta;x+\varepsilon)|\leq\sup_{ U_n}|H(\boldsymbol{\xi};x,\varepsilon/2;x+\varepsilon)|\\
+\,\bindnasrepma(\theta_1)^{-1}\sup_{ U_n}|(\varepsilon/2)\partial_\delta H(\boldsymbol{\xi};x,\varepsilon/2;x+\varepsilon)|+\frac{4/\kappa}{\bindnasrepma(\theta_1)}\sideset{}{_\delta^{\varepsilon/2}}\int\sup_{ U_n}|\eta\mathcal{M}[H](\boldsymbol{\xi};x,\eta;x+\varepsilon)|\,{\rm d}\eta.\end{multline}
(We recall the definition of $H$ from (\ref{tilde}).)  Third, we note that the strictly elliptic PDE, constructed in steps \ref{step11}--\ref{step13} of the proof of lemma \ref{firstlimitlem}, only has derivatives with respect to the coordinates of $\boldsymbol{\xi}$, has $x$, $\delta$, and $\varepsilon=y-x$ as parameters, and has no coefficients that vanish or blow up as $\delta\downarrow0$ \emph{or} as $\varepsilon\downarrow0$.  Thus, the Schauder interior estimate (\ref{Schauder}),
\be\label{Schaudermod}d_{n+1}\sup_{ U_{n+1}}|\partial^\varpi H(\boldsymbol{\xi};x,\delta;x+\varepsilon)|\leq C(R_n)\sup_{ U_n}|H(\boldsymbol{\xi};x,\delta;x+\varepsilon)|,\ee
holds uniformly for all $0<\delta<\varepsilon/2$ \emph{and} all sufficiently small $\varepsilon>0$.  (We define $d_n$, $R_n$, $C$, and the multi-index $\varpi$ in the discussion surrounding (\ref{Schauder}, \ref{multiindex1}).)  Now, with $U_1\subset\subset U_0$, $0<\delta<\varepsilon/2$, and $\varepsilon<1$, the power-law bound (\ref{powerlaw}) gives
\bea\label{prelimest1}\sup_{U_1}|H(\boldsymbol{\xi};x,\delta;x+\varepsilon)|&\leq&\sup_{U_0}|H(\boldsymbol{\xi};x,\delta;x+\varepsilon)|\\
\label{prelimest2}&\leq&C_0\delta^{-q}\varepsilon^{-p}(\varepsilon-\delta)^{-p}\leq2^pC_0\delta^{-q}\varepsilon^{-2p}
\eea
for some positive constants $C_0$, $p$, and $q:=p+\Delta^-(\theta_1)$.  After using (\ref{Schaudermod}) with $n=0$ and (\ref{prelimest1}, \ref{prelimest2}) to estimate each term in (\ref{estHmod}) with $n=1$ and integrating the result, we find that for some positive constants $c_0$, $C_1$, and $p'$,
\bea\sup_{U_1}|H(\boldsymbol{\xi};x,\delta;x+\varepsilon)|&\leq&2^{p+q}C_0\varepsilon^{-2p-q}\left[1+\frac{C(R_0)}{\bindnasrepma(\theta_1)}\right]+\frac{4/\kappa}{\bindnasrepma(\theta_1)}\sideset{}{_\delta^{\varepsilon/2}}\int c_0\eta^{-q}\varepsilon^{-2p}\,{\rm d}\eta\nonumber\\ 
\label{estHmodest}&\leq&C_1\delta^{-q+1}\varepsilon^{-p'}.\eea
The result (\ref{estHmodest}) is analogous to the previous result (\ref{thefirstest}) in the proof of lemma \ref{firstlimitlem}, except that the former bound has an explicit $\varepsilon$-dependence.  This difference arises because the supremum of the former does not involve the $\varepsilon=y-x$ variable, while the supremum of the latter does (or, more exactly, does involve $y$).  That the right side of (\ref{estHmodest}) is a power-law in $\varepsilon$ is important for our purpose, but the exact value of the power $p'$ is not.  Now, proceeding as in the rest of case \ref{firstlimitlempart1} of the proof of lemma \ref{firstlimitlem} and case \ref{secondlimitlempart1} of the proof of lemma \ref{secondlimitlem}, we repeat these estimates to ultimately obtain
\be\label{finalestHmodest}\sup_{\mathcal{K}}|E(\boldsymbol{\xi};x,\delta;x+\varepsilon)|\leq c\varepsilon^{-p_1}\ee
for some positive constants $c$ and $p_1$ and $0<\delta<\varepsilon/2$.  (We recall the definition of $E$ from (\ref{Kdefn}).)  The result (\ref{finalestHmodest}) is analogous to (\ref{lastSchauderK}) in the proof of lemma \ref{secondlimitlem}, except that, again, the former bound has an explicit $\varepsilon$-dependence while the latter does not.  Recalling the definition of $E$ (\ref{Kdefn}), we see that (\ref{finalestHmodest}) is identical to the first inequality of (\ref{initialestmod}).

The same argument that produces the upper estimate of (\ref{initialestmod}) produces the lower estimate of (\ref{initialestmod}) after we replace $\delta$ by $\delta':=\varepsilon-\delta$ and set $i=\iota$.  Thus, we use parts \ref{firstlimitlempart3} and \ref{secondlimitlempart2} of the proofs of lemmas \ref{firstlimitlem} and \ref{secondlimitlem} respectively, also with $\delta$ replaced by $\delta'$.  In this argument for the lower estimate, we reuse the elliptic PDE used for finding the upper estimate to obtain Schauder interior estimates that hold for all $0<\delta'<\varepsilon/2$ \emph{and} all sufficiently small $\varepsilon>0$.
\end{proof}

In the next lemma, we show that $p_0\geq-\Delta^+(h)$ in (\ref{initialest}).  This lemma is the main result of this article that allows us to prove lemma \ref{alltwoleglem}.

\begin{lem}\label{closelem}Suppose that $\kappa\in(0,8)$ and $M>2$, and for some $\iota\in\{3,4,\ldots,M\}$, let 
$\boldsymbol{x}_{\delta,\varepsilon}$ be defined as in (\ref{boldx}).
If $F:\Omega_0^M\rightarrow\mathbb{R}$ satisfies conditions \ref{nextcond1}--\ref{nextcond3} of lemma \ref{secondlimitlem} for $i\in\{\iota-2,\iota-1\}$, 
then for any compact $\mathcal{K}\subset\pi_{\iota-1,\iota}(\Omega_0^M)$, 
\be\label{closeest}\sup_{\mathcal{K}}|F(\boldsymbol{x}_{\delta,\varepsilon})|=O(\delta^{\Delta^+(\theta_1)}\varepsilon^{\Delta^+(h)}(\varepsilon-\delta)^{\Delta^+(h)})\quad\text{as $\delta,\varepsilon\downarrow0$.}\ee\end{lem}

\begin{proof} We let $\varepsilon:=x_\iota-x_{\iota-2}$, $\delta:=x_{\iota-1}-x_{\iota-2}<\varepsilon,$ $x:=x_{\iota-2}$ (figure \ref{Geometry}), we relabel the other $M-3$ coordinates of $\boldsymbol{x}$ in increasing order as $\xi_1$, $\xi_2,\ldots,\xi_{M-3}$, and we let $\boldsymbol{\xi}=(\xi_1,\xi_2,\ldots,\xi_{M-3})$.  With this new notation, we define
\be\label{Fdef2}{\rm F}(\boldsymbol{\xi};x,\delta,\varepsilon):=F(\xi_1,\xi_2,\ldots,\xi_{\iota-3},x,x+\delta,x+\varepsilon,\xi_{\iota-2},\ldots,\xi_{M-3}).\ee
Also, we choose an arbitrary compact subset $\mathcal{K}\subset\pi_{\iota-1,\iota}(\Omega_0^M)$.  Our goal is to prove that estimate (\ref{initialest}) of lemma \ref{firstcloselem} is true for any $p_0\geq-\Delta^+(h)$.

Most of the proof involves analysis of the null-state PDE (\ref{hpde}) centered on $x_{\iota-1}$.  In terms of the variables specified above, we write this PDE in the form $\mathscr{L}[{\rm F}]=\mathscr{M}[{\rm F}]$, where $\mathscr{L}[{\rm F}]$ is given by (\ref{Lformula}) and
\be
\label{Mformula}\mathscr{M}[{\rm F}](\boldsymbol{\xi};x,\delta,\varepsilon):=\Bigg[\frac{\partial_x}{\delta}+\sum_{j=1}^{M-3}\left(\frac{\theta_1}{(\xi_j-x-\delta)^2}-\frac{\partial_j}{\xi_j-x-\delta}\right)\Bigg]{\rm F}(\boldsymbol{\xi};x,\delta,\varepsilon).\ee
Here, we have collected all of the terms that are apparently largest when $\delta$ or $\varepsilon$ are very small into the differential operator $\mathscr{L}$ and collected the remaining terms into $\mathscr{M}$.  After introducing the change of coordinates (figure \ref{Geometry})
\be\label{translate}(\delta,\varepsilon)\quad\longrightarrow\quad(\rho:=\delta/\varepsilon,\varepsilon),\ee
the former $\mathscr{L}$  becomes a separable differential operator.  Indeed, in terms of these new coordinates, the null-state PDE centered on $x_{\iota-1}$ goes from $\mathscr{L}[{\rm F}]=\mathscr{M}[{\rm F}]$ to $\mathscr{P}[\mathsf{F}]=\mathscr{N}[\mathsf{F}]$, where
\begin{gather}\label{Nformula}\mathsf{F}(\boldsymbol{\xi};x,\rho,\varepsilon):={\rm F}(\boldsymbol{\xi};x,\varepsilon\rho,\varepsilon),\quad \mathscr{P}[\mathsf{F}](\boldsymbol{\xi};x,\rho,\varepsilon):=\mathscr{L}[{\rm F}](\boldsymbol{\xi};x,\varepsilon\rho,\varepsilon),\quad\mathscr{N}[\mathsf{F}](\boldsymbol{\xi};x,\rho,\varepsilon):=\mathscr{M}[{\rm F}](\boldsymbol{\xi};x,\varepsilon\rho,\varepsilon)\\ \label{Pformula9}
\Longrightarrow\quad\mathscr{P}[\mathsf{F}](\boldsymbol{\xi};x,\rho,\varepsilon)=\frac{1}{\varepsilon^2}\Bigg[\frac{\kappa}{4}\partial_\rho^2+\frac{(1-2\rho)\partial_\rho}{\rho(1-\rho)}-\frac{\theta_1}{\rho^2}-\frac{h}{(1-\rho)^2}\Bigg]\mathsf{F}(\boldsymbol{\xi};x,\rho,\varepsilon)+\frac{\partial_\varepsilon \mathsf{F}(\boldsymbol{\xi};x,\rho,\varepsilon)}{\varepsilon\rho(1-\rho)}.\end{gather}

\begin{figure}[t]
\centering
\includegraphics[scale=0.27]{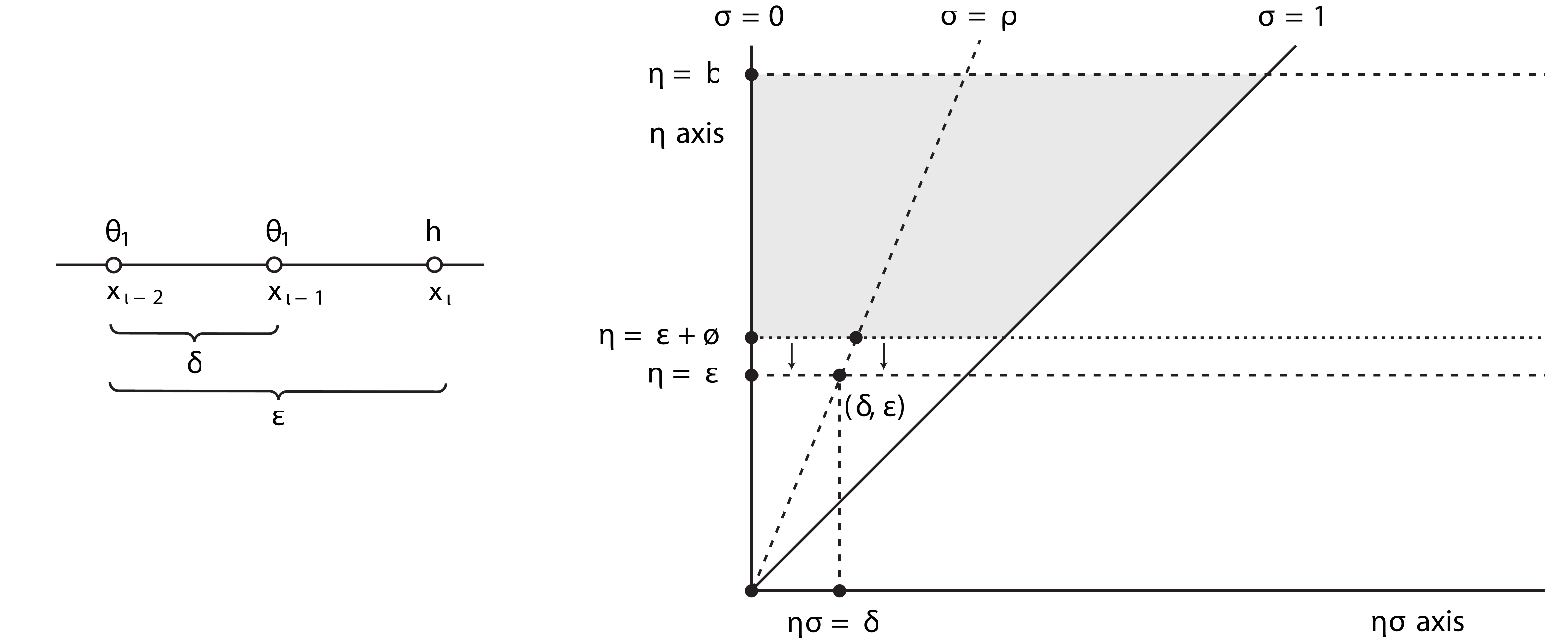}
\caption{Geometry of the statement of lemma \ref{closelem} (left) with the conformal weight of each point shown above that point, and geometry of the coordinate system (\ref{translate}) for the Green function (\ref{Greenfunc}) and the derivation of integral equation (\ref{simplestGreenID}) (right).}
\label{Geometry}
\end{figure}

Using an appropriate Green function, we invert the differential operator $\mathscr{P}$ to write the null-state PDE $\mathscr{P}[\mathsf{F}]=\mathscr{N}[\mathsf{F}]$ as an integral equation.  The Green function $G$ must solve the adjoint equation that follows from (\ref{Pformula9}), namely
\be\label{Pstar}\mathscr{P}^*[G](\rho,\varepsilon;\sigma,\eta)=\frac{1}{\eta^2}\Bigg[\mathscr{Q}^*-\frac{\eta\partial_\eta}{\sigma(1-\sigma)}\Bigg]G(\rho,\varepsilon;\sigma,\eta)=D(\sigma-\rho)D(\eta-\varepsilon),\ee
for all $0<\sigma<1$ and $0<\eta<b$ (where $b$ is so small that $x_{\iota-2}+\eta<x_{\iota+1}$ if $\iota<M$), with $\rho$ and $\varepsilon$ as parameters fixed in these respective ranges, where $D$ is the Dirac delta function, and where $\mathscr{Q}^*$ is the differential operator given by 
\be\label{Qstar}\mathscr{Q}^*[G](\rho,\varepsilon;\sigma,\eta):=\Bigg[\frac{\kappa}{4}\partial_\sigma^2-\frac{(1-2\sigma)\partial_\sigma}{\sigma(1-\sigma)}+\frac{1-\theta_1}{\sigma^2}+\frac{1-h}{(1-\sigma)^2}+\frac{1}{\sigma}+\frac{1}{1-\sigma}\Bigg]G(\rho,\varepsilon;\sigma,\eta).\ee
Now we specify boundary conditions for $G(\rho,\varepsilon;\sigma,\eta)$ as $\sigma\downarrow0$ or $\sigma\uparrow1$.  Identifying convenient  conditions is not completely straightforward because the coefficients in (\ref{Pstar}, \ref{Qstar}) blow up there.  However, the boundary behavior
\be\label{sigmabc}G(\rho,\varepsilon;\sigma,\eta)\underset{\sigma\downarrow0}{\sim}A_\ell\sigma^{\Delta^+(\theta_1)+4/\kappa},\quad G(\rho,\varepsilon;\sigma,\eta)\underset{\sigma\uparrow1}{\sim}A_r(1-\sigma)^{\Delta^+(h)+4/\kappa}\ee
for unspecified constants $A_\ell$ and $A_r$ will emerge as the most natural to use.  Furthermore, we require $G(\rho,\varepsilon;\sigma,\eta)$ to be causal with respect to $\varepsilon$:
\be\label{etabc}G(\rho,\varepsilon;\sigma,\eta)=0,\quad 0<\eta<\varepsilon.\ee 
Now, the related homogeneous differential equation $\mathscr{P}^*[G](\sigma,\eta)=0$ admits separable solutions of the form $G(\sigma,\eta)=\Sigma_\lambda(\sigma){\rm H}_\lambda(\eta)$, where $\lambda$ is the separation variable and $\Sigma$ and ${\rm H}$ respectively solve
\be\label{Qeqn}\Bigg[\mathscr{Q}^*+\frac{\lambda-1}{\sigma(1-\sigma)}\Bigg]\Sigma_\lambda(\sigma)=0,\quad[\eta\partial_\eta+(\lambda-1)]{\rm H}_\lambda(\eta)=0.\ee
The differential equation for ${\rm H}_{\lambda}$ admits power-law solutions, but that for $\Sigma_\lambda$ is more complicated, being (to within an integration factor) an irregular Sturm-Liouville problem on $[0,1]$ \cite{folland}.  After substituting
\begin{gather}\label{lambda}4\lambda_n=\kappa n^2+[8-\kappa+2\kappa(\Delta^+(h)+\Delta^+(\theta_1))]n+4(\Delta^+(h)+\Delta^+(\theta_1))+2\kappa\Delta^+(h)\Delta^+(\theta_1),\\ 
\label{fySigma}\Sigma_{\lambda_n}(\sigma)=f(\sigma)(P_n\circ y)(\sigma),\quad f(\sigma):=\sigma^{\Delta^+(\theta_1)+4/\kappa}(1-\sigma)^{\Delta^+(h)+4/\kappa},\quad y(\sigma):=2\sigma-1,\end{gather}
into the differential equation for $\Sigma_\lambda$ (\ref{Qeqn}) (with $n$ any real number), we find that $P_n$ satisfies the Jacobi differential equation $J^{(\alpha,\beta)}[P_n]=n(n+\alpha+\beta+1)P_n$ \cite{folland, szego}, with $J^{(\alpha,\beta)}$, $\alpha>0$ and $\beta>0$ given as follows (definition \ref{kpzdefn}):
\begin{gather}\label{JacobiODE}J^{(\alpha,\beta)}:=(1-y^2)\partial_y^2+[\beta-\alpha-(\alpha+\beta+2)y]\partial_y\\
\label{Jacobiparams}\alpha:=\bindnasrepma(h)=\frac{4}{\kappa}+2\Delta^+(h)-1>0,\quad\beta:=\bindnasrepma(\theta_1)=\frac{4}{\kappa}+2\Delta^+(\theta_1)-1>0.
\end{gather}
The boundary condition (\ref{sigmabc}) requires that $P_n(y)$ not blow up as $y\downarrow-1$ ($\sigma\downarrow0$) or as $y\uparrow1$ ($\sigma\uparrow1)$.  This condition restricts $n$ to a non-negative integer, so we find that $P_n$ is therefore the $n$th Jacobi polynomial, given by \cite{szego}
\be\label{Jacobipolys}P_n^{(\alpha,\beta)}(y):=\frac{\Gamma(\alpha+n+1)}{\Gamma(n+1)\Gamma(\alpha+\beta+n+1)}\sum_{m=0}^n\binom{n}{m}\frac{\Gamma(\alpha+\beta+n+m+1)}{\Gamma(\alpha+m+1)}\left(\frac{y-1}{2}\right)^m,\quad-1\leq y\leq1.\ee
The set $\{P_n^{(\alpha,\beta)}\}_{n=0}^\infty$ is an orthogonal basis for $L_\varrho^2(-1,1)$ with weight function $\varrho(y)=\varrho^{(\alpha,\beta)}(y):=(1-y)^\alpha(1+y)^\beta$ \cite{szego,folland}.  (We recall that
\be\label{wdefn}L^2_\omega(a,b):=\left\{f:[a,b]\rightarrow\mathbb{R}\,\Bigg|\,\int_0^1 |f(y)|^2\omega(y)\,{\rm d}y<\infty\right\},\quad\text{$\omega$ a positive-valued, continuous function on $[a,b]$},\ee
defines the Hilbert space of $\omega$-square-integrable functions.)  Thus, $\mathscr{B}:=\{P_n^{(\alpha,\beta)}\circ y\}_{n=0}^\infty$ is an orthogonal basis for $L^2_w(0,1)$ with its weight function
\be \label{wtfnct}
w^{(\alpha,\beta)}(\sigma):=\sigma^\beta(1-\sigma)^\alpha
\ee
 \cite{szego}.  The squares of the norms of these two basis elements are \cite{szego}
\bea\label{norm1}&&\|P_n^{(\alpha,\beta)}\|_{L^2_\varrho(-1,1)}^2=h_n^{(\alpha,\beta)}:=\frac{2^{\alpha+\beta+1}\Gamma(n+\alpha+1)\Gamma(n+\beta+1)}{(2n+\alpha+\beta+1)\Gamma(n+1)\Gamma(n+\alpha+\beta+1)},\quad\\
\label{norm2}&&\|P_n^{(\alpha,\beta)}\circ y\|_{L^2_w(0,1)}^2=2^{-\alpha-\beta-1}h_n^{(\alpha,\beta)}.\eea
The boundary condition (\ref{sigmabc}) with (\ref{fySigma}) implies that $G(\rho,\varepsilon;\sigma,\eta)/f(\sigma)$, as a function of $\sigma$ with $\rho$, $\varepsilon$, and $\eta$ fixed, is an element of $L^2_w(0,1)$.  Therefore, it equals (in the $L^2_w(0,1)$ sense) its Jacobi-Fourier expansion over $\mathscr{B}$, and we have
\be G(\rho,\varepsilon;\sigma,\eta)=f(\sigma)\sum_{n=0}^\infty c_n(\rho,\varepsilon;\eta)P_n^{(\alpha,\beta)}(2\sigma-1),\quad0<\sigma<1,\ee
for some functions $c_n$.  Similarly, we expand the Dirac delta function $D(\sigma-\rho)$ over $\mathscr{B}$.
After inserting these expansions into (\ref{Pstar}) and exchanging the infinite summation with differentiation (assuming for now that this is allowed), we find
\be\label{lambdaetaode}[\eta\partial_\eta+(\lambda_n-1)]c_n(\rho,\varepsilon;\eta)=-2^{\alpha+\beta+1}\varepsilon^2\,\frac{\rho(1-\rho)w^{(\alpha,\beta)}(\rho)}{h_n^{(\alpha,\beta)}f(\rho)}P_n^{(\alpha,\beta)}(2\rho-1)D(\eta-\varepsilon),\quad n\in\mathbb{Z}^+\cup\{0\}.\ee
The boundary condition (\ref{etabc}) implies that $c_n(\rho,\varepsilon;\eta)=0$ for $0<\eta<\varepsilon$ and all $n\in\mathbb{Z}^+\cup\{0\}$.  Hence, the unique solution to (\ref{lambdaetaode}) is 
\be c_n(\rho,\varepsilon;\eta)=-2^{\alpha+\beta+1}\Theta(\eta-\varepsilon)\,\varepsilon\,\frac{\rho(1-\rho)w^{(\alpha,\beta)}(\rho)}{h_n^{(\alpha,\beta)}f(\rho)} P_n^{(\alpha,\beta)}(2\rho-1)\left(\frac{\varepsilon}{\eta}\right)^{\lambda_n-1},\quad n\in\mathbb{Z}^+\cup\{0\},\ee
where $\Theta$ is the Heaviside step function.  Putting everything together, we find the following formula for the Green function solving (\ref{Pstar}) and subject to boundary conditions (\ref{sigmabc}, \ref{etabc}):
\begin{multline}\label{Greenfunc}G(\rho,\varepsilon;\sigma,\eta)=-\Theta(\eta-\varepsilon)\sigma^{\beta+1}(1-\sigma)^{\alpha+1}\left(\frac{\rho}{\sigma}\right)^{\Delta^+(\theta_1)}\left(\frac{1-\rho}{1-\sigma}\right)^{\Delta^+(h)}\\
\times\eta\sum_{n=0}^\infty\left(\frac{\varepsilon}{\eta}\right)^{\lambda_n}\frac{P^{(\alpha,\beta)}_n(2\rho-1)P^{(\alpha,\beta)}_n(2\sigma-1)}{2^{-\alpha-\beta-1}h_n^{(\alpha,\beta)}}.\end{multline}
Here, $\alpha$ and $\beta$ are given by (\ref{Jacobiparams}), $h_n^{(\alpha,\beta)}$ is given by (\ref{norm1}), $\lambda_n$ is given by (\ref{lambda}), $P^{(\alpha,\beta)}_n$ is the $n$th Jacobi polynomial given by (\ref{Jacobipolys}), $\Delta^+$ is given by (\ref{kpz}), and $\theta_1$ is given by (\ref{bdysleg}).  We find it convenient to write this Green function as
\be\label{Greenfuncalt}G(\rho,\varepsilon;\sigma,\eta)=-\Theta(\eta-\varepsilon)\sigma^{\beta+1}(1-\sigma)^{\alpha+1}\left(\frac{\rho}{\sigma}\right)^{\Delta^+(\theta_1)}\left(\frac{1-\rho}{1-\sigma}\right)^{\Delta^+(h)}\eta\left(\frac{\varepsilon}{\eta}\right)^{\lambda_0}K^{(\alpha,\beta)}(\rho,\sigma,-\log(\varepsilon/\eta)^{\kappa/4}),\ee
where $K^{(\alpha,\beta)}(\rho,\sigma,t)$ is given by the  series (which is evidently absolutely and uniformly convergent over $(\rho,\sigma,t)\in[0,1]\times[0,1]\times[a,\infty)$ for any $a>0$): 
\be\label{Jacobiheat}K^{(\alpha,\beta)}(\rho,\sigma,t):=\sum_{n=0}^\infty e^{-tn(n+\alpha+\beta+1)}\frac{P^{(\alpha,\beta)}_n(2\rho-1)P^{(\alpha,\beta)}_n(2\sigma-1)}{2^{-\alpha-\beta-1}h_n^{(\alpha,\beta)}}.\ee
We note that the function $G^{(\alpha,\beta)}_t(y,z):=2^{-\alpha-\beta-1}K^{(\alpha,\beta)}(\frac{y+1}{2},\frac{z+1}{2},t)$ is the Jacobi heat kernel \cite{nowsj1, nowsj2}.  This particular heat kernel arises in the following application.  For any square-$\varrho$-integrable function $F:[-1,1]\rightarrow\mathbb{R}$, the function
\be u:[-1,1]\times(0,\infty)\rightarrow\mathbb{R},\quad u(y,t):=\int_{-1}^1G^{(\alpha,\beta)}_t(y,z)F(z)\varrho^{(\alpha,\beta)}(z)\,{\rm d}z \ee
solves the Jacobi heat equation $[J^{(\alpha,\beta)}-\partial_t]u(y,t)=0$ (\ref{JacobiODE}) with the initial condition $u(y,0+)=F(y)$ \cite{ander}.  In an abuse of terminology, we also refer to (\ref{Jacobiheat}) as the Jacobi heat kernel below.

We state a few facts concerning the Jacobi heat kernel that were recently discovered by A.\ Nowak and P.\ Sj\"ogren \cite{nowsj1, nowsj2} and that we use in this proof.  First \cite{nowsj1}, for all continuous $f:[0,1]\rightarrow\mathbb{R}$,
\be\label{idlim}\lim_{t\downarrow0}\int_0^1K^{(\alpha,\beta)}(\rho,\sigma,t)f(\sigma)\,w^{(\alpha,\beta)}(\sigma)\,{\rm d}\sigma=f(\rho)\quad\text{uniformly over $\rho\in[0,1]$}.\ee
Second, the authors of \cite{nowsj1, nowsj2} have derived sharp estimates for the short time behavior of the Jacobi heat kernel, and their result (stated slightly differently) reads as follows \cite{nowsj2}.  We let
\be\label{Lambda}\Lambda^{(\alpha,\beta)}(\theta,\phi,t):=[t+\sin(\theta/2)\sin(\phi/2)]^{-\alpha-1/2}[t+\cos(\theta/2)\cos(\phi/2)]^{-\beta-1/2},\ee
with $\theta,\phi\in[0,\pi]$, $t\geq0$, and $\alpha,\beta\geq-1/2$.  (We recall from (\ref{Jacobiparams}) that $\alpha,\beta>0$ in this proof.)  Then for any $T>0$, there exist positive constants $C$, $c_1$
and $c_2$, depending only on $\alpha$, $\beta$ and $T$, such that for all $\theta,\phi\in[0,\pi]$,
\be\label{bounds}\begin{cases}\begin{aligned}\displaystyle{\frac{\sqrt{c_1}}{C}\Lambda^{(\alpha,\beta)}(\theta,\phi,t)\frac{e^{-(\theta-\phi)^2/c_1t}}{\sqrt{\pi c_1t}}}&\leq K^{(\alpha,\beta)}(\cos^2\theta/2,\cos^2\phi/2,t) \\
&\leq \displaystyle{C\sqrt{c_2}\,\Lambda^{(\alpha,\beta)}(\theta,\phi,t)\frac{e^{-(\theta-\phi)^2/c_2t}}{\sqrt{\pi c_2t}}}\end{aligned}\,\,\,,& 0\leq t\leq T\\ 
\displaystyle{\frac{1}{C}\leq K^{(\alpha,\beta)}(\cos^2\theta/2,\cos^2\phi/2,t)\leq C},& t>T \end{cases}.\ee
We recognize the exponential divided by the square root in (\ref{bounds}) as the standard heat kernel for an infinite rod \cite{folland}, a fact that motivates (\ref{idlim}).  Equation (\ref{bounds}) shows that $K^{(\alpha,\beta)}(\rho,\sigma,t)$ is positive, a feature previously noted in \cite{karlin, gasper, bochner}.

Next, in analogy with the proof of lemma \ref{firstlimitlem}, we use the Green function (\ref{Greenfunc}) to derive an integral equation for $\mathsf{F}(\rho,\varepsilon)$.  The usual procedure is to equate two evaluations of the  definite integral
\be\label{difference}\int_0^b\int_0^1\big[\mathscr{P}[\mathsf{F}](\boldsymbol{\xi};x,\sigma,\eta)G(\rho,\varepsilon;\sigma,\eta)-\mathscr{P}^*[G](\rho,\varepsilon;\sigma,\eta)\mathsf{F}(\boldsymbol{\xi};x,\sigma,\eta)\big]\,{\rm d}\eta\,{\rm d}\sigma.\ee
First, we use definitions (\ref{Pformula9}--\ref{Qstar}) and  the product rule to write the integrand of (\ref{difference}) as a sum of derivatives with respect to $\sigma$ and $\eta$ and integrate, finding only boundary terms \cite{folland}.  With $G(\rho,\varepsilon;\sigma,\eta)=0$ for $\eta\leq\varepsilon$ (\ref{etabc}), we find
\begin{multline}\label{theboundaryterms}\frac{\kappa}{4}\int_\varepsilon^b\frac{1}{\eta^2}[\partial_\sigma\mathsf{F}(\boldsymbol{\xi};x,\sigma,\eta)G(\rho,\varepsilon;\sigma,\eta)-\mathsf{F}(\boldsymbol{\xi};x,\sigma,\eta)\partial_\sigma G(\rho,\varepsilon;\sigma,\eta)]\bigg|_{\sigma=0}^{\sigma=1}{\rm d}\eta\\
+\int_\varepsilon^b\frac{1-2\sigma}{\eta^2\sigma(1-\sigma)}\mathsf{F}(\boldsymbol{\xi};x,\sigma,\eta)G(\rho,\varepsilon;\sigma,\eta)\bigg|_{\sigma=0}^{\sigma=1}{\rm d}\eta+\int_0^1\frac{1}{b\sigma(1-\sigma)}\mathsf{F}(\boldsymbol{\xi};x,\sigma,b)G(\rho,\varepsilon;\sigma,b)\,{\rm d}\sigma.\end{multline}
Second, we insert $\mathscr{P}[\mathsf{F}]=\mathscr{N}[\mathsf{F}]$ and (\ref{Pstar}) into (\ref{difference}), and we integrate over the delta function in the first term.  This treatment of the formal equation (\ref{Pstar}) is non-rigorous, but if we momentarily allow it, then we find 
\be\label{straightintegration}\int_\varepsilon^b\int_0^1 \mathscr{N}[\mathsf{F}](\boldsymbol{\xi};x,\sigma,\eta)G(\rho,\varepsilon;\sigma,\eta)\,{\rm d}\eta\,{\rm d}\sigma-\mathsf{F}(\boldsymbol{\xi};x,\rho,\varepsilon).\ee
(Again, we have used $G(\rho,\varepsilon;\sigma,\eta)=0$ for $\eta\leq\varepsilon$.)  Equating (\ref{straightintegration}) with (\ref{theboundaryterms}) gives an integral equation with a form reminiscent of the form of the previous integral equations (\ref{tildeH}, \ref{I}):
\be\label{preGreenID}\mathsf{F}(\boldsymbol{\xi};x,\rho,\varepsilon)=\int_\varepsilon^b\int_0^1 \mathscr{N}[\mathsf{F}](\boldsymbol{\xi};x,\sigma,\eta)G(\rho,\varepsilon;\sigma,\eta)\,{\rm d}\eta\,{\rm d}\sigma-\text{boundary terms.}\ee

To rigorously derive (\ref{preGreenID}), we change the lower $\eta$-limit in (\ref{difference}) to $\eta=\varepsilon+\o$ with $0<\o\ll\varepsilon$, recalculate (\ref{theboundaryterms}) after this change, and equate the result to (\ref{difference}) with $\mathscr{P}^*[G]=0$ (because $\eta\geq\varepsilon+\o>\varepsilon$) and $\mathscr{P}[\mathsf{F}]=\mathscr{N}[\mathsf{F}]$, finding
\begin{multline}\label{GreenID}\int_{\varepsilon+\o}^b\int_0^1 \mathscr{N}[\mathsf{F}](\boldsymbol{\xi};x,\sigma,\eta)G(\rho,\varepsilon;\sigma,\eta)\,{\rm d}\eta\,{\rm d}\sigma\\
\begin{aligned}&=\frac{\kappa}{4}\int_{\varepsilon+\o}^b\frac{1}{\eta^2}[\partial_\sigma\mathsf{F}(\boldsymbol{\xi};x,\sigma,\eta)G(\rho,\varepsilon;\sigma,\eta)-\mathsf{F}(\boldsymbol{\xi};x,\sigma,\eta)\partial_\sigma G(\rho,\varepsilon;\sigma,\eta)]\bigg|_{\sigma=0}^{\sigma=1}{\rm d}\eta\\
&+\int_{\varepsilon+\o}^b\frac{1-2\sigma}{\eta^2\sigma(1-\sigma)}\mathsf{F}(\boldsymbol{\xi};x,\sigma,\eta)G(\rho,\varepsilon;\sigma,\eta)\bigg|_{\sigma=0}^{\sigma=1}{\rm d}\eta+\int_0^1\frac{1}{\eta\sigma(1-\sigma)}\mathsf{F}(\boldsymbol{\xi};x,\sigma,\eta)G(\rho,\varepsilon;\sigma,\eta)\bigg|^{\eta=b}_{\eta=\varepsilon+\o}d\sigma.\end{aligned}\end{multline}
Following a strategy described in \cite{roach} (Sect.\ 9.7), we later send $\o \downarrow0$ (figure \ref{Geometry}) in order to derive from (\ref{GreenID}) an integral equation with the desired form (\ref{preGreenID}).  

But first, we show that the terms on the right side of (\ref{GreenID}) involving integration with respect to $\eta$ in fact vanish.
Now, the boundary condition (\ref{sigmabc}) with the asymptotic behavior (see condition \ref{nextcond3} and the conclusion of lemma \ref{secondlimitlem})
\be\label{leftrightbehavior}\mathsf{F}(\boldsymbol{\xi};x,\sigma,\eta)\underset{\sigma\downarrow0}{\sim}B_\ell(\boldsymbol{\xi};x,\eta)\sigma^{\Delta^+(\theta_1)},\quad\mathsf{F}(\boldsymbol{\xi};x,\sigma,\eta)\underset{\sigma\uparrow1}{\sim}B_r(\boldsymbol{\xi};x,\eta)(1-\sigma)^{\Delta^+(h)},\ee
imply that the integrand of the second definite integral on the right side of (\ref{GreenID}) is $O(\sigma^\beta)$ as $\sigma\downarrow0$ (resp.\ O$((1-\sigma)^\alpha)$ as $\sigma\uparrow1$) with $\alpha,\beta>0$ (\ref{Jacobiparams}).  Thus, this definite integral vanishes as claimed.

By a similar argument, the first definite integral on the right side of (\ref{GreenID}) vanishes too.  Indeed, the second term in its integrand vanishes as $\sigma\downarrow0$ or $\sigma\uparrow1$ thanks to (\ref{Greenfunc}) and (\ref{leftrightbehavior}).  Now turning our attention to the first term, we estimate $\partial_\sigma\mathsf{F}(\boldsymbol{\xi};x,\sigma,\eta)$ with arguments very similar to those in the proof of lemma \ref{firstlimitlem}.  Expressed in terms of $\boldsymbol{\xi}$, $x$, $\varepsilon$, and $\rho$ (see above (\ref{Fdef2}) and (\ref{translate})) (we replace $\rho$ by $\sigma$ later), the null-state PDE (\ref{hpde}) with $j\not\in\{\iota-2,\iota-1,\iota\}$ is
\begin{multline}\label{newnotii+1}\Bigg[\frac{\kappa}{4}\partial_j^2+\sum_{k\neq j}\left(\frac{\partial_k}{\xi_k-\xi_j}-\frac{\theta_1}{(\xi_k-\xi_j)^2}\right)+\frac{\partial_x}{x-\xi_j}-\frac{\theta_1}{(x-\xi_j)^2}-\frac{\varepsilon\rho(1-\rho)\partial_\rho}{(x-\xi_j)(x+\varepsilon\rho-\xi_j)(x+\varepsilon-\xi_j)}\\
-\frac{\theta_1}{(x+\varepsilon\rho-\xi_j)^2}-\frac{\varepsilon\partial_\varepsilon}{(x-\xi_j)(x+\varepsilon-\xi_j)}-\frac{h}{(x+\varepsilon-\xi_j)^2}\Bigg]\mathsf{F}(\boldsymbol{\xi};x,\rho,\varepsilon)=0. \end{multline}
After summing (\ref{newnotii+1}) over $j\not\in\{\iota-2,\iota-1,\iota\}$, we use the conformal Ward identities (\ref{wardidh}) to replace the derivatives $\partial_x\mathsf{F}$, $\partial_\rho\mathsf{F}$, and $\partial_\varepsilon\mathsf{F}$ in the resulting PDE with derivatives of $\mathsf{F}$ in the coordinates of $\boldsymbol{\xi}$.  These identities are 
\be\begin{gathered}\label{newward}\bigg[\sideset{}{_j}\sum\partial_j+\partial_x\bigg]\mathsf{F}(\boldsymbol{\xi};x,\rho,\varepsilon)=0,\qquad\bigg[\sideset{}{_j}\sum(\xi_j\partial_j+\theta_1)+x\partial_x+\theta_1+\varepsilon\partial_\varepsilon+h+\theta_1\bigg]\mathsf{F}(\boldsymbol{\xi};x,\rho,\varepsilon)=0,\\
\bigg[\sideset{}{_j}\sum(\xi_j^2\partial_j+2\theta_1\xi_j)+x^2\partial_x+2\theta_1x-\rho(1-\rho)\varepsilon\partial_\rho+(2x+\varepsilon\rho)\theta_1+(2x+\varepsilon)\varepsilon\partial_\varepsilon+2h(x+\varepsilon)\bigg]\mathsf{F}(\boldsymbol{\xi};x,\rho,\varepsilon)=0.\end{gathered}\ee
Hence, we find a strictly elliptic PDE in the coordinates of $\boldsymbol{\xi}$, with $x$, $\rho$ and $\varepsilon$ as parameters, and whose coefficients do not blow up or vanish as $\varepsilon\downarrow0$, $\rho\downarrow0$, or $\rho\uparrow1$.  Therefore, the Schauder interior estimate (Cor.\ 6.3 of \cite{giltru}) gives
\be\label{newSchauder}d\sup_{ U}\left|\left\{\begin{array}{ccccc}\mathsf{F}, & \partial_j\mathsf{F}, & \rho(1-\rho)\partial_\rho\mathsf{F}, & \partial_x\mathsf{F}, & \varepsilon\partial_\varepsilon\mathsf{F} \end{array}\right\}(\boldsymbol{\xi};x,\rho,\varepsilon)\right|\leq C(R)\sup_{ V}|\mathsf{F}(\boldsymbol{\xi};x,\rho,\varepsilon)|\ee
for all $\rho\in(0,1)$ and $0<\varepsilon<b$, where $ U\subset\subset V\subset\subset\pi_{\iota-1,\iota}(\Omega_0^M)$ are open sets, $C$ is some positive-valued function, $d=\text{dist}(\partial U,\partial V)$, and $R=\text{diam}( V)/2$.  It follows from  (\ref{sigmabc}, \ref{leftrightbehavior}) and (\ref{newSchauder}) with $\sigma$ replacing $\rho$ that the first term of the first definite integral on the right side of (\ref{GreenID}) vanishes as $\sigma\downarrow0$ or $\sigma\uparrow1$ too.  Hence, (\ref{GreenID}) becomes
\begin{multline}\label{simplerGreenID}-\int_0^1\frac{1}{(\varepsilon+\o)\sigma(1-\sigma)}\mathsf{F}(\boldsymbol{\xi};x,\sigma,\varepsilon+\o)G(\rho,\varepsilon;\sigma,\varepsilon+\o)\,{\rm d}\sigma\\
=\int_{\varepsilon+\o}^b\int_0^1 \mathscr{N}[\mathsf{F}](\boldsymbol{\xi};x,\sigma,\eta)G(\rho,\varepsilon;\sigma,\eta)\,{\rm d}\eta\,{\rm d}\sigma-\int_0^1\frac{1}{b\sigma(1-\sigma)}\mathsf{F}(\boldsymbol{\xi};x,\sigma,b)G(\rho,\varepsilon;\sigma,b)\,{\rm d}\sigma.\end{multline}
It is easy to show that integrands appearing on either side of (\ref{simplerGreenID}) are $O(\sigma^\beta)$ as $\sigma\downarrow0$ and $O((1-\sigma)^\alpha)$ as $\sigma\uparrow1$.  Because $\alpha$ and $\beta$ are positive (\ref{Jacobiparams}), the definite integrals with respect to $\sigma$ on either side converge.

Next, we send $\o\downarrow0$ to derive an integral equation of the form (\ref{preGreenID}).  In this limit, the left side of (\ref{simplerGreenID}) goes to $\mathsf{F}(\boldsymbol{\xi};x,\rho,\varepsilon)$.  To motivate why this is true, we use (\ref{wtfnct}) and (\ref{Greenfunc}) to write the limit of this left side as 
\be\label{leftside1}\lim_{\o\downarrow0}\int_0^1\left(\frac{\rho}{\sigma}\right)^{\Delta^+(\theta_1)}\left(\frac{1-\rho}{1-\sigma}\right)^{\Delta^+(h)}\sum_{n=0}^\infty\left(\frac{\varepsilon}{\varepsilon+\o}\right)^{\lambda_n}\frac{P^{(\alpha,\beta)}_n(2\sigma-1)P^{(\alpha,\beta)}_n(2\rho-1)}{2^{-\alpha-\beta-1}h_n^{(\alpha,\beta)}}\mathsf{F}(\boldsymbol{\xi};x,\sigma,\varepsilon+\o)\,w^{(\alpha,\beta)}(\sigma)\,{\rm d}\sigma,\ee
we commute the integration with respect to $\sigma$ with the infinite sum, we set $\o=0$, and we recognize (essentially) the Jacobi-Fourier series for $\mathsf{F}(\boldsymbol{\xi};x,\sigma,\varepsilon)$.  Now, to prove that (\ref{leftside1}) indeed equals $\mathsf{F}(\boldsymbol{\xi};x,\sigma,\varepsilon)$, we use (\ref{Greenfuncalt}) to write it as
\newpage
\begin{multline}\label{leftside2}\lim_{t\downarrow0}e^{-4\lambda_0t/\kappa}\Bigg[\int_0^1K^{(\alpha,\beta)}(\rho,\sigma,t)\left(\frac{\rho}{\sigma}\right)^{\Delta^+(\theta_1)}\left(\frac{1-\rho}{1-\sigma}\right)^{\Delta^+(h)}\mathsf{F}(\boldsymbol{\xi};x,\sigma,\varepsilon)\,w^{(\alpha,\beta)}(\sigma)\,{\rm d}\sigma\\
+\int_0^1K^{(\alpha,\beta)}(\rho,\sigma,t)\left(\frac{\rho}{\sigma}\right)^{\Delta^+(\theta_1)}\left(\frac{1-\rho}{1-\sigma}\right)^{\Delta^+(h)}[\mathsf{F}(\boldsymbol{\xi};x,\sigma,\varepsilon e^{4t/\kappa})-\mathsf{F}(\boldsymbol{\xi};x,\sigma,\varepsilon)]\,w^{(\alpha,\beta)}(\sigma)\,{\rm d}\sigma\Bigg],\end{multline}
where $\varepsilon+\o=\varepsilon e^{4t/\kappa}$.  As a consequence of (\ref{idlim}), the first definite integral in (\ref{leftside2}) goes to $\mathsf{F}(\boldsymbol{\xi};x,\rho,\varepsilon)$ as $t\downarrow0$.  To show that the second definite integral vanishes as $t\downarrow0$, we let $\mathcal{E}_\rho$ be a closed interval within $(0,1)$ and containing $\rho$, we divide the region of integration into two parts thus,
\be\label{2ndterm}\left(\int_{\mathcal{E}_\rho}+\int_{(0,1)\setminus\mathcal{E}_\rho}\right)K^{(\alpha,\beta)}(\rho,\sigma,t)\left(\frac{\rho}{\sigma}\right)^{\Delta^+(\theta_1)}\left(\frac{1-\rho}{1-\sigma}\right)^{\Delta^+(h)}[\mathsf{F}(\boldsymbol{\xi};x,\sigma,\varepsilon e^{4t/\kappa})-\mathsf{F}(\boldsymbol{\xi};x,\sigma,\varepsilon)]\,w^{(\alpha,\beta)}(\sigma)\,{\rm d}\sigma,\ee
and we show that both definite integrals vanish as $t\downarrow0$.  
\begin{enumerate}
\item\label{thefirstitem} First, we show that the definite integral over $\mathcal{E}_\rho$ in (\ref{2ndterm}) vanishes as $t\downarrow0$.  The magnitude of this definite integral is bounded above by
\be\label{term2}\sup_{\mathcal{E}_\rho}|\mathsf{F}(\boldsymbol{\xi};x,\sigma,\varepsilon e^{4t/\kappa})-\mathsf{F}(\boldsymbol{\xi};x,\sigma,\varepsilon)|\int_0^1K^{(\alpha,\beta)}(\rho,\sigma,t)\left(\frac{\rho}{\sigma}\right)^{\Delta^+(\theta_1)}\left(\frac{1-\rho}{1-\sigma}\right)^{\Delta^+(h)}\,w^{(\alpha,\beta)}(\sigma)\,{\rm d}\sigma, \ee
because (\ref{bounds}) implies that $K^{(\alpha,\beta)}$ is positive-valued.  Furthermore, the definite integral in (\ref{term2}) approaches one as $t\downarrow0$ thanks to (\ref{idlim}).  Now, to show that the supremum in (\ref{term2}) vanishes, we use the integral equation 
\begin{multline}\label{Fprop}\mathsf{F}(\boldsymbol{\xi};x,\sigma,\varepsilon)-\mathsf{F}(\boldsymbol{\xi};x,\sigma,b)=\left[\left(\frac{\varepsilon}{b}\right)^{\Delta^-(h)}-1\right]\mathsf{F}(\boldsymbol{\xi};x,\sigma,b)-\frac{\kappa}{4}\varepsilon^{\Delta^-(h)}J(\varepsilon,b)\partial_b[b^{-\Delta^-(h)}\mathsf{F}(\boldsymbol{\xi};x,\sigma,b)]\\
+\int_\varepsilon^b\left(\frac{\varepsilon}{\eta}\right)^{\Delta^-(h)}J(\varepsilon,\eta)\Bigg[\frac{\sigma\partial_{\sigma}}{\eta^2(1-\sigma)^2}+\frac{\theta_1}{\eta^2(1-\sigma)^2}+\frac{\partial_x}{\eta}+\sum_{j=1}^{M-3}\left(\frac{\theta_1}{(\xi_j-x-\eta)^2}-\frac{\partial_j}{\xi_j-x-\eta}\right)\Bigg]\mathsf{F}(\boldsymbol{\xi};x,\sigma,\eta)\,{\rm d}\eta,\end{multline}
with $J$ defined in (\ref{Jgreen}) (except $h$ replaces $\theta_1$), $\varepsilon$ bounded above zero, and $\sigma\in\mathcal{E}_\rho$ bounded away from zero and one. Equation (\ref{Fprop}) follows from the null-state PDE (\ref{hpde}) with $j=\iota$ which, in terms of $\varepsilon$ and $\sigma=\delta/\varepsilon$, becomes
\begin{multline}\label{nullstateeps}\Bigg[\frac{\kappa}{4}\partial_\varepsilon^2+\frac{\partial_\varepsilon}{\varepsilon}-\frac{h}{\varepsilon^2}\Bigg]\mathsf{F}(\boldsymbol{\xi};x,\sigma,\eta)\\
=\Bigg[\frac{\sigma\partial_{\sigma}}{\varepsilon^2(1-\sigma)^2}+\frac{\theta_1}{\varepsilon^2(1-\sigma)^2}+\frac{\partial_x}{\varepsilon}+\sum_{j=1}^{M-3}\left(\frac{\theta_1}{(\xi_j-x-\varepsilon)^2}-\frac{\partial_j}{\xi_j-x-\varepsilon}\right)\Bigg]\mathsf{F}(\boldsymbol{\xi};x,\sigma,\varepsilon).\end{multline}
We find (\ref{Fprop}) after we invert the differential operator on the left side of (\ref{nullstateeps}), exactly as we did in the proof of lemma \red{3} of \cite{florkleb}.  (We note that $\varepsilon$ in (\ref{Fprop}, \ref{nullstateeps}) plays the role of $\delta$ in that proof.)  Thus, after setting $b=\varepsilon e^{-4t/\kappa}$ in (\ref{Fprop}), taking the supremum of (\ref{Fprop}) over $\sigma\in\mathcal{E}_\rho$, and noting that $J(\varepsilon,b)$ vanishes as $b\downarrow\varepsilon$, we find
\be\label{uniform}\sup_{\mathcal{E}_\rho}|\mathsf{F}(\boldsymbol{\xi};x,\sigma,\varepsilon e^{4t/\kappa})-\mathsf{F}(\boldsymbol{\xi};x,\sigma,\varepsilon)|\xrightarrow[\,\,\,t\downarrow0\,\,\,]{}0.\ee
Hence, (\ref{term2}) vanishes as $t\downarrow0$.  But because (\ref{term2}) is an upper bound of the magnitude of the definite integral over $\mathcal{E_\rho}$ in (\ref{2ndterm}), we conclude that this latter definite integral vanishes as $t\downarrow0$ too.
\item\label{theseconditem} Finally, we show that the definite integral over $(0,1)\setminus\mathcal{E}_\rho$ in (\ref{2ndterm}) vanishes as $t\downarrow0$.  Recalling that $K^{(\alpha,\beta)}$ is positive-valued, the magnitude of this definite integral is bounded above by
\be\label{term3}\sup_{(0,1)\setminus\mathcal{E}_\rho}|\mathsf{F}(\boldsymbol{\xi};x,\sigma,\varepsilon e^{4t/\kappa})-\mathsf{F}(\boldsymbol{\xi};x,\sigma,\varepsilon)|\int_{(0,1)\setminus\mathcal{E}_\rho}K^{(\alpha,\beta)}(\rho,\sigma,t)\left(\frac{\rho}{\sigma}\right)^{\Delta^+(\theta_1)}\left(\frac{1-\rho}{1-\sigma}\right)^{\Delta^+(h)}\,w^{(\alpha,\beta)}(\sigma)\,{\rm d}\sigma.\ee
After estimating the supremum, inserting the change of variables $\cos^2\theta/2=\sigma$ and $\cos^2\phi/2=\rho$, applying the estimate (\ref{bounds}) for some fixed $T>0$ (small enough so $x_\iota<x_{\iota+1}$ if the latter coordinate exists), and using the inequality $\Lambda^{(\alpha,\beta)}(\theta,\phi,t)\leq\Lambda^{(\alpha,\beta)}(\theta,\phi,0)$ for $t>0$ (\ref{Lambda}), we find that (\ref{term3}) is bounded above by
\begin{multline}\label{supsint}2C\sqrt{c_2}\Bigg(\sup_{\substack{\sigma\in(0,1) \\ t\in(0,T)}}|\mathsf{F}(\boldsymbol{\xi};x,\sigma,\varepsilon e^{4t/\kappa})|+\sup_{\sigma\in(0,1)}|\mathsf{F}(\boldsymbol{\xi};x,\sigma,\varepsilon)|\Bigg)\\
\times\int_0^\pi\frac{e^{-(\theta-\phi)^2/c_2t}}{\sqrt{\pi c_2 t}}\left(\frac{\sin\theta/2}{\sin\phi/2}\right)^{4/\kappa-1/2}\left(\frac{\cos\theta/2}{\cos\phi/2}\right)^{4/\kappa-1/2}\chi_{\{\cos^2\theta/2\not\in\mathcal{E}_\rho\}}(\theta)\,{\rm d}\theta,\end{multline}
where $\chi_S(x)$ is the indicator function on the set $S\subset\mathbb{R}$.  Because $\mathsf{F}(\boldsymbol{\xi};x,\sigma,\varepsilon e^{4t/\kappa})$, with $\boldsymbol{\xi}$, $x$, and $\varepsilon$ fixed, extends to a continuous function on $(\sigma,t)\in[0,1]\times[0,T]$ (thanks to the uniformness of the limits $\sigma\downarrow0$ and $\sigma\uparrow1$, as stated in lemma \ref{secondlimitlem}), the supremums in (\ref{supsint}) are finite.  Furthermore, as $t\rightarrow0$, the definite integral in (\ref{supsint}) converges to $\chi_{\{\cos^2\theta/2\not\in\mathcal{E}_\rho\}}(\theta=\phi)$ (see theorem \red{7.3} of \cite{folland}), and this equals zero because $\cos^2\phi/2=\rho\in\mathcal{E}_\rho$.  Hence, (\ref{term3}) vanishes as $t\downarrow0$.  Finally, because (\ref{term3}) is an upper bound of the magnitude of the definite integral over $(0,1)\setminus\mathcal{E_\rho}$ in (\ref{2ndterm}), we conclude that this latter definite integral vanishes as $t\downarrow0$ too.
\end{enumerate}
From items \ref{thefirstitem} and \ref{theseconditem} above, it follows that the second definite integral in (\ref{leftside2}) vanishes as $t\downarrow0$.  Therefore, (\ref{leftside2}) goes to $\mathsf{F}(\boldsymbol{\xi};x,\rho,\varepsilon)$ as $t\downarrow0$,
so we find the integral equation
\be\label{simplestGreenID}\mathsf{F}(\boldsymbol{\xi};x,\rho,\varepsilon)=\int_0^1\int_{\varepsilon}^b \mathscr{N}[\mathsf{F}](\boldsymbol{\xi};x,\sigma,\eta)G(\rho,\varepsilon;\sigma,\eta)\,{\rm d}\eta\,{\rm d}\sigma-\int_0^1\frac{1}{b\sigma(1-\sigma)}\mathsf{F}(\boldsymbol{\xi};x,\sigma,b)G(\rho,\varepsilon;\sigma,b)\,{\rm d}\sigma\ee
after sending $\o\downarrow0$ in (\ref{simplerGreenID}).  This has the form (\ref{preGreenID}) that we sought.

Now, we use the integral equation (\ref{simplestGreenID}) to find the smallest value that $p_0$ may attain in the estimate (\ref{initialest}), the nub of this lemma.  To this end, we replace $\mathcal{K}$ by any bounded open set $ U_0\subset\subset\pi_{\iota-1,\iota}(\Omega_0^M)$ and $\delta$ by $\varepsilon\rho$ in (\ref{initialest}), finding
\be\label{U0}\sup_{ U_0}|\mathsf{F}(\boldsymbol{\xi};x,\rho,\varepsilon)|=O(\rho^{\Delta^+(\theta_1)}\varepsilon^{-q_0}(1-\rho)^{\Delta^+(h)}),\quad -q_0:=-p_0+\Delta^+(\theta_1)+\Delta^+(h).\ee
Next, we choose open sets $U_0,$ $U_1,\ldots, U_m$, with $m:=\lceil q_0+\lambda_0\rceil$ (\ref{lambda}), such that $\mathcal{K}\subset\subset U_m\subset\subset U_{m-1}\subset\subset\ldots\subset\subset U_0\subset\subset\pi_{\iota-1,\iota}(\Omega_0^M)$.  From (\ref{simplestGreenID}), we have
\begin{multline}\label{PstarGnext}\sup_{ U_n}|\mathsf{F}(\boldsymbol{\xi};x,\rho,\varepsilon)|\leq\int_\varepsilon^b\int_0^1\sup_{ U_n}|\eta\sigma(1-\sigma)\mathscr{N}[\mathsf{F}](\boldsymbol{\xi};x,\sigma,\eta)|\frac{|G(\rho,\varepsilon;\sigma,\eta)|}{\eta\sigma(1-\sigma)}\,{\rm d}\eta\,{\rm d}\sigma\\
+\int_0^1\sup_{ U_n}|\mathsf{F}(\boldsymbol{\xi};x,\sigma,b)|\frac{|G(\rho,\varepsilon;\sigma,b)|}{b\sigma(1-\sigma)}\,{\rm d}\sigma\end{multline}
for all $n\in\{0,1,\ldots,m\}$.  Next, we estimate the second definite integral on the right side of (\ref{PstarGnext}).  It follows from (\ref{U0}) with $U_n\subset U_0$ replacing $U_0$ and (\ref{Greenfuncalt}) that for some constant $C_0$ and all $n\in\{0,1,\ldots,m\}$, 
\be\label{2nddefintest}\int_0^1\sup_{ U_n}|\mathsf{F}(\boldsymbol{\xi};x,\sigma,b)|\frac{|G(\rho,\varepsilon;\sigma,b)|}{b\sigma(1-\sigma)}\,{\rm d}\sigma\leq C_0\rho^{\Delta^+(\theta_1)}\varepsilon^{\lambda_0}(1-\rho)^{\Delta^+(h)}\int_0^1K^{(\alpha,\beta)}(\rho,\sigma,-\log(\varepsilon/b)^{\kappa/4})\,w^{(\alpha,\beta)}(\sigma)\,{\rm d}\sigma, \ee
where we define $w^{(\alpha,\beta)}(\sigma)$ in (\ref{wtfnct}).  Because $K^{(\alpha,\beta)}(\rho,\sigma,t)$ is bounded above by a positive constant for all $\rho,\sigma\in[0,1]$ and $t>0$ (\ref{bounds}), the definite integral on the right side of (\ref{2nddefintest}) remains bounded as $\varepsilon\downarrow0$, and (\ref{PstarGnext}) becomes
\be\label{PstarGnext2}\sup_{ U_n}|\mathsf{F}(\boldsymbol{\xi};x,\rho,\varepsilon)|\leq\int_\varepsilon^b\int_0^1\sup_{ U_n}|\eta\sigma(1-\sigma)\mathscr{N}[\mathsf{F}](\boldsymbol{\xi};x,\sigma,\eta)|\frac{|G(\rho,\varepsilon;\sigma,\eta)|}{\eta\sigma(1-\sigma)}\,{\rm d}\eta\,{\rm d}\sigma+O(\rho^{\Delta^+(\theta_1)}\varepsilon^{\lambda_0}(1-\rho)^{\Delta^+(h)}).\ee
Next, we estimate the definite integral on the right side of (\ref{PstarGnext2}) with $n=1$.  After inserting the explicit formula (\ref{Mformula}, \ref{Nformula}) for $\mathscr{N}$ and using (\ref{Greenfuncalt}) and estimates (\ref{newSchauder}) with $ U=U_1$ and $ V=U_0$ and (\ref{U0}), we find that for some $C_1>0$,
\begin{multline}\label{1stdefintest1}\sup_{ U_1}|\mathsf{F}(\boldsymbol{\xi};x,\rho,\varepsilon)|\leq O(\rho^{\Delta^+(\theta_1)}\varepsilon^{\lambda_0}(1-\rho)^{\Delta^+(h)})\\
+C_1\rho^{\Delta^+(\theta_1)}(1-\rho)^{\Delta^+(h)}\int_\varepsilon^b\int_0^1\eta^{-q_0}\left(\frac{\varepsilon}{\eta}\right)^{\lambda_0}K^{(\alpha,\beta)}(\rho,\sigma,-\log(\varepsilon/\eta)^{\kappa/4})\,w^{(\alpha,\beta)}(\sigma)\,{\rm d}\sigma\,{\rm d}\eta.\end{multline}

Next, we determine the behavior of the definite integral in (\ref{1stdefintest1}) as $\varepsilon\downarrow0$ or $\rho\downarrow0$ or $\rho\uparrow1$.  We introduce the substitution $\eta=\varepsilon e^{4t/\kappa}$, chose a $T\in(0,\log(b/\varepsilon)^{\kappa/4})$, and insert the estimates (\ref{bounds}) into this definite integral to find
\begin{multline}\label{1stdefintest2}\int_\varepsilon^b\int_0^1\eta^{-q_0}\left(\frac{\varepsilon}{\eta}\right)^{\lambda_0}K^{(\alpha,\beta)}(\rho,\sigma,(-\log(\varepsilon/\eta)^{\kappa/4})\,w^{(\alpha,\beta)}(\sigma)\,{\rm d}\eta\,{\rm d}\sigma\\
\begin{aligned}&\leq\frac{4}{\kappa}\varepsilon^{-q_0+1}\int_0^Te^{(-q_0-\lambda_0+1)4t/\kappa}\int_0^1K^{(\alpha,\beta)}(\rho,\sigma,t)\,w^{(\alpha,\beta)}(\sigma)\,{\rm d}\sigma\,{\rm d}t\\
&+\frac{4}{\kappa}C\varepsilon^{-q_0+1}\int_T^{\log(b/\varepsilon)^{\kappa/4}}e^{(-q_0-\lambda_0+1)4t/\kappa}\int_0^1w^{(\alpha,\beta)}(\sigma)\,{\rm d}\sigma\,{\rm d}t.\end{aligned}\end{multline}
To determine the behavior of the right side of (\ref{1stdefintest2}) as $\rho\downarrow0$ or $\rho\uparrow1$, we note that the second definite integral on this right side does not depend on $\rho$, but the first definite integral does.  Now within the latter, the definite integral 
\be\label{sigmaint}\int_0^1K^{(\alpha,\beta)}(\rho,\sigma,t)\,w^{(\alpha,\beta)}(\sigma)\,{\rm d}\sigma\ee
is a continuous function of $(\rho,t)\in[0,1]\times[a,T]$ for any $a>0$.  Furthermore, (\ref{sigmaint}) converges to one uniformly in $\rho\in[0,1]$ as $t\downarrow0$ thanks to (\ref{idlim}), so it extends to a continuous function of $(\rho,t)\in[0,1]\times[0,T]$.  Therefore, the right side of (\ref{1stdefintest2}) equals some function of $\rho$, bounded as $\rho\downarrow0$ or $\rho\uparrow1$, multiplied by $\varepsilon^{-q_0+1}$.

Next, we determine the behavior of the right side of (\ref{1stdefintest2}) as $\varepsilon\downarrow0$.  If we choose $p_0$ big enough in (\ref{initialest}) so $m>2$, then the power $-q_0-\lambda_0+1$ is negative.  It therefore follows that the first term on the right side of (\ref{1stdefintest2}) dominates, so that the entire right side is $O(\varepsilon^{-q_0+1})$ as $\varepsilon\downarrow0$ uniformly in $\rho\in[0,1]$.  Furthermore, after inserting this estimate  into (\ref{1stdefintest1}), we find
\be\label{U1}\sup_{ U_1}|\mathsf{F}(\boldsymbol{\xi};x,\rho,\varepsilon)|=O(\rho^{\Delta^+(\theta_1)}\varepsilon^{-q_0+1}(1-\rho)^{\Delta^+(h)}).\ee
We note that the power of $\varepsilon$ in (\ref{U1}) has increased by one from its original value in (\ref{U0}).  Now, we insert this estimate back into (\ref{PstarGnext2}) with $n=2$ and $-q_0$ replaced by $-q_0+1$, and we repeat this analysis another $m-2$ times to find the estimate (\ref{U0}) with $ U_0$ replaced by $U_{m-1}$ and $-q_0$ replaced by $-q_0+m-1$.  Finally, we insert this estimate back into (\ref{PstarGnext2}) with $n=m$, and after using (\ref{1stdefintest2}) with $-q_0$ replaced by $-q_0+m-1$, we find
\begin{multline}\label{1stdefintestlast}\sup_{ U_m}|\mathsf{F}(\boldsymbol{\xi};x,\rho,\varepsilon)|\leq O(\rho^{\Delta^+(\theta_1)}\varepsilon^{\lambda_0}(1-\rho)^{\Delta^+(h)})+O(\rho^{\Delta^+(\theta_1)}\varepsilon^{-q_0+m}(1-\rho)^{\Delta^+(h)})\\
+\frac{4}{\kappa}C_mC\rho^{\Delta^+(\theta_1)}\varepsilon^{-q_0+m}(1-\rho)^{\Delta^+(h)}\int_T^{-\log(\varepsilon/b)^{\kappa/4}}e^{(-q_0+m-\lambda_0)4t/\kappa}\,{\rm d}t\int_0^1w^{(\alpha,\beta)}(\sigma)\,{\rm d}\sigma,\end{multline}
where $C_m$ is some constant.  Now with $-q_0+m-\lambda_0>0$ for the first time in this sequence of $m$ estimations, the integrand, which is $O(\varepsilon^{q_0-m+\lambda_0})$, grows without bound as $t\rightarrow\infty$.  Hence, both of the first and third terms on the right side of (\ref{1stdefintestlast}) dominate as $\varepsilon\downarrow0$ because $-q_0+m > \lambda_0$, so with $\mathcal{K}\subset\subset U_m$, we have
\be\label{K}\sup_{\mathcal{K}}|\mathsf{F}(\boldsymbol{\xi};x,\rho,\varepsilon)|=O(\rho^{\Delta^+(\theta_1)}\varepsilon^{\lambda_0}(1-\rho)^{\Delta^+(h)}).\ee
Thus, we have $q_0\geq-\lambda_0$ in (\ref{U0}).  Now, from $p_0=q_0+\Delta^+(\theta_1)+\Delta^+(h)$ (\ref{U0}) and the expression for the zeroth eigenvalue (\ref{lambda}) with the identity $\Delta^+(\theta_1)=2/\kappa$ (\ref{kpzeq}),
\be\lambda_0=\Delta^+(h)+\Delta^+(\theta_1)+\frac{\kappa}{2}\Delta^+(h)\Delta^+(\theta_1)=2\Delta^+(h)+\Delta^+(\theta_1),\ee
we find that $p_0\geq-\Delta^+(h)$.  Finally, after inserting $p_0=-\Delta^+(h)$ in (\ref{initialest}), we find (\ref{closeest}) as desired.
\end{proof}

\section{Proof of lemma \ref{alltwoleglem}}\label{theproof}

In this section, we present the proof of lemma \ref{alltwoleglem}, the purpose of this article.  This proof incorporates the results of lemmas \ref{kpzlem}--\ref{farlem} and \ref{closelem}, and it closely follows the reasoning outlined in section \ref{Methodology} and at the beginning of section \ref{twointervals}.

\alltwoleglem*
\begin{proof} We assume that $F_1:=F$ is not zero and prove the lemma by contradiction.  Now, $F_1$ has two basic properties that we use in this proof:
\begin{enumerate}[I.]
\item\label{F1item1} As an element of $\mathcal{S}_N\setminus\{0\}$, $F_1$ satisfies conditions \ref{cond1} and \ref{cond2} stated in lemma \ref{firstlimitlem} with $M=\iota=2N$ and $h=\theta_1$.  
\item\label{F1item2} The following limits are guaranteed to exist by lemma \ref{firstlimitlem} and also vanish by hypothesis, i.e.,
\be\label{thelimits}\lim_{x_j\rightarrow x_{j-1}}(x_j-x_{j-1})^{-\Delta^-(\theta_1)}F_1(\boldsymbol{x})=0,\quad  j\in\{3,4,\ldots,2N\},\quad\boldsymbol{x}\in\Omega_0^{2N}.\ee
\end{enumerate}
Next, we gather some facts concerning the alternative limit
\be\label{F2lim} (F_2\circ\pi_{2N})(\boldsymbol{x})\,\,\,:=\lim_{x_{2N}\rightarrow x_{2N-1}}(x_{2N}-x_{2N-1})^{-\Delta^+(\theta_1)}F_1(\boldsymbol{x}),\quad\boldsymbol{x}\in\Omega_0^{2N}.\ee
\begin{enumerate}
\item According to lemma \ref{secondlimitlem}, the limit (\ref{F2lim}) exists and is not zero.
\item According to lemma \ref{pdelem}, the limit (\ref{F2lim}) satisfies conditions \ref{cond1} and \ref{cond2} stated in lemma \ref{firstlimitlem} with $M=\iota=2N-1$ and $h=h^+=\Delta^+(\theta_1)+\theta_1+\theta_1=\theta_2$ (lemma \ref{kpzlem}).  
\item Therefore, according to lemma \ref{firstlimitlem}, the limits
\be\label{prethelimits}\lim_{x_j\rightarrow x_{j-1}}(x_j-x_{j-1})^{-\Delta^-(d)}(F_2\circ\pi_{2N})(\boldsymbol{x}),\quad d=\begin{cases}\theta_1, & j\in\{3,4,\ldots,2N-2\} \\ \theta_2, & j=2N-1\end{cases},\quad\boldsymbol{x}\in\Omega_0^{2N}\ee
exist.  Now we argue that these limits equal zero.
\item First, we argue that the limits (\ref{prethelimits}) with $j\in\{3,4,\ldots,2N-2\}$ equal zero.  Now, thanks to properties \ref{F1item1} and \ref{F1item2}, $F_1$ satisfies conditions \ref{nextcond1}--\ref{nextcond3} of lemma \ref{secondlimitlem}, with $M=\iota=2N$, $h=\theta_1$, and $i\in\{\iota,j\}$, for all $j\leq2N-2$.  Therefore, we may use the estimate (\ref{farest}) of lemma \ref{farlem} to find that whenever $x_j-x_{j-1}$ and $x_{2N}-x_{2N-1}$ are less than some sufficiently small $b>0$, 
\be (x_j-x_{j-1})^{-\Delta^-(\theta_1)}(x_{2N}-x_{2N-1})^{-\Delta^+(\theta_1)}|F_1(\boldsymbol{x})|\leq C(x_j-x_{j-1})^{\Delta^+(\theta_1)-\Delta^-(\theta_1)}\ee
for some constant $C$.  Then, after sending $x_{2N}\rightarrow x_{2N-1}$, we find
\be (x_j-x_{j-1})^{-\Delta^-(\theta_1)}|(F_2\circ\pi_{2N})(\boldsymbol{x})|\leq C(x_j-x_{j-1})^{\Delta^+(\theta_1)-\Delta^-(\theta_1)}.\ee
Because $\Delta^+(\theta_1)-\Delta^-(\theta_1) > 0$, after sending $x_j\rightarrow x_{j-1}$, we find that the limits (\ref{prethelimits}) with $j\in\{3,4,\ldots,2N-2\}$ equal zero.
\item Last, we argue that the limit (\ref{prethelimits}) is zero for $j=2N-1$.  Again, thanks to properties \ref{F1item1} and \ref{F1item2}, $F_1$ satisfies conditions \ref{nextcond1}--\ref{nextcond3} of lemma \ref{secondlimitlem}, with $M=\iota=2N$, $h=\theta_1$, and $i\in\{\iota-1,\iota\}$.  Therefore, we may use the estimate (\ref{closeest}) of lemma \ref{closelem} to find that whenever $x_{2N}-x_{2N-2}$ is less than some sufficiently small $b>0$,
\begin{multline} (x_{2N-1}-x_{2N-2})^{-\Delta^-(\theta_2)}(x_{2N}-x_{2N-1})^{-\Delta^+(\theta_1)}|F_1(\boldsymbol{x})|\\
\leq C(x_{2N-1}-x_{2N-2})^{\Delta^+(\theta_1)-\Delta^-(\theta_2)}(x_{2N}-x_{2N-2})^{\Delta^+(\theta_1)}\end{multline}
for some constant $C$.  Then, after sending $x_{2N}\rightarrow x_{2N-1}$, we find
\be(x_{2N-1}-x_{2N-2})^{-\Delta^-(\theta_2)}|(F_2\circ\pi_{2N})(\boldsymbol{x})|\leq C(x_{2N-1}-x_{2N-2})^{2\Delta^+(\theta_1)-\Delta^-(\theta_2)}.\ee
Because $2\Delta^+(\theta_1)-\Delta^-(\theta_2)>0$, after sending $x_{2N-1}\rightarrow x_{2N-2}$, we find that limit (\ref{prethelimits}) with $j=2N-1$ equals zero.
\end{enumerate}
We summarize what we have established about the limit $F_2:\Omega_0^{2N-1}\rightarrow\mathbb{R}$ (\ref{F2lim}):
\begin{enumerate}[I.]
\item\label{F2item1} $F_2$ is nonzero and satisfies conditions \ref{cond1} and \ref{cond2} stated in lemma \ref{firstlimitlem} with $M=\iota=2N-1$ and $h=\theta_2$.  Stated in terms of CFT, we interpret these properties to mean that $F_2$ is a correlation function of $2N-2$ one-leg boundary operators at $x_1,$ $x_2,\ldots,x_{2N-2}$ and a single two-leg boundary operator at $x_{2N-1}$.
\item\label{F2item2} The following limits vanish:
\be\label{newtwoleg}\lim_{x_j\rightarrow x_{j-1}}(x_j-x_{j-1})^{-\Delta^-(d)}F_2(\boldsymbol{x})=0,\quad d=\begin{cases}\theta_1, & j\in\{3,4,\ldots,2N-2\} \\ \theta_2, & j=2N-1\end{cases},\quad\boldsymbol{x}\in\Omega_0^{2N-1}.\ee
\end{enumerate}

Properties \ref{F2item1} and \ref{F2item2} of $F_2$ are respectively analogous with properties \ref{F1item1} and \ref{F1item2} of $F_1$.  Therefore, we naturally expect that we may pass from $F_2$ to some function $F_3$ with properties analogous to those of $F_2$, then from $F_3$ to some function $F_4$ with properties analogous to those of $F_3$, and so on, until we pass from $F_{2N-2}$ to $F_{2N-1}$, as section \ref{Methodology} describes.  Each passage from, say, $F_s$ to $F_{s+1}$ simplifies the picture by removing another point among $x_1,$ $x_2,\ldots,x_{2N}$ from the system.  If we reach it, then the ultimate function $F_{2N-1}$ depends on only the two coordinates $x_1$ and $x_2$, and understanding this function may be considerably easier than understanding the original function $F_1$.

Now we construct the collection of mentioned functions $S=\{F_3,F_4,\ldots,F_{2N-1}\}$.  In particular, we show only the $(s+1)$th step of the construction in which we build $F_{s+2}$ from $F_{s+1}$, so by proceeding sequentially down the list of indices $s\in\{1,2,\ldots,2N-1\}$ starting with $s=1$, we construct all of the functions in $S$.  Now, $F_s:\Omega_0^{2N-s+1}\rightarrow\mathbb{R}$ has the following two properties: 
\begin{enumerate}[I.]
\item\label{Fsitem1} $F_s$ is nonzero and satisfies conditions \ref{cond1} and \ref{cond2} stated in lemma \ref{firstlimitlem} with $M=\iota=2N-s+1$ and $h=\theta_s$.  Stated in terms of CFT, we interpret these properties to mean that $F_s$ is a correlation function of $2N-s$ one-leg boundary operators at $x_1,$ $x_2,\ldots,x_{2N-s}$ and a single $s$-leg boundary operator at $x_{2N-s+1}$.
\item\label{Fsitem2} The following limits vanish:
\be\label{snewtwoleg}\lim_{x_j\rightarrow x_{j-1}}(x_j-x_{j-1})^{-\Delta^-(d)}F_s(\boldsymbol{x})=0,\quad d=\begin{cases}\theta_1, & j\in\{3,4,\ldots,2N-s\} \\ \theta_s, & j=2N-s+1\end{cases},\quad\boldsymbol{x}\in\Omega_0^{2N-s+1}.\ee
\end{enumerate}
Next, we gather some facts concerning the alternative limit
\be\label{sF2lim} (F_{s+1}\circ\pi_{2N-s+1})(\boldsymbol{x})\,\,\,:=\lim_{\substack{x_{2N-s+1}\\ \rightarrow x_{2N-s}}}(x_{2N-s+1}-x_{2N-s})^{-\Delta^+(\theta_s)}F_s(\boldsymbol{x}),\quad\boldsymbol{x}\in\Omega_0^{2N-s+1}.\ee
\begin{enumerate}
\item According to lemma \ref{secondlimitlem}, the limit (\ref{sF2lim}) exists and is not zero.
\item According to lemma \ref{pdelem}, the limit (\ref{sF2lim}) satisfies conditions \ref{cond1} and \ref{cond2} stated in lemma \ref{firstlimitlem} with $M=\iota=2N-s$ and $h=h^+=\Delta^+(\theta_s)+\theta_1+\theta_s=\theta_{s+1}$ (lemma \ref{kpzlem}).  
\item Therefore, according to lemma \ref{firstlimitlem}, the limits
\be\label{sthelimits}\lim_{x_j\rightarrow x_{j-1}}(x_j-x_{j-1})^{-\Delta^-(d)}(F_{s+1}\circ\pi_{2N-s+1})(\boldsymbol{x}),\quad d=\begin{cases}\theta_1, & j\in\{3,4,\ldots,2N-s-1\} \\ \theta_{s+1}, & j=2N-s\end{cases},\quad\boldsymbol{x}\in\Omega_0^{2N-s+1}.\ee
exist.  Now we argue that these limits equal zero.
\item First, we argue that the limits (\ref{sthelimits}) with $j\in\{3,4,\ldots,2N-s-1\}$ equal zero.  Now, thanks to properties \ref{Fsitem1} and \ref{Fsitem2} of $F_s$, $F_s$ satisfies conditions \ref{nextcond1}--\ref{nextcond3} of lemma \ref{secondlimitlem}, with $M=\iota=2N-s+1$, $h=\theta_s$, and $i\in\{\iota,j\}$, for all $j\leq2N-s-1$.  Therefore, we may use the estimate (\ref{farest}) of lemma \ref{farlem} to find that whenever $x_j-x_{j-1}$ and $x_{2N-s+1}-x_{2N-s}$ are less than some sufficiently small $b>0$, 
\be (x_j-x_{j-1})^{-\Delta^-(\theta_1)}(x_{2N-s+1}-x_{2N-s})^{-\Delta^+(\theta_s)}|F_s(\boldsymbol{x})|\leq C(x_j-x_{j-1})^{\Delta^+(\theta_1)-\Delta^-(\theta_1)}\ee
for some constant $C$.  Then, after sending $x_{2N-s+1}\rightarrow x_{2N-s}$, we find
\be (x_j-x_{j-1})^{-\Delta^-(\theta_1)}|(F_{s+1}\circ\pi_{2N-s+1})(\boldsymbol{x})|\leq C(x_j-x_{j-1})^{\Delta^+(\theta_1)-\Delta^-(\theta_1)}.\ee
Because $\Delta^+(\theta_1)-\Delta^-(\theta_1) > 0$, after sending $x_j\rightarrow x_{j-1}$, we find that the limits (\ref{sthelimits}) with $j\in\{3,4,\ldots,2N-s-1\}$ equal zero.
\item Finally, we argue that the limit (\ref{sthelimits}) is zero for $j=2N-s$.  Again, thanks to properties \ref{Fsitem1} and \ref{Fsitem2}, $F_s$ satisfies conditions \ref{nextcond1}--\ref{nextcond3} of lemma \ref{secondlimitlem}, with $M=\iota=2N-s+1$, $h=\theta_s$, and $i\in\{\iota-1,\iota\}$.  Therefore, we may use the estimate (\ref{closeest}) of lemma \ref{closelem} to find that whenever $x_{2N-s+1}-x_{2N-s-1}$ is less than some sufficiently small $b>0$,
\begin{multline} (x_{2N-s}-x_{2N-s-1})^{-\Delta^-(\theta_{s+1})}(x_{2N-s+1}-x_{2N-s})^{-\Delta^+(\theta_s)}|F_s(\boldsymbol{x})|\\
\leq C(x_{2N-s}-x_{2N-s-1})^{\Delta^+(\theta_1)-\Delta^-(\theta_{s+1})}(x_{2N-s+1}-x_{2N-s-1})^{\Delta^+(\theta_s)}\end{multline}
for some constant $C$.  Then, after sending $x_{2N-s+1}\rightarrow x_{2N-s}$, we find
\be(x_{2N-s}-x_{2N-s-1})^{-\Delta^-(\theta_{s+1})}|(F_{s+1}\circ\pi_{2N-s+1})(\boldsymbol{x})|\leq C(x_{2N-s}-x_{2N-s-1})^{\Delta^+(\theta_1)+\Delta^+(\theta_s)-\Delta^-(\theta_{s+1})}.\ee
Finally, because $\Delta^+(\theta_1)+\Delta^+(\theta_s)-\Delta^-(\theta_{s+1})>0$, after sending $x_{2N-s}\rightarrow x_{2N-s-1}$, we find that the limit (\ref{sthelimits}) with $j=2N-s$ equals zero.
\end{enumerate} 
We may summarize these properties of the limit $F_{s+1}:\Omega_0^{2N-s}\rightarrow\mathbb{R}$ (\ref{sF2lim}) by replacing $s$ with $s+1$ in the properties \ref{Fsitem1} and \ref{Fsitem2} for $F_s$ above.  Thus, we may construct $F_{s+2}$ from $F_{s+1}$, and so on.

Now we consider $s=2N-1$.  Here, property \ref{Fsitem1} implies that $F_{2N-1}$, a function of $(x_1,x_2)\in\Omega_0^2$, is both not zero and solves the system of conformal Ward identities (\ref{wardidh}) with $M=\iota=2$ and $h=\theta_{2N-1}\neq\theta_1$.  However, it is easy to show that zero is the only solution to this system.  Thus, the premise that $F_1:=F$ is not zero leads to a contradiction.
\end{proof}

\section{Summary}

In this article and \cite{florkleb,florkleb3,florkleb4}, we study a solution space $\mathcal{S}_N$ for a system of $2N$ null-state PDEs (\ref{nullstate}) and three conformal Ward identities (\ref{wardid}) governing conformal field theory (CFT) correlation functions of $2N$ one-leg boundary operators ($\phi_{1,2}$ or $\phi_{2,1}$).  In \cite{florkleb}, we prove that $\dim\mathcal{S}_N\leq C_N$, with $C_N$ the $N$th Catalan number, assuming lemma \red{14} of \cite{florkleb} (restated as lemma \ref{alltwoleglem} here) is true.  This lemma asserts that if all intervals $(x_2,x_3),$ $(x_3,x_4),\ldots,(x_{2N-2},x_{2N-1})$, and $(x_{2N-1},x_{2N})$ are ``two-leg intervals" of a element $F(\boldsymbol{x})$ of $\mathcal{S}_N$ ($x_i$ is the $i$th coordinate of $\boldsymbol{x}$, and $x_i<x_j$ if $i<j$), then $F(\boldsymbol{x})$ is zero.  (In terms of CFT, if $(x_{i-1},x_i)$ is a two-leg interval of $F$, then only the two-leg boundary operator ($\phi_{1,3}$ or $\phi_{3,1}$) appears in the operator product expansion (OPE) of the pair of one-leg boundary operators at $x_{i-1}$ and $x_i$.)  

The purpose of this article is to prove lemma  \ref{alltwoleglem}, and in non-rigorous CFT terms, the proof goes as follows.  Supposing that a solution $F$ is not zero, if all of the above intervals are two-leg intervals of $F$, then successively shrinking neighboring intervals produces successively higher multi-leg boundary operators, until we have one interval, a one-leg boundary operator at its left endpoint, and a $(2N-1)$-leg boundary operator at its right endpoint.  Conformal covariance implies that the two-point correlation function of these two primary operators necessarily vanishes.  On the other hand, we rigorously prove that if $F$ is not zero, then this two-point function cannot vanish, a contradiction.

The proof of lemma \ref{alltwoleglem} also justifies some facts about $M$-point CFT correlation functions $F$ containing $M-1$ one-leg boundary operators and one  primary boundary operator of conformal weight $h$.  Lemma \ref{firstlimitlem} states that an appropriate limit (\ref{hlim-}) of $F(\boldsymbol{x})$ as $x_i\rightarrow x_{i-1}$ is approached uniformly in the locations of the other coordinates of $\boldsymbol{x}$ (bounded away from each other).  In CFT, this limit corresponds to the leading term of the conformal family with smaller weight in the OPE of the two boundary primary operators at $x_{i-1}$ and $x_i$ respectively.  One of these operators is a one-leg boundary operator.  Lemma \ref{secondlimitlem} and \ref{pdelem} show that if this limit is zero, then a different limit (\ref{hlim+}) of $F(\boldsymbol{x})$ is approached with the same uniformness, and it corresponds to the leading term of the conformal family with larger conformal weight in the mentioned OPE and satisfies null-state PDEs and conformal Ward identities.  Finally, lemmas \ref{farlem}, \ref{firstcloselem}, and \ref{closelem} give uniform estimates on the behavior of $F$ as the lengths of two of its intervals go to zero simultaneously.

In this article, we use a certain technique for obtaining estimates of $F$ as the lengths of one or two of its intervals vanish.  In lemmas \ref{firstlimitlem} and \ref{secondlimitlem}, only one interval's length vanishes.  Here, we use a Green function with two variables to express the null-state PDE centered on one of the endpoints as an integral equation.  This equation contains derivatives of $F$ with respect to variables not involved with the Green function.  We use the other null-state PDEs and conformal Ward identities to construct an elliptic PDE for $F$, from which Schauder interior estimates that bound these derivatives by $F$ itself follow. Then repeated integration of the integral equation with these estimates improves bounds on the growth of $F$ until we reach an optimal bound.  In lemma \ref{farlem},  two non-adjacent intervals' lengths simultaneously vanish.  Here, we execute the method of the case with one vanishing interval length twice.  Finally, in lemmas \ref{firstcloselem} and \ref{closelem}, two adjacent intervals' lengths simultaneously vanish.  Here, we use a Green function with four variables to express the null-state PDE centered on one of the endpoints as an integral equation and bound the growth of $F$ as these two interval lengths simultaneously vanish.  The Green function of this last case factors into power functions multiplying the Jacobi heat kernel \cite{nowsj1}.  We use recent estimates of this kernel \cite{nowsj2} to complete our proof of lemma \ref{closelem}.  

\section{Acknowledgements}

We thank J.\ Rauch and V.\ Elling for helpful conversations concerning lemma \ref{alltwoleglem} and C.\ Townley Flores for proofreading the manuscript.  This work was supported by National Science Foundation Grant No.\ PHY-0855335 (SMF).

\end{document}